\documentclass[11pt,letterpaper]{article}

\usepackage[margin=1in,includeheadfoot]{geometry}
\usepackage{graphicx} 
\usepackage{float} 
\usepackage{pdflscape}
\usepackage{amsthm}
\usepackage{thm-restate}
\usepackage{mathtools}
\usepackage{amssymb} 
\usepackage[dvipsnames,svgnames]{xcolor}
\usepackage[hidelinks]{hyperref}
\hypersetup{
	unicode=true,          
	colorlinks=true,        
	linkcolor=DarkRed,          
	citecolor=DarkGreen,        
	filecolor=magenta,      
	urlcolor=cyan           
}
\usepackage[capitalize, nameinlink,noabbrev]{cleveref}
\usepackage{algorithm} 
\usepackage{algpseudocodex} 

\crefname{algorithm}{Algorithm}{Algorithms}
\usepackage{textcomp} 
\usepackage{xspace}
\usepackage[textwidth=20mm,obeyFinal]{todonotes}
\setuptodonotes{size=\tiny}

\usepackage[backend=biber,style=alphabetic,sorting=nty]{biblatex}         
\let\citet\textcite
\usepackage{enumitem}
\usepackage{tikz}
\usepackage{calc}
\usepackage{subcaption}
\usetikzlibrary{decorations.pathreplacing,calc,arrows.meta,fit,positioning,shapes.geometric,backgrounds,patterns}
\usepackage[normalem]{ulem}
\usepackage[leftbars,xcolor]{changebar}

\usepackage{thmtools}
\usepackage{booktabs}

\theoremstyle{plain}
\newtheorem{theorem}{Theorem}
\newtheorem{lemma}[theorem]{Lemma}
\newtheorem{corollary}[theorem]{Corollary}
\theoremstyle{definition}
\newtheorem{definition}[theorem]{Definition}

\theoremstyle{remark}

\AddToHook{env/lemma/begin}{\crefalias{theorem}{lemma}}
\AddToHook{env/corollary/begin}{\crefalias{theorem}{corollary}}
\AddToHook{env/definition/begin}{\crefalias{theorem}{definition}}
\AddToHook{env/invariant/begin}{\crefalias{theorem}{invariant}}
\AddToHook{env/claim/begin}{\crefalias{theorem}{claim}}
\AddToHook{env/observation/begin}{\crefalias{theorem}{observation}}
\AddToHook{env/remark/begin}{\crefalias{theorem}{remark}}

\newcommand{\rank}{\mathrm{rank}}

\newcommand{\gap}{\mathrm{gap}}

\newcommand{\algo}[1]{\textsf{#1}}

\usepackage{fontawesome}

\DeclareMathOperator{\lazygap}{intervalgap}
\DeclareMathOperator{\nonlazygap}{pointgap}

\newcommand{\OO}{O}

\newcommand{\contractedgapbound}{G}

\newcommand{\fI}{\mathcal {I}}
\newcommand{\fU}{\mathcal {U}}

\DeclareMathOperator{\OPT}{OPT}

\DeclareMathOperator{\depth}{depth}
\DeclareMathOperator{\cost}{cost}
\DeclareMathOperator{\Path}{path}
\let\top=\undefined
\DeclareMathOperator{\top}{top}
\DeclareMathOperator{\bottom}{bottom}
\DeclareMathOperator{\heapchildren}{HeapChildren}
\DeclareMathOperator{\heapparent}{HeapParent}
\DeclareMathOperator{\leftheapchildren}{LeftHeapChildren}
\DeclareMathOperator{\rightheapchildren}{RightHeapChildren}
\DeclareMathOperator{\intervals}{intervals}

\def\defop#1{\expandafter\def\csname#1\endcsname{\textsc{#1}}}
\defop{Convert}
\defop{Create}
\defop{Insert}
\defop{Delete}
\defop{Split}
\defop{Pass}
\defop{ValueChange}
\defop{ValueDecrease}
\defop{ValueIncrease}

\makeatletter
\AtBeginDocument{
  \hypersetup{
    pdftitle = {\@title},
    pdfauthor = {\ourauthors},
	pdfsubject = {\ourabstract},
	pdfkeywords = {\ourkeywords},
  }
}
\makeatother
\newif\ifsubmission

\submissionfalse

\title{Splay trees are almost dynamically optimal}
\date{}

\ifsubmission
\author{Anonymous Authors}
\def\affiliations{}
\def\ourauthors{Anonymous Authors}
\else
\author{%
  Petr Chmel${}^1$\\\texttt{chmel@iuuk.mff.cuni.cz}
  \and
  Bernhard Haeupler${}^{2,3}$\\\texttt{bernhard.haeupler@insait.ai}
  \and
  Richard Hladík${}^2$\\\texttt{rihl@rihl.cz}
  \and
  Michal Koucký${}^1$\\\texttt{koucky@iuuk.mff.cuni.cz}
  \and
  Antti Roeyskoe${}^2$\\\texttt{aroeyskoe@inf.ethz.ch}
  \and
  Václav Rozhoň${}^1$\\\texttt{vaclavrozhon@gmail.com}
  \and
  Ondřej Sladký${}^2$\\\texttt{ondra.sladky@gmail.com}
  \and
  Robert E. Tarjan${}^4$\\\texttt{ret@cs.princeton.edu}
}
\def\affiliations{%
\footnotetext[1]{Charles University}
\footnotetext[2]{ETH Zurich}
\footnotetext[3]{INSAIT, Sofia University ``St.~Kliment Ohridski''}
\footnotetext[4]{Princeton University}
\setcounter{footnote}{4}
}
\def\ourauthors{Petr Chmel, Bernhard Haeupler, Richard Hladík, Michal Koucký, Antti Roeyskoe, Václav Rozhoň, Ondřej Sladký, Robert E. Tarjan}
\fi
\def\ourkeywords{splay trees, dynamic optimality}
\def\ourabstract{We prove that splay trees are dynamically optimal up to a factor of log log n times the square of log log log n.}

\begin{document}
\maketitle
\affiliations
\begin{abstract}
Sleator and Tarjan~\cite{SleatorTarjan85} conjectured that splay trees are dynamically optimal -- that on every access sequence, they perform within a constant factor of the optimal offline dynamic binary search tree. Despite four decades of work, no $o(\log n)$ competitive ratio was known.

\smallskip
\noindent We prove that splay trees are
$O(\log\log n \cdot \log^2\log\log n)=\tilde{O}(\log\log n)$-competitive.
\end{abstract}

\section{Introduction}

Binary search trees are among the most fundamental data structures for maintaining an ordered set of $n$ elements. A binary search tree consists of nodes arranged in a tree; each node stores one key, with all smaller keys in its left subtree and all larger keys in its right subtree.
By only comparing elements and navigating among the nodes, binary search trees are able to support a wide range of operations on the ordered set such as membership queries, inserts, deletes, or predecessor queries. There are techniques which ensure the tree is balanced -- having depth of $\Theta(\log n)$, which implies that the operations work in $\Theta(\log n)$ worst-case time~\cite{adelson1962algorithm,Guibas1978}.
In the worst case, this is the best possible -- any comparison-based procedure requires $\Omega(\log n)$ time to perform an access in the worst case.

While this is optimal in the worst case, there are access sequences which are easy to resolve faster than in logarithmic time per access. For example,
if the accessed elements are drawn from a fixed distribution, one can perform the accesses in time proportional to the entropy of the distribution using Huffman trees~\cite{huffman1952method}. If the accessed elements have high locality -- for instance, if they are accessed sequentially -- constant time per access suffices. In these and many more cases, worst-case optimality fails to capture that on many instances of interest one can do much better than the $\Theta(\log n)$-time access guaranteed by balanced binary search trees.

Dynamic optimality is a much stronger guarantee than worst-case optimality. A data structure is dynamically optimal if, on any sufficiently long access sequence, its performance is asymptotically the same as the optimal tree tailored to the particular access sequence even with the advantage of knowing the whole sequence in advance.
In other words, a dynamically optimal binary search tree is constant-competitive with the optimal search tree for any sufficiently long sequence.
In some sense, dynamic optimality is the strongest guarantee possible in terms of time complexity. Up to constant factors, one cannot perform better on any sequence of accesses than a dynamically optimal binary search tree.

While no binary search tree algorithm is known to be dynamically optimal~\cite{Iacono13InPursuit,Russo19Study,Kozma16Thesis}, there are two main candidates -- splay trees~\cite{SleatorTarjan85} and GREEDY~\cite{Lucas88Canonical,Munro00LinearSearch,DemaineHIKP09Geometry}.
In fact, the whole concept of dynamic optimality originated with splay trees.
Prior to this work, no $o(\log n)$ competitive ratio was known for either of them.
Their $O(\log n)$ competitive ratios follow from the fact that accesses in both take amortized time $O(\log n)$~\cite{SleatorTarjan85,Fox11MaximallyGreedy}, as is also true for any balanced binary search tree.
Several binary search tree algorithms are known to achieve $O(\log\log n)$-competitiveness; the first was the Tango tree~\cite{DemaineHIP07Tango}, followed by a line of work using similar proof ideas~\cite{WangDS06MultiSplay,Georgakopoulos08ChainSplay,BoseDouiebDujmovicFagerberg10Zipper}. All of these track the so-called interleave lower bound on the optimal tree~\cite{Wilber89}, but since this lower bound can be separated from the optimal tree by a factor of $\Omega(\log\log n)$, this line of analyses cannot prove dynamic optimality. For Tango trees specifically, the $\Theta(\log\log n)$ analysis is tight.

Splay trees, which were introduced in 1985 by Sleator and Tarjan~\cite{SleatorTarjan85}, are binary search trees that, unlike other trees, are not balanced by a set of static rules, but adjust themselves after every operation based on which element was accessed.
In particular, after an access, the accessed element is moved to the root by a series of rotations and double-rotations, which are called zig-zigs and zig-zags.
Intuitively, this is a sensible strategy: first, frequently accessed elements need to be close to the root; second, the strategy of putting an element at the front is constant-competitive for linked lists~\cite{SleatorTarjan85ListUpdate}.
Much of the evidence for why splay trees are conjectured to be dynamically optimal is that they are known to perform well on a range of restricted families of access sequences.
First, Tarjan and Sleator showed that splay trees are up to a constant factor as good as any other static binary search tree, i.e., a tree that is not allowed to modify its structure~\cite{SleatorTarjan85}.
Next, Tarjan~\cite{Tarjan85SequentialAccess} showed that on sequential accesses, splay trees take amortized time $\Theta(1)$ per access.
Furthermore, splay trees are known to perform well on a family of pre-order and post-order traversal sequences~\cite{LevyTarjan19PrePost} and if the updates happen on the extremes as in a deque, each operation takes nearly-constant time~\cite{Pettie08Deque}.
For other related problems, splay trees are conjectured to behave asymptotically optimally. One example is the split conjecture, which states that any sequence of splits in a splay tree takes total linear time.
However, prior to this work, no competitive ratio beyond $O(\log n)$ was known for splay trees.

Here, we give the first $o(\log n)$ competitive ratio for splay trees, showing that splay trees are
$\tilde{O}(\log\log n)$-competitive,
which, ignoring $\log\log\log n$ factors, matches the best currently known competitive binary search trees.
While there are already binary search trees known to achieve this competitive ratio, this is the first time such a ratio is given for a binary search tree that is conjectured to be dynamically optimal.

Unlike existing proofs of sublogarithmic competitiveness for other binary search trees~\cite{DemaineHIP07Tango,WangDS06MultiSplay}, we do not charge against a lower bound on the work of the optimal tree, which is known to be non-tight~\cite{LecomteWeinstein20Wilber}.
Instead, we charge directly to the optimal tree, circumventing the non-trivial gap between this lower bound and the optimum.

The first step of our proof extends the original proof of static optimality: besides charging work to a static tree, we build structures that allow us to analyze splay trees even when the optimal tree performs rotations.
In fact, the core of our argument -- defining the rank of a node by its depth in the optimal tree and defining heavy paths in the splay tree to lead towards vertices closest to the root of the optimal tree -- can be seen as an alternative to the ranks used in the original static-optimality analysis and can also be used to rederive static optimality.
The proof then has three more core ingredients.
First, we introduce lazy intervals, which allow us to absorb changes in the ranks after a tree rotates.
Second, we observe that the structure of heavy paths after a splay changes similarly to a pairing heap after a delete-min.
This lets us use recent progress in the analysis
of multi-pass and pure pairing heaps~\cite{Sinnamon2025,PairingHeaps2026} to help bound the total number of zig-zigs performed.
Third, we show that we can assume a normal form on the optimal tree, which restricts it to certain properties while slowing it down only by a constant factor.
With this, we show that the number of zig-zags can be asymptotically bounded by the number of zig-zigs by observing how the structure of heavy paths changes.

Finally, we apply the same techniques to prove the split conjecture for splay trees up to a $\tilde{O}(\log\log n)$ factor.

\section{Setting the stage}
\label{sec:preliminaries}

In this section, we set up the basics needed to state our result: we define the model in which the time complexity of binary search trees is measured and introduce splay trees.

\subsection{Problem setup}

\paragraph{BST model.}

Formally, this is the BST model: we have to dynamically maintain a binary search tree (BST) $T$. This is an $n$-node binary tree where every node has a unique key from $\fU = \{1, \ldots, n\}$; this is without loss of generality because elements are only compared, so only their order matters. $T$ must at all times satisfy the \emph{ordering property}: for any node $u$, its key is larger than all keys in the left subtree of $u$ and smaller than all keys in the right subtree of $u$.
The tree undergoes a sequence of accesses $x_1, \dots, x_m$ (with $x_i \in \fU$). Upon each access, we can repeatedly perform the following operation on the nodes of the tree:
\begin{itemize}[itemsep=0pt]
\item $\algo{rotate}(x)$. Change the tree locally around $x$. If $x$ is the left child of $y$, $y$ now becomes the right child of $x$, the formerly right child of $x$ becomes the new left child of $y$, and $x$ becomes a child of $y$'s parent. The case where $x$ is the right child is analogous.
\end{itemize}
After each access, the node being accessed must be in the root. The cost of an access is $1$ plus the number of rotations performed. The total cost of a BST data structure $T$ on the access sequence $X = x_1, \ldots, x_m$ is denoted by $\cost(T, X)$.

Note that this differs from the more standard BST model by requiring that the accessed element end up at the root. We show the equivalence of these models in \Cref{app:bst-models}.

\paragraph{Relating to the optimal tree.}
For a fixed access sequence $X = x_1, \dots, x_m$, the optimal tree $T^\text{OPT}_X$ is a dynamic BST that, among all possible BSTs, incurs the smallest cost on the given access sequence while knowing the whole sequence in advance. We will often write just $T^\text{OPT}$ when the access sequence is clear from the context.
By $T$, we denote an online BST algorithm that we compare to $T^\text{OPT}$.
In our case, $T$ always denotes the splay tree.
For the analysis, we serialize the two executions so that, for access $x_i$, the optimal tree first performs its access and all rotations needed to bring $x_i$ to the root. Then, while $T^\text{OPT}$ is fixed, the splay tree performs its splay of $x_i$. Thus, during the splay-tree part of the access, $x_i$ is already the root of $T^\text{OPT}$.
Furthermore, the initial configuration for $T^\text{OPT}$ is assumed to be the optimal one, whereas for the splay tree we assume that we have no control over the initial configuration.

\subsection{Splay trees}

\begin{figure}[t]
\centering

\begin{subfigure}[t]{0.48\textwidth}
\centering
\begin{tikzpicture}[
  scale=0.68,
  transform shape,
  level distance=6mm,
  level 1/.style={sibling distance=16mm},
  level 2/.style={sibling distance=16mm},
  level 3/.style={sibling distance=16mm},
  level 4/.style={sibling distance=16mm},
  level 5/.style={sibling distance=16mm},
  level 6/.style={sibling distance=16mm},
  level 7/.style={sibling distance=16mm},
  internal/.style={circle, draw, minimum size=5.5mm, inner sep=1pt},
  >=latex
]

\node[internal] {}
  child {node[internal] {}
    child[missing]
    child {node[internal] {}
      child {node[internal] {}
        child[missing]
        child {node[internal] {}
          child {node[internal] {}
            child[missing]
            child {node[internal] {}
              child {node[internal] {$x$}}
              child[missing]
            }
          }
          child[missing]
        }
      }
      child[missing]
    }
  }
  child[missing];

\end{tikzpicture}
\caption{Initial tree state.}
\label{fig:zig-zag-rotate:initial-path}
\end{subfigure}
\hfill
\begin{subfigure}[t]{0.48\textwidth}
\centering
\begin{tikzpicture}[
  scale=0.68,
  transform shape,
  level distance=8mm,
  level 1/.style={sibling distance=35mm},
  level 2/.style={sibling distance=8mm},
  level 3/.style={sibling distance=8mm},
  level 4/.style={sibling distance=8mm},
  level 5/.style={sibling distance=8mm},
  level 6/.style={sibling distance=8mm},
  level 7/.style={sibling distance=8mm},
  internal/.style={circle, draw, minimum size=5.5mm, inner sep=1pt},
  >=latex
]

\node[internal] {$x$}
  child {node[internal] {}
    child[missing]
    child {node[internal] {}
      child[missing]
      child {node[internal] {}
        child[missing]
        child {node[internal] {}}
      }
    }
  }
  child {node[internal] {}
    child {node[internal] {}
      child {node[internal] {}
        child {node[internal] {}}
        child[missing]
      }
      child[missing]
    }
    child[missing]
  };

\end{tikzpicture}
\caption{Rotating $x$ up by single rotations, which is the same as zig-zagging in splay.}
\label{fig:zig-zag-rotate:rotate-x-up}
\end{subfigure}

\caption{Result of repeatedly rotating $x$ on a path with alternating left and right edges.}
\label{fig:zig-zag-rotate}
\end{figure}
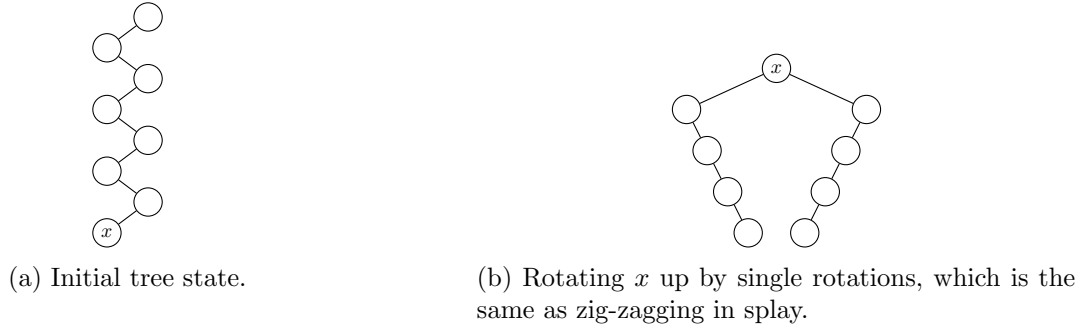

\begin{figure}[t]
\centering

\begin{subfigure}[t]{0.32\textwidth}
\centering
\begin{tikzpicture}[
  scale=0.68,
  transform shape,
  level distance=6mm,
  level 1/.style={sibling distance=10mm},
  level 2/.style={sibling distance=10mm},
  level 3/.style={sibling distance=10mm},
  level 4/.style={sibling distance=10mm},
  level 5/.style={sibling distance=10mm},
  level 6/.style={sibling distance=10mm},
  internal/.style={circle, draw, minimum size=5.5mm, inner sep=1pt},
  >=latex
]

\node[internal] {}
  child {node[internal] {}
    child {node[internal] {}
    child {node[internal] {}
    child {node[internal] {}
      child {node[internal] {}
        child {node[internal] {}
          child {node[internal] {}
            child {node[internal] {$x$}}
            child[missing]
          }
          child[missing]
        }
        child[missing]
        }
        child[missing]
        }
        child[missing]
      }
      child[missing]
    }
    child[missing]
  }
  child[missing];

\end{tikzpicture}
\caption{Initial tree state.}
\label{fig:rotate-vs-splay:initial-path}
\end{subfigure}
\hfill
\begin{subfigure}[t]{0.32\textwidth}
\centering
\begin{tikzpicture}[
  scale=0.68,
  transform shape,
  level distance=6mm,
  level 1/.style={sibling distance=10mm},
  level 2/.style={sibling distance=10mm},
  level 3/.style={sibling distance=10mm},
  level 4/.style={sibling distance=10mm},
  level 5/.style={sibling distance=10mm},
  level 6/.style={sibling distance=10mm},
  internal/.style={circle, draw, minimum size=5.5mm, inner sep=1pt},
  >=latex
]

\node[internal] {$x$}
  child[missing]
  child {node[internal] {}
    child {node[internal] {}
    child {node[internal] {}
    child {node[internal] {}
      child {node[internal] {}
        child {node[internal] {}
          child {node[internal] {}
            child {node[internal] {}}
            child[missing]
          }
          child[missing]
        }
        child[missing]
        }
        child[missing]
        }
        child[missing]
      }
      child[missing]
    }
    child[missing]
  };

\end{tikzpicture}
\caption{Rotating $x$ up by single rotations.}
\label{fig:rotate-vs-splay:rotate-x-up}
\end{subfigure}
\hfill
\begin{subfigure}[t]{0.32\textwidth}
\centering
\begin{tikzpicture}[
  scale=0.68,
  transform shape,
  level distance=6mm,
  level 1/.style={sibling distance=14mm},
  level 2/.style={sibling distance=14mm},
  level 3/.style={sibling distance=14mm},
  level 4/.style={sibling distance=14mm},
  internal/.style={circle, draw, minimum size=5.5mm, inner sep=1pt},
  >=latex
]

\node[internal] {$x$}
  child[missing]
  child {node[internal] {}
    child {node[internal] {}
      child {node[internal] {}
      child {node[internal] {}
        child[missing]
        child {node[internal] {}}
      }
        child {node[internal] {}}
        }
      child {node[internal] {}}
    }
    child {node[internal] {}}
  };

\end{tikzpicture}
\caption{Splaying $x$ by repeated zig-zigs.}
\label{fig:rotate-vs-splay:splay-x}
\end{subfigure}

\caption{Comparison of repeatedly rotating $x$ upward using single rotations and splaying $x$ on a path with left edges only.}
\label{fig:rotate-vs-splay}
\end{figure}
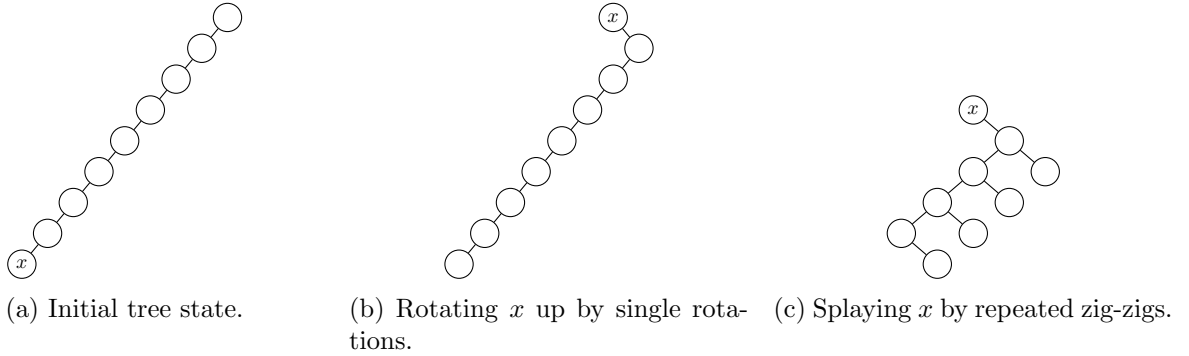

\begin{figure}[t]
\centering

\begin{subfigure}[t]{0.49\textwidth}
\centering
\begin{tikzpicture}[
  scale=0.62,
  transform shape,
  level distance=10mm,
  level 1/.style={sibling distance=16mm},
  level 2/.style={sibling distance=14mm},
  level 3/.style={sibling distance=12mm},
  subtree/.style={
    draw,
    isosceles triangle,
    shape border rotate=90,
    minimum height=8mm,
    minimum width=5mm,
    yshift=-7.5mm
  },
  subtree child/.style={
    edge from parent path={
      (\tikzparentnode) -- (\tikzchildnode.north)
    }
  },
  internal/.style={circle, draw},
  >=latex
]

\begin{scope}[xshift=-3.2cm]
\node[internal] {$z$}
  child {node[internal] {$y$}
    child {node[internal] {$x$}
      child[subtree child] {node[subtree] {$A$}}
      child[subtree child] {node[subtree] {$B$}}
    }
    child[subtree child] {node[subtree] {$C$}}
  }
  child[subtree child] {node[subtree] {$D$}};
\end{scope}

\draw[->, thick] (-0.65,-1.05) -- (0.65,-1.05);

\begin{scope}[xshift=3.2cm]
\node[internal] {$x$}
  child[subtree child] {node[subtree] {$A$}}
  child {node[internal] {$y$}
    child[subtree child] {node[subtree] {$B$}}
    child {node[internal] {$z$}
      child[subtree child] {node[subtree] {$C$}}
      child[subtree child] {node[subtree] {$D$}}
    }
  };
\end{scope}

\end{tikzpicture}
\caption{Zig-zig}
\label{fig:zig-zig}
\end{subfigure}
\hfill
\begin{subfigure}[t]{0.49\textwidth}
\centering
\begin{tikzpicture}[
  scale=0.62,
  transform shape,
  level distance=10mm,
  level 1/.style={sibling distance=28mm},
  level 2/.style={sibling distance=14mm},
  level 3/.style={sibling distance=14mm},
  subtree/.style={
    draw,
    isosceles triangle,
    shape border rotate=90,
    minimum height=8mm,
    minimum width=6mm,
    yshift=-7.5mm
  },
  subtree child/.style={
    edge from parent path={
      (\tikzparentnode) -- (\tikzchildnode.north)
    }
  },
  internal/.style={circle, draw},
  >=latex
]

\begin{scope}[xshift=-3.2cm]
\node[internal] {$z$}
  child {node[internal] {$y$}
    child[subtree child] {node[subtree] {$A$}}
    child {node[internal] {$x$}
      child[subtree child] {node[subtree] {$B$}}
      child[subtree child] {node[subtree] {$C$}}
    }
  }
  child[subtree child] {node[subtree] {$D$}};
\end{scope}

\draw[->, thick] (-0.65,-1.05) -- (0.65,-1.05);

\begin{scope}[xshift=3.2cm]
\node[internal] {$x$}
  child {node[internal] {$y$}
    child[subtree child] {node[subtree] {$A$}}
    child[subtree child] {node[subtree] {$B$}}
  }
  child {node[internal] {$z$}
    child[subtree child] {node[subtree] {$C$}}
    child[subtree child] {node[subtree] {$D$}}
  };
\end{scope}

\end{tikzpicture}
\caption{Zig-zag}
\label{fig:zig-zag}
\end{subfigure}

\caption{The zig-zig and zig-zag splay steps.}
\label{fig:zig-zig-zig-zag}
\end{figure}
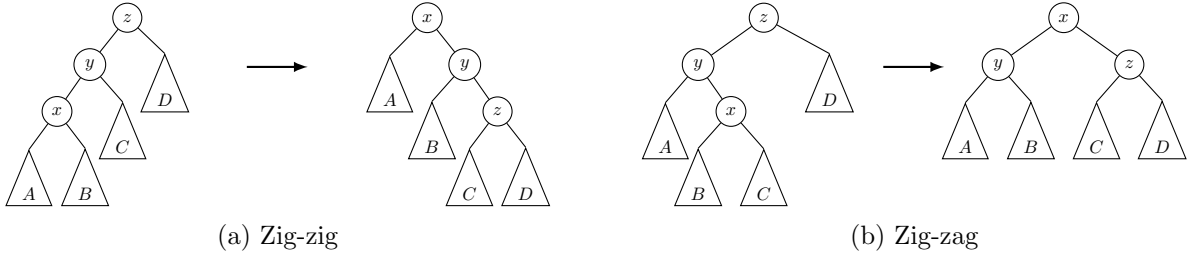

Splay trees~\cite{SleatorTarjan85} are binary search trees that adjust themselves with the aim of performing well on any sequence.
The first main idea behind the design is that once a node is accessed, it is moved to the root by a series of rotations. 
The intuition for this is twofold. First, the same move-to-front heuristic turns out to be dynamically optimal in the setting where the data structures are linked lists instead of binary search trees~\cite{SleatorTarjan85ListUpdate}. Second, if an element is accessed often, it is better for it to be at or near the root, as otherwise the competing tree pays a lot while the optimal tree pays only a little.

The simplest strategy that moves the accessed element to the root is to rotate the accessed element by single rotations all the way to the root~\cite{AllenMunro78SelfOrganizing}.
If the path to the accessed element alternates between left and right edges, the structural changes to the tree look reasonable: the path is split into two of roughly equal length -- see \Cref{fig:zig-zag-rotate}.
However, if the path to the accessed element encounters a long left or a long right path, these paths remain essentially unchanged and thus can be traversed unnecessarily many times, rendering this strategy inefficient. For example, on a sequential access starting from a long path this algorithm has total cost $\Theta(n^2)$ and so average cost per access of $\Theta(n)$, much worse than constant cost per access which is possible for such a sequence.
Splay trees do the next simplest thing; the splay procedure is such that whenever a long path of only left or only right edges is encountered, it is at least partially rebalanced.
The difference between these two approaches is illustrated in \Cref{fig:rotate-vs-splay}.

More precisely, splaying a node $x$ to the root works through the following sequence of double rotations and rotations, depending on the local structure around $x$.

\begin{itemize}
    \item If $x$ and its parent are either both left children or both right children, then a \emph{zig-zig} is performed by rotating first the parent of $x$ and then $x$. Unlike two rotations of $x$, this ensures rebalancing of the path.
    \item If $x$ is a right child and its parent is a left child, or vice versa, a \emph{zig-zag} is performed by rotating  $x$ twice.
    \item Right below the root, $x$ might not have a grandparent. In this case a \emph{zig} is performed by only rotating $x$.
\end{itemize}

Zig-zig and zig-zag operations are illustrated in \Cref{fig:zig-zig-zig-zag}.

Since each step in the splay algorithm moves $x$ one level closer to the root while performing at most 2 rotations, the cost of accessing the element $x$ is $O(\depth_T(x))$.

\section{Warm-up: static optimality}

In this section, we sketch a re-proof of static optimality for sufficiently long sequences using the elementary combinatorial elements used in the proof of dynamic optimality.
This section helps the reader understand the basics of splay trees and become acquainted with our notation in a familiar setting. The combinatorial elements introduced in this section are largely variations on existing concepts.
We further remark that the re-proof of static optimality is not fully rigorous, as the main aim of this section is to familiarize the reader with the concepts rather than to give a full proof.

\paragraph{Ternary view of splay trees for simpler proofs.}
First, to make all proofs substantially simpler and avoid multiple corner cases, we impose the following ternary view on splay trees. For every node $v$, we create an additional middle value child that stores the original value of $v$, and we keep its left and right subtrees. In addition, for nodes that are missing a left or right child, we add a null node. We assign dummy values to null nodes in a way consistent with the tree structure.
We imagine that the static optimal tree contains these nodes as well -- in fact, their place is uniquely determined -- and we define the ranks of null nodes analogously to other nodes. The newly added nodes are called external, while all other nodes are called internal.
Overall, in the ternary view, all internal nodes are hollow and have exactly three children, and their structure exactly matches that of the original tree. The external nodes are new, and each stores a real or dummy value.
The transition to the ternary view is illustrated in \Cref{fig:overview:ternary-view}.
Rotations happen on internal nodes, and all value nodes always remain attached to their internal parent; see \Cref{fig:overview:rotation}.

\paragraph{Relating the splay tree to the optimal tree through ranks and heavy paths.}
As in the dynamic case, the re-proof of the static case starts by relating the splay tree to an optimal tree $T^\text{OPT}_\text{static}$.
For every node $v$ in the splay tree $T$, we define its rank to be its depth in $T^\text{OPT}_\text{static}$.
Consider an edge $(u, v)$ with $u$ being the parent of $v$. The minimum rank among the nodes in the subtree of $u$ is less than or equal to the minimum rank among the nodes in the subtree of $v$.
If it is equal, we call the edge \emph{heavy}, otherwise it is \emph{light}.
An important property of the ranks is that, since they are defined by depths in a tree, in every subtree of $T$ there is only one node that minimizes the rank. This builds upon a folklore result that every consecutive interval of nodes has a unique node minimizing the depth in the optimal tree, which we prove in \Cref{lemma:unique-depth-minimum}. Thus, for each internal node, exactly one of its three children is connected to it via a heavy edge.
This decomposes the tree into a collection of heavy paths and light edges, with each light edge connecting paths of different ranks.
For an internal node, the rank of the bottom-most node on the same heavy path plays a role analogous to that of the rank in the original access-lemma analysis of splay trees, where the rank was the logarithm of the sum of weights.

\tikzset{heavy/.style={edge from parent/.style={draw, line width=1pt}}}
\newcommand{\heavy}{edge from parent={draw,line width=5pt} }

\begin{figure}[H]
\centering

\begin{subfigure}{\textwidth}
\centering
\begin{tikzpicture}[
  level distance=1.1cm,
  every node/.style={minimum size=5.5mm, inner sep=1pt},
  internal/.style={circle, draw},
  value/.style={rectangle, draw},
  nullnode/.style={inner sep=1pt},
  >=latex
]

\begin{scope}[
  xshift=-3.8cm,
  every node/.style={circle, draw, minimum size=5.5mm},
  level 1/.style={sibling distance=2.2cm},
  level 2/.style={sibling distance=1.7cm},
  level 3/.style={sibling distance=1.3cm}
]

\node {5}
  child {node {1}
    child[missing]
    child {node {3}
      child {node {2}}
      child {node {4}}
    }
  }
  child {node {6}};

\end{scope}

\draw[->, thick] (-1.8,-1.2) -- (-0.5,-1.2);

\begin{scope}[xshift=3.8cm,
  level 1/.style={sibling distance=2.0cm},
  level 2/.style={sibling distance=1.2cm},
  level 3/.style={sibling distance=2.0cm},
  level 4/.style={sibling distance=1.2cm}
]

\node[internal] {}
  child {node[internal] {}
    child {node[nullnode] {\small $\perp_0$}}
    child {node[value] {1}}
    child {node[internal] {}
      child {node[internal] {}
        child {node[nullnode] {\small $\perp_1$}}
        child {node[value] {2}}
        child {node[nullnode] {\small $\perp_2$}}
      }
      child {node[value] {3}}
      child {node[internal] {}
        child {node[nullnode] {\small $\perp_3$}}
        child {node[value] {4}}
        child {node[nullnode] {\small $\perp_4$}}
      }
    }
  }
  child {node[value] {5}}
  child {node[internal] {}
    child {node[nullnode] {\small $\perp_5$}}
    child {node[value] {6}}
    child {node[nullnode] {\small $\perp_6$}}
  };

\end{scope}

\end{tikzpicture}
\caption{Transition to the ternary view with null nodes. }
\label{fig:overview:ternary-view}
\end{subfigure}

\vspace{0.8em}

\begin{subfigure}{\textwidth}
\centering
\begin{tikzpicture}[
  level distance=9mm,
  sibling distance=12mm,
  subtree/.style={
    draw,
    isosceles triangle,
    shape border rotate=90,
    minimum height=8mm,
    minimum width=6mm,
    yshift=-7.5mm
  },
  subtree child/.style={
    edge from parent path={
      (\tikzparentnode) -- (\tikzchildnode.north)
    }
  },
  internal/.style={circle, draw},
  value/.style={rectangle, draw},
  >=latex
]

\begin{scope}[xshift=-3cm]

\node[internal] {$y$}
  child {node[internal] {$x$}
    child[subtree child] {node[subtree] {$A$}}
    child {node[value] {$u$}}
    child[subtree child] {node[subtree] {$B$}}
  }
  child {node[value] {$v$}}
  child[subtree child] {node[subtree] {$C$}};

\end{scope}

\draw[->, thick] (-0.5,0) -- (0.8,0);

\begin{scope}[xshift=3cm]

\node[internal] {$x$}
  child[subtree child] {node[subtree] {$A$}}
  child {node[value] {$u$}}
  child {node[internal] {$y$}
    child[subtree child] {node[subtree] {$B$}}
    child {node[value] {$v$}}
    child[subtree child] {node[subtree] {$C$}}
  };

\end{scope}

\end{tikzpicture}
\caption{An example of a left rotation in the ternary view of BSTs.}\label{fig:overview:rotation}
\end{subfigure}

\vspace{0.8em}

\begin{subfigure}{\textwidth}
\centering
\begin{tikzpicture}[
  level distance=1.1cm,
  every node/.style={minimum size=5.5mm, inner sep=1pt},
  internal/.style={circle, draw},
  value/.style={rectangle, draw},
  nullnode/.style={inner sep=1pt},
  >=latex
]
\begin{scope}[
  xshift=-3.8cm,
  every node/.style={circle, draw, minimum size=5.5mm},
  nullnode/.style={draw=none, rectangle, inner sep=0pt},
  level 1/.style={sibling distance=2.8cm},
  level 2/.style={sibling distance=1.4cm},
  level 3/.style={sibling distance=0.8cm}
]
\node {4}
  child {node {2}
    child {node {1}
      child {node[nullnode] {\small $\perp_0$}}
      child {node[nullnode] {\small $\perp_1$}}
    }
    child {node {3}
      child {node[nullnode] {\small $\perp_2$}}
      child {node[nullnode] {\small $\perp_3$}}
    }
  }
  child {node {6}
    child {node {5}
      child {node[nullnode] {\small $\perp_4$}}
      child {node[nullnode] {\small $\perp_5$}}
    }
    child {node[nullnode] {\small $\perp_6$}}
  };
\end{scope}
\node at (-3.8cm,-4.2cm) {$T^{\text{OPT}}$};
\begin{scope}[xshift=3.8cm,
  edge from parent/.style={draw,dashed,line width=1pt},
  level 1/.style={sibling distance=2.2cm,solid,line width=.4pt},
  level 2/.style={sibling distance=1.5cm,solid,line width=.4pt},
  level 3/.style={sibling distance=2.0cm,solid,line width=.4pt},
  level 4/.style={sibling distance=1.3cm,solid,line width=.4pt}
]
\node[internal] (root) {}
  child {node[internal] (a) {}
    child {node[nullnode] {\small $\perp_0$} node[below=0mm] {$\scriptstyle r=3$}}
    child {node[value] {1} node[below=2mm] {$\scriptstyle r=2$}}
    child {node[internal] (b) {}
      child {node[internal] (c) {}
        child {node[nullnode] {\small $\perp_1$} node[below=0mm] {$\scriptstyle r=3$}}
        child {node[value] (n2) {2} node[below=2mm] {$\scriptstyle r=1$}}
        child {node[nullnode] {\small $\perp_2$} node[below=0mm] {$\scriptstyle r=3$}}
      }
      child {node[value] {3} node[below=2mm] {$\scriptstyle r=2$}}
      child {node[internal] (d) {}
        child {node[nullnode] {\small $\perp_3$} node[below=0mm] {$\scriptstyle r=3$}}
        child {node[value] (n4) {4} node[below=2mm] {$\scriptstyle r=0$} }
        child {node[nullnode] {\small $\perp_4$} node[below=0mm] {$\scriptstyle r=3$}}
      }
    }
  }
  child {node[value] {5} node[below=2mm] {$\scriptstyle r=2$}}
  child {node[internal] (e) {}
    child {node[nullnode] {\small $\perp_5$} node[below=0mm] {$\scriptstyle r=3$}}
    child {node[value] (n6) {6} node[below=2mm] {$\scriptstyle r=1$}}
    child {node[nullnode] {\small $\perp_6$} node[below=0mm] {$\scriptstyle r=2$}}
  };
  \draw[ultra thick]
  (root) -- (a) -- (b) -- (d) -- (n4);
  \draw[ultra thick]
  (e) -- (n6);
  \draw[ultra thick]
  (c) -- (n2);
\end{scope}
\end{tikzpicture}
\caption{The ranks of external nodes are based on their depth in the example optimal tree $T^\text{OPT}$. Based on these ranks, the tree is decomposed into heavy paths (thick) which are connected by light edges (dashed). The heavy path corresponding to value 4 has rank 0 so it contains the root. The parent of external node 1 is the only bend on this heavy path.}\label{fig:overview:heavy-light}
\end{subfigure}

\caption{An overview of the concepts for splay-tree analysis.}
\label{fig:overview}
\end{figure}
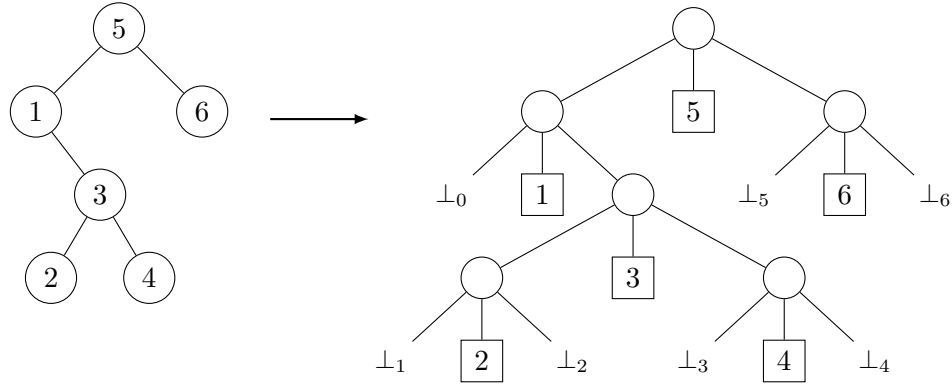
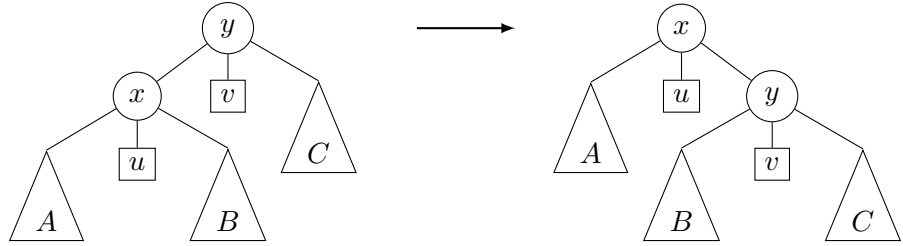
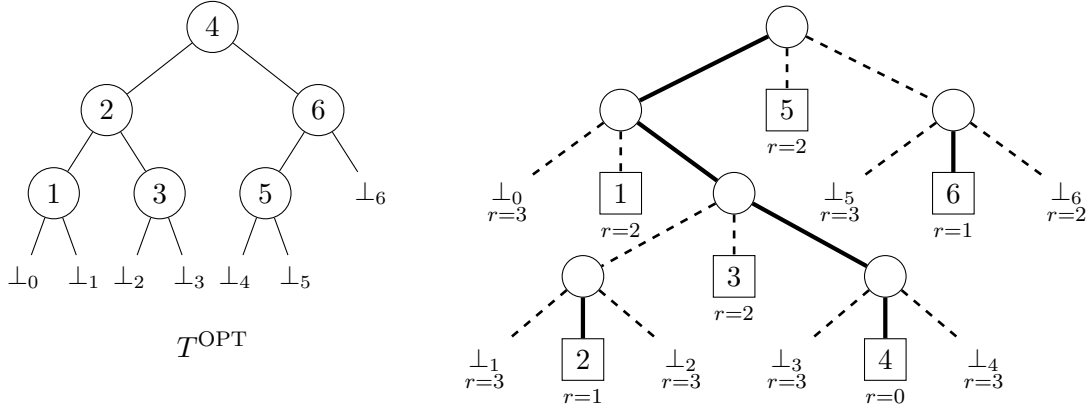

The ternary view, in particular, makes the structure of heavy paths nice. Every heavy path now ends in an external node and there is exactly one such external node on each heavy path. We can therefore associate each heavy path with this external node, which is its bottom-most node, and say the rank of the heavy path is the rank of this bottom-most node.
In addition, if we imagine that the external nodes connected by light edges have degenerate heavy paths consisting of a single node, then every external node also has its own heavy path. See also \Cref{fig:overview:heavy-light}.

\paragraph{Potential via gaps.}
We now have almost everything needed to prove static optimality. 
As in the original proof of static optimality, where rank differences along edges were central, we define the \emph{gap} of a heavy path $P$ as the difference between the rank of $P$ and the rank of the heavy path to which it is connected by a light edge. We also associate the gap of a heavy path with the light edge to its top-most node. The gap of the root heavy path, which contains the root, is $0$. The sum of the gaps of all heavy paths will serve as a potential: we show that each operation over heavy edges decreases the potential, and we bound the increases caused by operations over light edges, thus concluding that the splay tree cannot do too much work.

\paragraph{Operations over light edges.}
When a node $x$ is accessed, the splay tree rotates it to the root via a series of zig-zigs and zig-zags and possibly one zig.
The first case to consider is when an operation happens over a light edge.
By this we mean that, if $x^\circ$ denotes the internal parent of the accessed external node and $y^\circ$ denotes the parent of $x^\circ$, then $x^\circ$ or $y^\circ$ is connected to its parent by a light edge.

To bound this case, first observe that there cannot be too many light edges on the path from $x$ to the root. In fact, there are at most $\rank(x)$ such edges, as every light edge decreases the rank by at least one.
Second, if the operation happens over a light edge with gap $t$,
the sum of gaps increases by at most $O(t)$. To prove this fact, one would have to verify it in many different situations based on how exactly the light and heavy edges are positioned.
On the path from $x$ to the root, the gaps on the light edges sum to exactly the rank of $x$. Thus, the operations over light edges for one access
increase the sum of gaps by at most the rank of $x$.\footnote{This case never occurs in the dynamic setting. There, we can assume that the optimal tree always accesses the root, which is the only external node with rank $0$, so there are only heavy edges on the path from it to the root.
In the static case, we cannot make this assumption without making the optimal tree dynamic.}

\paragraph{Zig-zigs on heavy edges decrease gaps.}
Next, we show that every zig-zig over heavy edges decreases the sum of gaps by at least one. It suffices to consider a right zig-zig, as the situation with a left zig-zig is analogous. Define $x^\circ$ and $y^\circ$ as above, and let $z^\circ$ be the parent of $y^\circ$.
Let $z$ be the middle child of $z^\circ$ and $y^+$ the least-rank node in the right subtree of $y^\circ$; define $z^+$ analogously.
During a zig-zig, $z^\circ$ is shifted downward and $y^+$, $z$, $z^+$ are the new least-rank vertices in the respective child subtrees of $z^\circ$.
The one with the least rank -- let us analyze the case where this is $y^+$ -- acquires $z^\circ$. For the other two nodes, $z$ and $z^+$, the heavy path relative to which their gaps are measured changes from the previous one -- consider the case where it was $x$ -- to $y^+$. Before the operation, $y^+$ was connected by a light edge to $y^\circ$ because we assumed that the operation was over heavy edges, so its gap was measured with respect to $x$. Therefore, it had higher rank, which implies that the gaps of $z$ and $z^+$ decrease by at least one, since the rank of the reference node in the definition of the gap increases.
See \cref{fig:gaps-after-zig-zig} for an illustration of this case.
It can be checked that this holds in the other cases as well and that, for the other nodes, the ranks do not increase.


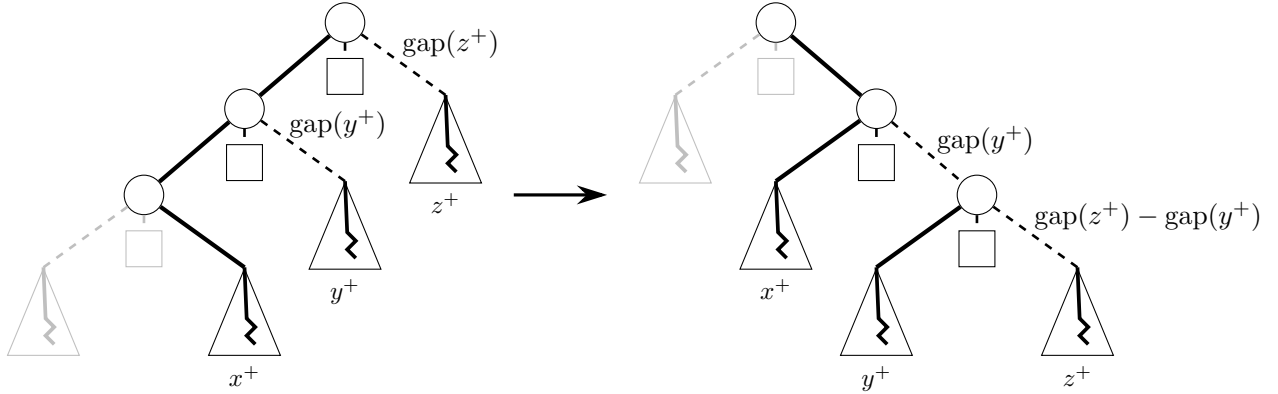
\begin{figure}[t]
\centering

\begin{tikzpicture}[
    scale=0.95,
    transform shape,
    vertex/.style={
        circle,
        draw,
        fill=white,
        minimum size=5.5mm,
        inner sep=0pt
    },
    childsquare/.style={
        rectangle,
        draw,
        fill=white,
        minimum size=5mm,
        inner sep=0pt,
        font=\small
    },
    trig/.style={
        isosceles triangle,
        draw,
        fill=white,
        minimum height=9mm,
        minimum width=10mm,
        inner sep=0pt,
        shape border rotate=90
    },
    subtree/.style={
        trig,
        yshift=-11mm
    },
    dashed edge/.style={
        dashed,
        line width=1pt
    },
    solid edge/.style={
        line width=1.6pt
    },
    grey edge/.style={
        line width=1.6pt,
        lightgray
    },
    grey dashed edge/.style={
        dashed,
        line width=1pt,
        lightgray
    },
    grey node/.style={
        draw=lightgray,
        text=lightgray
    },
    grey subtree/.style={
        subtree,
        draw=lightgray,
        text=lightgray
    },
    transform arrow/.style={
        -{Stealth[length=4mm,width=2.5mm]},
        line width=1.2pt
    },
    every label/.style={
        font=\small
    }
]

\def\xshift{1.4}
\def\yshift{0.8}
\def\mainxshift{\xshift}
\def\mainyshift{\yshift*1.5}


\node[vertex] (Lz) at (-4, \mainyshift) {};
\node[vertex] (Ly) at (-4-\mainxshift, 0) {};
\node[vertex] (Lx) at (-4-2*\mainxshift,-\mainyshift) {};

\draw[solid edge] (Lz) -- (Ly) -- (Lx);

\node[childsquare,grey node] (Lxsq) at ($(Lx)+(0,-0.75)$) {};
\node[childsquare] (Lysq) at ($(Ly)+(0,-0.75)$) {};
\node[childsquare]           (Lzsq) at ($(Lz)+(0,-0.75)$) {};

\node[grey subtree] (Lxm) at ($(Lx)+(-\xshift,-\yshift)$) {};
\node[subtree,label=below:{$x^{+}$}] (Lxp)
    at ($(Lx)+(\xshift,-\yshift)$) {};

\node[subtree,label=below:{$y^{+}$}] (Lyp)
    at ($(Ly)+(\xshift,-\yshift)$) {};

\node[subtree,label=below:{$z^{+}$}] (Lzp)
    at ($(Lz)+(\xshift,-\yshift)$) {};

\draw[grey dashed edge] (Lx) -- (Lxm.apex);
\draw[grey dashed edge] (Lx) -- (Lxsq);
\draw[solid edge]       (Lx) -- (Lxp.apex);

\draw[dashed edge] (Ly) -- (Lysq);
\draw[dashed edge]      (Ly) -- node[above right,pos=0.30,xshift=-1mm,yshift=-1.5mm] {$\operatorname{gap}(y^{+})$} (Lyp.apex);

\draw[dashed edge]      (Lz) -- (Lzsq);
\draw[dashed edge]      (Lz) -- node[above right,pos=0.45,xshift=-1mm,yshift=-1mm] {$\operatorname{gap}(z^{+})$} (Lzp.apex);


\draw[transform arrow] (-1.65,-\mainyshift) -- (-0.35,-\mainyshift);


\node[vertex] (Rx) at (2, \mainyshift) {};
\node[vertex] (Ry) at (2+\mainxshift, 0) {};
\node[vertex] (Rz) at (2+2*\mainxshift,-\mainyshift) {};

\draw[solid edge] (Rx) -- (Ry);
\draw[dashed edge] (Ry) -- node[above right,pos=0.55,xshift=-.5mm,yshift=-1.5mm] {$\operatorname{gap}(y^{+})$} (Rz);

\node[childsquare,grey node] (Rxsq) at ($(Rx)+(0,-0.75)$) {};
\node[childsquare] (Rysq) at ($(Ry)+(0,-0.75)$) {};
\node[childsquare]           (Rzsq) at ($(Rz)+(0,-0.75)$) {};

\node[grey subtree] (Rxm) at ($(Rx)+(-\xshift,-\yshift)$) {};

\node[subtree,label=below:{$x^{+}$}] (Ryp)
    at ($(Ry)+(-\xshift,-\yshift)$) {};

\node[subtree,label=below:{$y^{+}$}] (RzpL)
    at ($(Rz)+(-\xshift,-\yshift)$) {};

\node[subtree,label=below:{$z^{+}$}] (Rzp)
    at ($(Rz)+(\xshift,-\yshift)$) {};

\draw[grey dashed edge]        (Rx) -- (Rxm.apex);
\draw[grey dashed edge] (Rx) -- (Rxsq);

\draw[solid edge]       (Ry) -- (Ryp.apex);
\draw[dashed edge] (Ry) -- (Rysq);

\draw[solid edge]       (Rz) -- (RzpL.apex);
\draw[dashed edge]      (Rz) -- (Rzsq);
\draw[dashed edge]      (Rz) -- node[above right,pos=0.45,xshift=-1mm,yshift=-1mm]
    {$\operatorname{gap}(z^{+})-\operatorname{gap}(y^{+})$} (Rzp.apex);


\foreach \T in {Lxp,Lyp,Lzp,Ryp,RzpL,Rzp}{
    \draw[line width=1.4pt]
        ($(\T.center)+(0.15,-0.18)$)
        -- ++(-0.12,0.12)
        -- ++( 0.12,0.12)
        -- ++(-0.12,0.12)
        -- (\T.apex);
}

\foreach \T in {Lxm,Rxm}{
    \draw[line width=1.4pt, lightgray]
        ($(\T.center)+(0.15,-0.18)$)
        -- ++(-0.12,0.12)
        -- ++( 0.12,0.12)
        -- ++(-0.12,0.12)
        -- (\T.apex);
}

\end{tikzpicture}

\caption{The situation after a zig-zig over two heavy edges in one particular configuration, assuming the rank of $y^+$ is smaller than the rank of $z^+$. $\gap(\cdot)$ denotes the gap before the zig-zig. The gap of $z^+$ decreases as before it was computed based on $x^+$ and after it is computed based on $y^+$ which has higher rank.}
\label{fig:gaps-after-zig-zig}

\end{figure}

\paragraph{Concluding static optimality for sufficiently long sequences.}
Similarly, a zig-zag over two heavy edges can be shown to decrease the potential by one. Therefore, for every splay, the potential can increase by at most asymptotically the rank of the accessed node.
There can be at most $\rank(x)$ zig-zigs or zig-zags over light edges, where $x$ is the splayed vertex. Since every other zig-zig or zig-zag decreases the potential by at least one, the amortized number of rotations is $O(\rank(x))$, possibly plus $1$ for the final zig.

As $\rank(x)$ is the depth of $x$ in the static optimal tree and $1+\rank(x)$ is the access cost, and since the initial potential is at most $O(n^2)$, the total work done by splay trees is $O(n^2 + \cost(T^\text{OPT}_\text{static}))$.
With a slightly improved potential, the dynamic analysis reduces the $O(n^2)$ factor to $\tilde{O}(n\log\log n)$.

\section{Technical overview}\label{sec:technical-overview}

Having concluded a warm-up proof of static optimality, let us focus on dynamic optimality.
In this section, we provide intuition for our proof and an overview of all the important concepts.

Just as the static optimality of splay trees is proved by charging the work done by splay trees to that done by an optimal static tree, our argument charges the work done by splay trees to an optimal tree that changes over time.

As in the warm-up section, the proof builds upon the concepts used to sketch static optimality. We impose the ternary view on the splay tree, define node ranks based on their depths in the optimal tree, define heavy paths, and consider their gaps.
The proof itself will rely on a deeper understanding of these combinatorial concepts.

\paragraph{Heap view of heavy paths.}
Imagine contracting every heavy path into the bottom-most node and keeping all the children of nodes on the heavy path as children of the contracted node.
In this contracted tree, there are only light edges and thus the ranks increase with every edge.
Hence, this contracted tree is heap-ordered on the ranks, meaning that a node has lower rank than all its children, and we call this the \emph{heap view of heavy paths}; see \Cref{fig:overview:heap-children-heap-parents}. To distinguish between the children in the tree itself and in the heap view, we call the children and parents in the heap view the \emph{heap-children} and \emph{heap-parents}, respectively.
Heap-children whose values precede that of the heap-parent in symmetric order are called \emph{left heap-children}; \emph{right heap-children} are defined analogously.
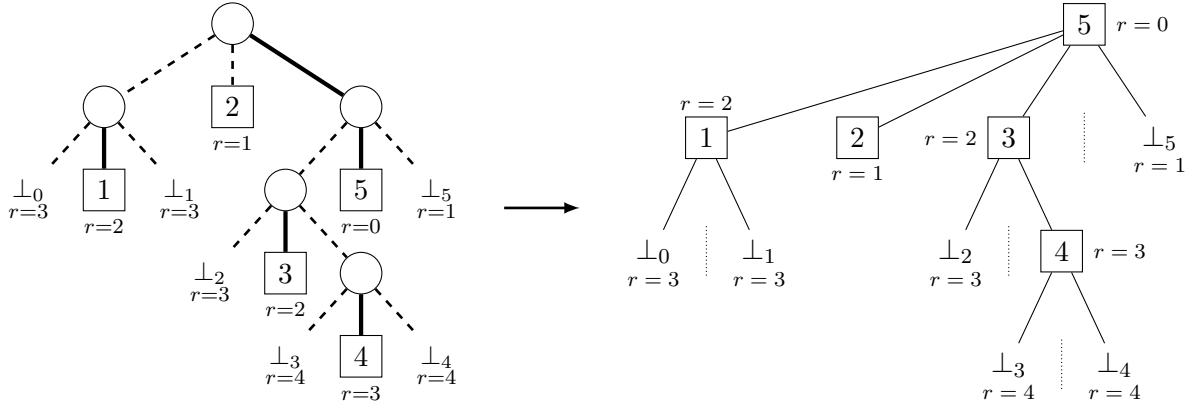
\begin{figure}[H]
\centering

\begin{tabular}{@{}c@{\hspace{1.5em}}c@{\hspace{1.5em}}c@{}}

\begin{tikzpicture}[
  baseline=(current bounding box.center),
  level distance=1.1cm,
  every node/.style={minimum size=5.5mm, inner sep=1pt},
  internal/.style={circle, draw},
  value/.style={rectangle, draw},
  nullnode/.style={inner sep=1pt},
  edge from parent/.style={draw,dashed,line width=1pt},
  level 1/.style={sibling distance=1.7cm,solid, line width=.4pt},
  level 2/.style={sibling distance=1cm,  solid, line width=.4pt},
  level 3/.style={sibling distance=1.0cm,solid, line width=.4pt},
  level 4/.style={sibling distance=1cm,  solid, line width=.4pt},
  >=latex
]

\node[internal] (root) {}
  child {node[internal] (e) {}
    child {node[nullnode] {\small $\perp_0$} node[below=0mm] {$\scriptstyle r=3$}}
    child {node[value] (n6) {1} node[below=2mm] {$\scriptstyle r=2$}}
    child {node[nullnode] {\small $\perp_1$} node[below=0mm] {$\scriptstyle r=3$}}
  }
  child {node[value] {2} node[below=2mm] {$\scriptstyle r=1$}}
  child {node[internal] (a) {}
    child {node[internal] (b) {}
      child {node[nullnode] {\small $\perp_2$} node[below=0mm] {$\scriptstyle r=3$}}
      child {node[value] (c) {3} node[below=2mm] {$\scriptstyle r=2$}}
      child {node[internal] (d) {}
        child {node[nullnode] {\small $\perp_3$} node[below=0mm] {$\scriptstyle r=4$}}
        child {node[value] (n2) {4} node[below=2mm] {$\scriptstyle r=3$}}
        child {node[nullnode] {\small $\perp_4$} node[below=0mm] {$\scriptstyle r=4$}}
      }
    }
    child {node[value] (f) {5} node[below=2mm] {$\scriptstyle r=0$}}
    child {node[nullnode] {\small $\perp_5$} node[below=0mm] {$\scriptstyle r=1$}}
    };

  \draw[ultra thick] (root) -- (a) -- (f);
  \draw[ultra thick] (e) -- (n6);
  \draw[ultra thick] (b) -- (c);
  \draw[ultra thick] (d) -- (n2);

\end{tikzpicture}
&
\begin{tikzpicture}[baseline=(current bounding box.center),>=latex]
  \draw[->, line width=1pt] (-0.5,0) -- (0.5,0);
\end{tikzpicture}
&
\begin{tikzpicture}[
  baseline=(current bounding box.center),
  every node/.style={minimum size=5.5mm, inner sep=1pt},
  value/.style={rectangle, draw},
  perpnode/.style={rectangle, minimum size=5.5mm, inner sep=1pt},
  rval/.style={font=\scriptsize, inner sep=1pt},
  auxline/.style={densely dotted},
  invisible/.style={edge from parent/.style={draw=none}},
  level distance=1.5cm,
  level 1/.style={sibling distance=2cm},
  level 2/.style={sibling distance=1.4cm},
  level 3/.style={sibling distance=1.4cm},
  >=latex
]

\node[value] (v5) {5}
  child {node[value] (v1) {1}
    child {node[perpnode] (p0) {$\perp_0$}}
    child {node[perpnode] (p1) {$\perp_1$}}
  }
  child {node[value] (v2) {2}}
  child {node[value] (v3) {3}
    child {node[perpnode] (p2) {$\perp_2$}}
    child {node[value] (v4) {4}
      child {node[perpnode] (p3) {$\perp_3$}}
      child {node[perpnode] (p4) {$\perp_4$}}
    }
  }
  child {node[perpnode] (p5) {$\perp_5$}}
  child[invisible] {}
  child[invisible] {};

\node[rval, anchor=west]  at ([xshift=1mm]v5.east)   {$r=0$};
\node[rval, anchor=south]  at ([yshift=-1mm]v1.north)   {$r=2$};
\node[rval, anchor=north] at ([yshift=1mm]v2.south) {$r=1$};
\node[rval, anchor=east]  at ([xshift=-1mm]v3.west)   {$r=2$};
\node[rval, anchor=west]  at ([xshift=1mm]v4.east)   {$r=3$};
\node[rval, anchor=north] at ([yshift=2mm]p0.south) {$r=3$};
\node[rval, anchor=north] at ([yshift=2mm]p1.south) {$r=3$};
\node[rval, anchor=north] at ([yshift=2mm]p2.south) {$r=3$};
\node[rval, anchor=north] at ([yshift=2mm]p3.south) {$r=4$};
\node[rval, anchor=north] at ([yshift=2mm]p4.south) {$r=4$};
\node[rval, anchor=north] at ([yshift=2mm]p5.south) {$r=1$};

\draw[auxline] ($(p0)!0.5!(p1)+(0,0.33)$) -- ($(p0)!0.5!(p1)+(0,-0.33)$);
\draw[auxline] ($(p2)!0.5!(v4)+(0,0.33)$) -- ($(p2)!0.5!(v4)+(0,-0.33)$);
\draw[auxline] ($(p3)!0.5!(p4)+(0,0.33)$) -- ($(p3)!0.5!(p4)+(0,-0.33)$);
\draw[auxline] ($(v3)!0.5!(p5)+(0,0.33)$) -- ($(v3)!0.5!(p5)+(0,-0.33)$);

\end{tikzpicture}

\end{tabular}

\caption{Left: an example splay tree and its heavy-path decomposition. Right: the heap view obtained from the tree on the left, where heap-children are ordered as in the symmetric order and the dashed line separates left and right heap-children.}
\label{fig:overview:heap}
\label{fig:overview:heap-children-heap-parents}
\end{figure}

\paragraph{Lazy intervals for handling gap changes caused by rotations.}
The straightforward idea of using the sum of gaps as a potential while allowing rotations of the optimal tree does not work. A rotation in the optimal tree can change up to $\Theta(n)$ ranks and gaps, which would yield a useless $O(n)$-competitiveness proof.
Although a rotation changes the ranks of many nodes, the affected nodes form a consecutive interval $[i, j]$ in symmetric order, as they correspond to subtrees in the optimal tree. Therefore, we can group consecutive heap-children of each heavy path into so-called \emph{lazy intervals}. Instead of updating the gaps, we simply record the amount by which all gaps in an interval are shifted.
The amount by which the whole interval is shifted is called an \emph{interval gap} and the gaps without this shift are called \emph{point gaps}.
The important property of lazy intervals is that, for zig-zigs occurring within a lazy interval, the interval gap cancels out of the corresponding changes to the gaps. Thus, the point gaps behave as if they were the gaps themselves; see \Cref{fig:point-gaps-after-zig-zigs}.
In other words, it does not matter whether we first update the gaps and then perform zig-zigs, or we perform the zig-zigs and only later update the gaps lazily when needed.
Thus we can track the sum of point gaps instead of the sum of gaps and every rotation can then be handled by creating $O(1)$ new lazy intervals.


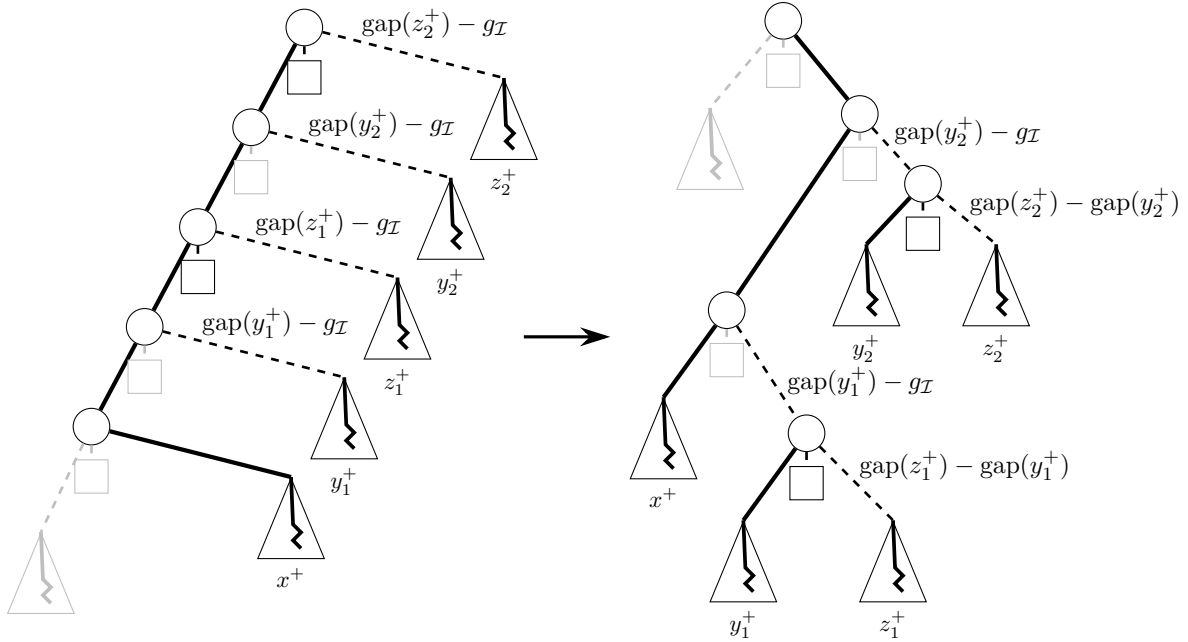
\begin{figure}[t]
\centering

\def\intervalgapshortcut{g_\mathcal I}
\def\vstretch{1.5}
\def\hstretch{0.8}

\begin{tikzpicture}[
    scale=0.88,
    transform shape,
    vertex/.style={
        circle,
        draw,
        fill=white,
        minimum size=5.5mm,
        inner sep=0pt
    },
    childsquare/.style={
        rectangle,
        draw,
        fill=white,
        minimum size=5mm,
        inner sep=0pt,
        font=\small
    },
    trig/.style={
        isosceles triangle,
        draw,
        fill=white,
        minimum height=9mm,
        minimum width=10mm,
        inner sep=0pt,
        shape border rotate=90
    },
    subtree/.style={
        trig,
        yshift=-11mm*\vstretch
    },
    dashed edge/.style={
        dashed,
        line width=1pt
    },
    solid edge/.style={
        line width=1.6pt
    },
    grey edge/.style={
        line width=1.6pt,
        lightgray
    },
    grey dashed edge/.style={
        dashed,
        line width=1pt,
        lightgray
    },
    grey node/.style={
        draw=lightgray,
        text=lightgray
    },
    grey subtree/.style={
        subtree,
        draw=lightgray,
        text=lightgray
    },
    gaplabel/.style={
        above right,
        pos=0.2,
        xshift=-1mm,
        yshift=-1mm
    },
    gaplabel2/.style={
        above right,
        pos=0.5,
        xshift=-1mm,
        yshift=-1mm
    },
    transform arrow/.style={
        -{Stealth[length=4mm,width=2.5mm]},
        line width=1.2pt
    },
    every label/.style={
        font=\small
    }
]

\def\childvec{3,0}
\def\ggap#1{$\gap(#1) - \intervalgapshortcut$}

\begin{scope}[shift={(-3.5,5.2)}]
\node[vertex] (LzTwo) at (-1*\hstretch, -1*\vstretch) {};
\node[vertex] (LyTwo) at (-2*\hstretch, -2*\vstretch) {};
\node[vertex] (LzOne) at (-3*\hstretch, -3*\vstretch) {};
\node[vertex] (LyOne) at (-4*\hstretch, -4*\vstretch) {};
\node[vertex] (Lx)    at (-5*\hstretch, -5*\vstretch) {};

\draw[solid edge] (LzTwo) -- (LyTwo) -- (LzOne) -- (LyOne) -- (Lx);

\node[childsquare,grey node] (Lxsq)     at ($(Lx)+(0,-0.75)$) {};
\node[childsquare,grey node] (LyOnesq)  at ($(LyOne)+(0,-0.75)$) {};
\node[childsquare]           (LzOnesq)  at ($(LzOne)+(0,-0.75)$) {};
\node[childsquare,grey node] (LyTwosq)  at ($(LyTwo)+(0,-0.75)$) {};
\node[childsquare]           (LzTwosq)  at ($(LzTwo)+(0,-0.75)$) {};

\node[grey subtree] (Lxm) at ($(Lx)+(-0.95*\hstretch,-0.55*\vstretch)$) {};
\node[subtree,label=below:{$x^{+}$}] (Lxp)
    at ($(Lx)+(\childvec)$) {};

\node[subtree,label=below:{$y_1^{+}$}] (LyOneP)
    at ($(LyOne)+(\childvec)$) {};

\node[subtree,label=below:{$z_1^{+}$}] (LzOneP)
    at ($(LzOne)+(\childvec)$) {};

\node[subtree,label=below:{$y_2^{+}$}] (LyTwoP)
    at ($(LyTwo)+(\childvec)$) {};

\node[subtree,label=below:{$z_2^{+}$}] (LzTwoP)
    at ($(LzTwo)+(\childvec)$) {};

\draw[grey dashed edge] (Lx)    -- (Lxm.apex);
\draw[grey dashed edge] (Lx)    -- (Lxsq);
\draw[solid edge]       (Lx)    -- (Lxp.apex);

\draw[grey dashed edge] (LyOne) -- (LyOnesq);
\draw[dashed edge]      (LyOne) -- node[gaplabel]
    {\ggap{y_1^+}} (LyOneP.apex);

\draw[dashed edge]      (LzOne) -- (LzOnesq);
\draw[dashed edge]      (LzOne) -- node[gaplabel]
    {\ggap{z_1^+}} (LzOneP.apex);

\draw[grey dashed edge] (LyTwo) -- (LyTwosq);
\draw[dashed edge]      (LyTwo) -- node[gaplabel]
    {\ggap{y_2^+}} (LyTwoP.apex);

\draw[dashed edge]      (LzTwo) -- (LzTwosq);
\draw[dashed edge]      (LzTwo) -- node[gaplabel]
    {\ggap{z_2^+}} (LzTwoP.apex);

\end{scope}

\draw[transform arrow] (-1,-0.95) -- (0.30,-0.95);


\node[vertex] (Rx)    at (2.90,  3.80) {};
\node[vertex] (RyTwo) at (4.05,  2.40) {};
\node[vertex] (RyOne) at (2.05, -0.55) {};
\node[vertex] (RzTwo) at (5.00,  1.35) {};
\node[vertex] (RzOne) at (3.25, -2.40) {};

\draw[solid edge]  (Rx)    -- (RyTwo);
\draw[solid edge]  (RyTwo) -- (RyOne);
\draw[dashed edge] (RyTwo) -- node[gaplabel2]
    {\ggap{y_2^+}} (RzTwo);
\draw[dashed edge] (RyOne) -- node[gaplabel2,pos=0.8]
    {\ggap{y_1^+}} (RzOne);

\node[childsquare,grey node] (Rxsq)    at ($(Rx)+(0,-0.75)$) {};
\node[childsquare,grey node] (RyTwosq) at ($(RyTwo)+(0,-0.75)$) {};
\node[childsquare,grey node] (RyOnesq) at ($(RyOne)+(0,-0.75)$) {};
\node[childsquare]           (RzTwosq) at ($(RzTwo)+(0,-0.75)$) {};
\node[childsquare]           (RzOnesq) at ($(RzOne)+(0,-0.75)$) {};

\node[grey subtree] (Rxm) at ($(Rx)+(-1.10,-0.55)$) {};

\node[subtree,label=below:{$x^{+}$}] (RxP)
    at ($(RyOne)+(-0.95,-0.55)$) {};

\node[subtree,label=below:{$y_1^{+}$}] (RyOneP)
    at ($(RzOne)+(-0.95,-0.55)$) {};

\node[subtree,label=below:{$z_1^{+}$}] (RzOneP)
    at ($(RzOne)+(1.30,-0.55)$) {};

\node[subtree,label=below:{$y_2^{+}$}] (RyTwoP)
    at ($(RzTwo)+(-0.85,-0.15)$) {};

\node[subtree,label=below:{$z_2^{+}$}] (RzTwoP)
    at ($(RzTwo)+(1.10,-0.15)$) {};

\draw[grey dashed edge]        (Rx)    -- (Rxm.apex);
\draw[grey dashed edge] (Rx)    -- (Rxsq);

\draw[solid edge]       (RyOne) -- (RxP.apex);
\draw[grey dashed edge] (RyOne) -- (RyOnesq);

\draw[solid edge]       (RzOne) -- (RyOneP.apex);
\draw[dashed edge]      (RzOne) -- (RzOnesq);
\draw[dashed edge]      (RzOne) -- node[gaplabel2]
    {$\gap(z_1^+)-\gap(y_1^+)$} (RzOneP.apex);

\draw[grey dashed edge] (RyTwo) -- (RyTwosq);

\draw[solid edge]       (RzTwo) -- (RyTwoP.apex);
\draw[dashed edge]      (RzTwo) -- (RzTwosq);
\draw[dashed edge]      (RzTwo) -- node[gaplabel2]
    {$\gap(z_2^+)-\gap(y_2^+)$} (RzTwoP.apex);


\foreach \T in {Lxp,LyOneP,LzOneP,LyTwoP,LzTwoP,RxP,RyOneP,RzOneP,RyTwoP,RzTwoP}{
    \draw[line width=1.4pt]
        ($(\T.center)+(0.15,-0.18)$)
        -- ++(-0.12,0.12)
        -- ++( 0.12,0.12)
        -- ++(-0.12,0.12)
        -- (\T.apex);
}

\foreach \T in {Lxm,Rxm}{
    \draw[line width=1.4pt, lightgray]
        ($(\T.center)+(0.15,-0.18)$)
        -- ++(-0.12,0.12)
        -- ++( 0.12,0.12)
        -- ++(-0.12,0.12)
        -- (\T.apex);
}

\end{tikzpicture}

\caption{The change of point gaps after two consecutive zig-zigs over heavy edges in one particular configuration with all heap-children in the same lazy interval $\mathcal I$. We use $\intervalgapshortcut$ as a shorthand for $\lazygap(\mathcal I)$. $\gap(\cdot)$ denotes the gap before the zig-zigs.
Note that for the resulting gap of $z_i^+$, the interval gaps cancel.}
\label{fig:point-gaps-after-zig-zigs}

\end{figure}

\paragraph{Contracting gaps to bound the number of zig-zigs.}
The remaining difficulty is posed by \emph{boundary zig-zigs}, which affect heap-children from two different lazy intervals.
To simplify the discussion of zig-zigs, recall that a zig-zig takes three consecutive heap-children and hangs the two with higher rank in a lazy interval of the one with lower rank.
We decompose this event into two \emph{pairings}, in each of which one heap-child is hung below another heap-child with lower rank.
In this terminology, the problem lies with \emph{boundary pairings}, where the two heavy paths are not from the same lazy interval.
An attempt to resolve this would be to modify the optimal tree so that it always has logarithmic depth and thus logarithmic ranks while incurring only constant slowdown. But even when charging to this modified tree,
each boundary pairing can increase the sum of point gaps by up to $\Theta(\log n)$. This again could not be used to prove an $o(\log n)$ competitive ratio.

The solution is to contract the point gaps by passing them through an integer-valued logarithmic function. As a potential, we then sum the contracted point gaps instead of the point gaps. As a result, each boundary pairing increases the sum of contracted point gaps by $O(\log\log n)$, so we can hope for $\tilde{O}(\log \log n)$-competitiveness.
The contraction, however, complicates matters, as a decrease in a point gap does not imply a decrease in the corresponding contracted point gap, so not all internal pairings decrease the sum of contracted point gaps.
The internal pairings where the point gaps are of the same magnitude, which we call \emph{good}, still decrease the sum of contracted point gaps.
This is because decreasing the input of the logarithm by a constant fraction decreases its output by a constant.
We are able to bound good pairings by the total increase in the sum of contracted point gaps.
The other internal pairings are called \emph{bad}, and bounding them requires more work.
See also \Cref{fig:contracted-gaps-after-zig-zigs}.


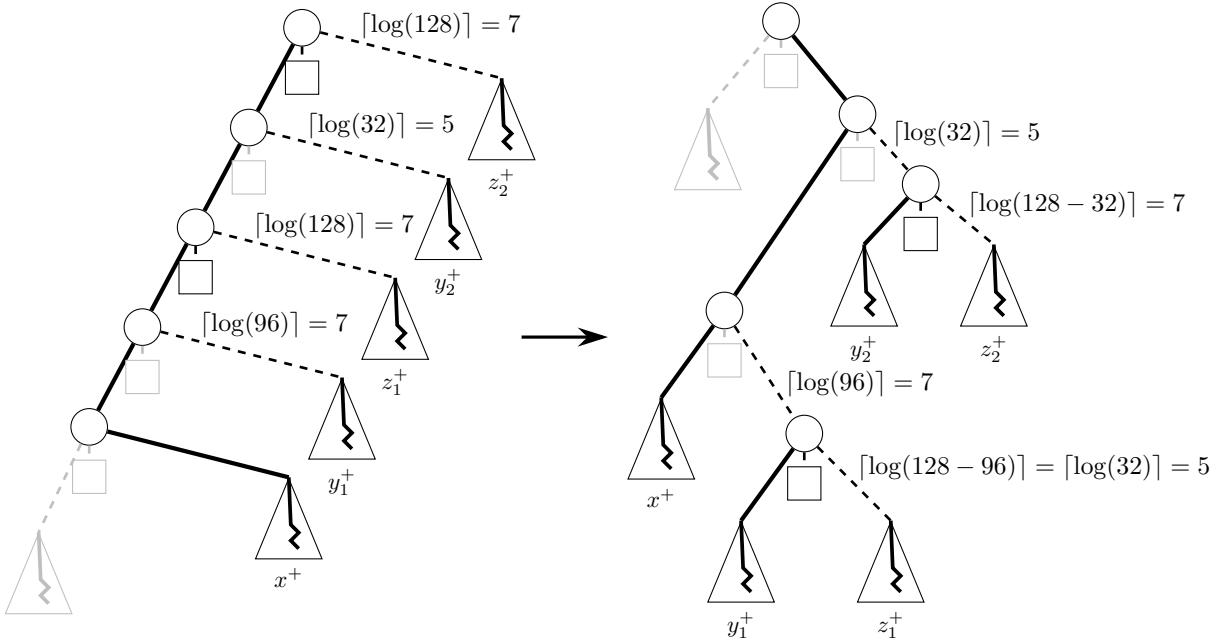
\begin{figure}[t]
\centering

\def\vstretch{1.5}
\def\hstretch{0.8}

\begin{tikzpicture}[
    scale=0.88,
    transform shape,
    vertex/.style={
        circle,
        draw,
        fill=white,
        minimum size=5.5mm,
        inner sep=0pt
    },
    childsquare/.style={
        rectangle,
        draw,
        fill=white,
        minimum size=5mm,
        inner sep=0pt,
        font=\small
    },
    trig/.style={
        isosceles triangle,
        draw,
        fill=white,
        minimum height=9mm,
        minimum width=10mm,
        inner sep=0pt,
        shape border rotate=90
    },
    subtree/.style={
        trig,
        yshift=-11mm*\vstretch
    },
    dashed edge/.style={
        dashed,
        line width=1pt
    },
    solid edge/.style={
        line width=1.6pt
    },
    grey edge/.style={
        line width=1.6pt,
        lightgray
    },
    grey dashed edge/.style={
        dashed,
        line width=1pt,
        lightgray
    },
    grey node/.style={
        draw=lightgray,
        text=lightgray
    },
    grey subtree/.style={
        subtree,
        draw=lightgray,
        text=lightgray
    },
    gaplabel/.style={
        above right,
        pos=0.2,
        xshift=-1mm,
        yshift=-1mm
    },
    gaplabel2/.style={
        above right,
        pos=0.5,
        xshift=-1mm,
        yshift=-1mm
    },
    transform arrow/.style={
        -{Stealth[length=4mm,width=2.5mm]},
        line width=1.2pt
    },
    every label/.style={
        font=\small
    }
]

\def\childvec{3,0}

\begin{scope}[shift={(-3.5,5.2)}]
\node[vertex] (LzTwo) at (-1*\hstretch, -1*\vstretch) {};
\node[vertex] (LyTwo) at (-2*\hstretch, -2*\vstretch) {};
\node[vertex] (LzOne) at (-3*\hstretch, -3*\vstretch) {};
\node[vertex] (LyOne) at (-4*\hstretch, -4*\vstretch) {};
\node[vertex] (Lx)    at (-5*\hstretch, -5*\vstretch) {};

\draw[solid edge] (LzTwo) -- (LyTwo) -- (LzOne) -- (LyOne) -- (Lx);

\node[childsquare,grey node] (Lxsq)     at ($(Lx)+(0,-0.75)$) {};
\node[childsquare,grey node] (LyOnesq)  at ($(LyOne)+(0,-0.75)$) {};
\node[childsquare]           (LzOnesq)  at ($(LzOne)+(0,-0.75)$) {};
\node[childsquare,grey node] (LyTwosq)  at ($(LyTwo)+(0,-0.75)$) {};
\node[childsquare]           (LzTwosq)  at ($(LzTwo)+(0,-0.75)$) {};

\node[grey subtree] (Lxm) at ($(Lx)+(-0.95*\hstretch,-0.55*\vstretch)$) {};
\node[subtree,label=below:{$x^{+}$}] (Lxp)
    at ($(Lx)+(\childvec)$) {};

\node[subtree,label=below:{$y_1^{+}$}] (LyOneP)
    at ($(LyOne)+(\childvec)$) {};

\node[subtree,label=below:{$z_1^{+}$}] (LzOneP)
    at ($(LzOne)+(\childvec)$) {};

\node[subtree,label=below:{$y_2^{+}$}] (LyTwoP)
    at ($(LyTwo)+(\childvec)$) {};

\node[subtree,label=below:{$z_2^{+}$}] (LzTwoP)
    at ($(LzTwo)+(\childvec)$) {};

\draw[grey dashed edge] (Lx)    -- (Lxm.apex);
\draw[grey dashed edge] (Lx)    -- (Lxsq);
\draw[solid edge]       (Lx)    -- (Lxp.apex);

\draw[grey dashed edge] (LyOne) -- (LyOnesq);
\draw[dashed edge]      (LyOne) -- node[gaplabel]
    {$\lceil\log(96)\rceil=7$} (LyOneP.apex);

\draw[dashed edge]      (LzOne) -- (LzOnesq);
\draw[dashed edge]      (LzOne) -- node[gaplabel]
    {$\lceil\log(128)\rceil=7$} (LzOneP.apex);

\draw[grey dashed edge] (LyTwo) -- (LyTwosq);
\draw[dashed edge]      (LyTwo) -- node[gaplabel]
    {$\lceil\log(32)\rceil=5$} (LyTwoP.apex);

\draw[dashed edge]      (LzTwo) -- (LzTwosq);
\draw[dashed edge]      (LzTwo) -- node[gaplabel]
    {$\lceil\log(128)\rceil=7$} (LzTwoP.apex);

\end{scope}

\draw[transform arrow] (-1,-0.95) -- (0.30,-0.95);


\node[vertex] (Rx)    at (2.90,  3.80) {};
\node[vertex] (RyTwo) at (4.05,  2.40) {};
\node[vertex] (RyOne) at (2.05, -0.55) {};
\node[vertex] (RzTwo) at (5.00,  1.35) {};
\node[vertex] (RzOne) at (3.25, -2.40) {};

\draw[solid edge]  (Rx)    -- (RyTwo);
\draw[solid edge]  (RyTwo) -- (RyOne);
\draw[dashed edge] (RyTwo) -- node[gaplabel2]
    {$\lceil\log(32)\rceil=5$} (RzTwo);
\draw[dashed edge] (RyOne) -- node[gaplabel2,pos=0.8]
    {$\lceil\log(96)\rceil=7$} (RzOne);

\node[childsquare,grey node] (Rxsq)    at ($(Rx)+(0,-0.75)$) {};
\node[childsquare,grey node] (RyTwosq) at ($(RyTwo)+(0,-0.75)$) {};
\node[childsquare,grey node] (RyOnesq) at ($(RyOne)+(0,-0.75)$) {};
\node[childsquare]           (RzTwosq) at ($(RzTwo)+(0,-0.75)$) {};
\node[childsquare]           (RzOnesq) at ($(RzOne)+(0,-0.75)$) {};

\node[grey subtree] (Rxm) at ($(Rx)+(-1.10,-0.55)$) {};

\node[subtree,label=below:{$x^{+}$}] (RxP)
    at ($(RyOne)+(-0.95,-0.55)$) {};

\node[subtree,label=below:{$y_1^{+}$}] (RyOneP)
    at ($(RzOne)+(-0.95,-0.55)$) {};

\node[subtree,label=below:{$z_1^{+}$}] (RzOneP)
    at ($(RzOne)+(1.30,-0.55)$) {};

\node[subtree,label=below:{$y_2^{+}$}] (RyTwoP)
    at ($(RzTwo)+(-0.85,-0.15)$) {};

\node[subtree,label=below:{$z_2^{+}$}] (RzTwoP)
    at ($(RzTwo)+(1.10,-0.15)$) {};

\draw[grey dashed edge]        (Rx)    -- (Rxm.apex);
\draw[grey dashed edge] (Rx)    -- (Rxsq);

\draw[solid edge]       (RyOne) -- (RxP.apex);
\draw[grey dashed edge] (RyOne) -- (RyOnesq);

\draw[solid edge]       (RzOne) -- (RyOneP.apex);
\draw[dashed edge]      (RzOne) -- (RzOnesq);
\draw[dashed edge]      (RzOne) -- node[gaplabel2]
    {$\lceil\log(128-96)\rceil=\lceil\log(32)\rceil=5$} (RzOneP.apex);

\draw[grey dashed edge] (RyTwo) -- (RyTwosq);

\draw[solid edge]       (RzTwo) -- (RyTwoP.apex);
\draw[dashed edge]      (RzTwo) -- (RzTwosq);
\draw[dashed edge]      (RzTwo) -- node[gaplabel2]
    {$\lceil\log(128-32)\rceil=7$} (RzTwoP.apex);


\foreach \T in {Lxp,LyOneP,LzOneP,LyTwoP,LzTwoP,RxP,RyOneP,RzOneP,RyTwoP,RzTwoP}{
    \draw[line width=1.4pt]
        ($(\T.center)+(0.15,-0.18)$)
        -- ++(-0.12,0.12)
        -- ++( 0.12,0.12)
        -- ++(-0.12,0.12)
        -- (\T.apex);
}

\foreach \T in {Lxm,Rxm}{
    \draw[line width=1.4pt, lightgray]
        ($(\T.center)+(0.15,-0.18)$)
        -- ++(-0.12,0.12)
        -- ++( 0.12,0.12)
        -- ++(-0.12,0.12)
        -- (\T.apex);
}

\end{tikzpicture}

\caption{The change to contracted point gaps after two consecutive zig-zigs over heavy edges with corresponding heap-children in the same lazy interval. The pairing $(y_1^+, z_1^+)$ is good as the two heap-children had point gaps of the same magnitude, so the contracted point gap of $ z_1^+$ decreases. The pairing $(y_2^+, z_2^+)$ is bad as the point gaps are of different magnitude, therefore the contracted point gap of $ z_2^+$ is not guaranteed to decrease. Note that in reality we use $\lceil 1 + \log(1 + x) \rceil$ for the contracted gap, to avoid edge cases. We use $\lceil \log(x) \rceil$ in this figure instead for simplicity. In both, the logarithms are binary.}
\label{fig:contracted-gaps-after-zig-zigs}

\end{figure}

The next step is to show that there are not too many boundary pairings. Some occur at the boundaries of very large lazy intervals, and their effects can be subsumed by internal pairings. We call these pairings \emph{unimportant}.
On the other hand, important pairings cannot happen too often within any one lazy interval, and since we do not create too many lazy intervals, their total number is bounded.

Once we also bound the number of bad pairings, we can relate the number of zig-zigs to the work done by the optimal tree: the number of zig-zigs exceeds the cost of the optimal tree by at most a factor of $\tilde{O}(\log\log n)$, plus an additional factor of $\tilde{O}(n\log\log n)$ for initialization.

\paragraph{Relating splaying to pairing heaps and using the core of pairing heap analysis to bound bad pairings.}
To bound the bad pairings, we recall the heap view of heavy paths and realize that splaying behaves similarly to how delete-min behaves in pairing heaps.

In pairing heaps, delete-min works by first removing the minimum.
Afterwards, the individual children are correctly heap-ordered, but in separate heaps. 

Next, we perform a pairing step: we partition the child list into consecutive pairs, and for each pair, the lower-value child is kept in the list while the higher-value child is rehung below it. 
Finally, the assembly step makes the least-value child the new root~\cite{Fredman1986PairingHeap}.
In splay trees, the splayed vertex is not deleted but moved to the root. Among its heap-children on the right, every consecutive triplet is combined through two pairings: the heap-child with the lowest rank becomes the heap-parent of the other two. The right heap-child with the least rank remains a heap-child of the splayed vertex. An analogous process occurs to the left of the splayed vertex.
Hence, the core process on either side of the splayed vertex is nearly the same as that of delete-min in pairing heaps.
Therefore, we abstract the core of the analysis of the amortized complexity of pairing heaps~\cite{Sinnamon2025,PairingHeaps2026} into an interface
that also applies to splay trees. We show that, by modifying lazy intervals in a controlled way through a few allowed operations, the lazy intervals can be mapped to this interface. This allows us to bound the number of bad internal pairings in terms of the numbers of good internal pairings and boundary pairings, which is the last ingredient needed to bound the number of zig-zigs.

\paragraph{Bounding the number of zig-zags by bends and by normalizing the optimal tree.}
The last part is to bound the zig-zags. Unlike in the static case, the contraction does not allow us to bound zig-zags by the sum of contracted point gaps.
We bound them by introducing another combinatorial element -- \emph{bends}. A bend is an internal node where the heavy path ``turns'' -- it is a left child, but the heavy path continues to the right, or vice versa.
To charge zig-zags to bends and ensure that each splay creates only a controlled number of bends, we require, among other things, that the optimal tree be \emph{ribless}, i.e., that two vertices that are close in symmetric order also have similar heights in the tree. We show that we can impose these assumptions on the optimal tree while slowing it down by only a constant factor.
Then, every optimal-tree rotation creates $O(1)$ bends, and every splay creates $O(1)$ bends plus one for each zig-zig.
Every zig-zag removes one bend, so the total number of zig-zags can be asymptotically bounded by the number of zig-zigs, the cost of the optimal tree, and the maximum initial number $n$ of bends.

Overall, these ingredients suffice to prove $\tilde{O}(\log\log n)$-competitiveness for sequences of length $\tilde{\Omega}(n\log\log n)$.
The same techniques can also be used to prove the split conjecture within a factor of $\tilde{O}(\log\log n)$ (\Cref{app:split-conjecture}).

\section{Proof of splay tree competitiveness}

In this section, we prove that the splay tree is $O(\log\log n \cdot \log^2\log\log n)$-competitive; that is, we prove \Cref{thm:main}.
Note that we also formally define and re-introduce the concepts introduced informally in previous sections.

\begin{theorem}\label{thm:main}
    Assume a set $\mathcal{U}$ of $n$ elements  and an access sequence  $X = x_1, \dots, x_m$, with $x_i \in \mathcal{U}$.
    Let $T$ be the splay tree and let $\cost (T, X)$ be the time complexity in the BST model of the splay tree, starting from an arbitrary configuration, when it performs accesses as in $X$.
    Furthermore, let $T^\text{OPT}_X$ be the optimal tree for the sequence $X$, i.e., the dynamic tree that incurs the lowest cost among all trees.
    Then
    $\cost (T, X) = O((\cost (T^\text{OPT}_X, X) + n) \cdot \log\log n \cdot \log^2\log\log n)$.
\end{theorem}

We remark that at least an additional additive term of $O(n)$ is necessary because we do not assume any initial structure for the splay tree. Thus, the first access itself can require $\Omega(n)$ time, while the optimal tree can still serve it in constant time.

The $o(n\log n)$ additive term in \cref{thm:main} makes it possible to apply our techniques to the split conjecture for splay trees as well. There, the goal is to match the $\Theta(n)$ cost of the optimal tree, so any large additive factor would hinder this application. We prove the split conjecture with total cost $O(n \log\log n \log^2\log\log n)$ in \Cref{app:split-conjecture}.

\paragraph{Ternary tree view of the splay tree.}
To simplify all proofs later on, we consider the following \emph{ternary view} of the splay tree that ensures that only leaves store values, and furthermore, every non-leaf has exactly three children.
To obtain it from the standard binary view, for every node $v$, we add an additional middle child where $v$'s original value is moved, so internal vertices no longer hold any value. The nodes added this way are called \emph{value nodes}. We keep $v$'s original left and right child.

Furthermore, to ensure that every node has exactly three children, we add $n+1$ \emph{null nodes}
$\perp_0, \perp_1, \dots \perp_n$, where $\perp_i$ falls in between $i$ and $i+1$ in the symmetric order, $\perp_0$ is left of $1$ and $\perp_n$ is right of $n$. Observe that for any tree, there is a unique position of the null nodes. One can imagine the null node $\perp_i$ corresponding to value $i + 1/2$.
Null nodes and value nodes are collectively called \emph{external nodes}, while all other nodes are called \emph{internal}. An illustration of the ternary view is in \Cref{fig:overview:ternary-view}.
The splays of the tree work the same way as in the binary view, never detaching the value node
from its parent internal node. Since the null nodes are never accessed, rotations corresponding to the splay never attach any children to null nodes. A rotation in the ternary view is depicted in \Cref{fig:overview:rotation}.

Similarly as for $T$, we add the null nodes $\perp_i$ to the optimal tree $T^\text{OPT}$. However, we do not impose the ternary view on the optimal tree.

\subsection{Restrictions on the optimal tree}\label{sec:constraints-on-opt-overview}

Our proof of the competitiveness of splay trees proceeds by relating the splay tree to the optimal tree. However, we require several properties of the tree used in this comparison. Therefore, we do not relate the splay tree to the optimal tree $T^\text{OPT}$ itself, but rather to a nearly optimal tree $T^*$ whose complexity is only a constant factor larger than that of $T^\text{OPT}$ and which has the following properties simultaneously at all times.

\begin{enumerate}
	\item $T^*$ is \emph{ribless}, meaning that all paths that go only to left children, as well as all paths that go only to right children, have length $O(1)$, except possibly for the two spines from the root toward the left-most and right-most elements.
	\item All rotations in $T^*$ happen within constant depth of the root.
	\item The depth of $T^*$ is $O(\log n)$.
\end{enumerate}

The existence of such a tree $T^*$ is implied by \Cref{lemma:opt-constraints} and is proved in \Cref{S:normal-form}.

\begin{restatable}{theorem}{normalform}   
\label{lemma:opt-constraints}
For any access sequence $X = x_1,\dots,x_m$ from $\mathcal{U} = [n]$, there exists a sequence of trees $T_0,T_1,\dots,T_\ell$ and a sequence of times $t_1 < \dots < t_m \in [\ell]$,
where $\ell=\OO(\cost(T^\text{OPT}_X))$, such that
\begin{enumerate}
    \item For all $i \in [\ell]$, $T_i$ is a binary search tree on $[n]$ that is ribless of depth $\OO(\log n)$,
    \item For all $i \in [\ell]$, $T_{i}$ is obtained from $T_{i-1}$ by a single rotation at depth at most $\OO(1)$,
    \item For all $j \in [m]$, $T_{t_j}$ has root $x_j$.
\end{enumerate}
\end{restatable}

\Cref{lemma:opt-constraints} guarantees the existence of the nearly optimal tree $T^*$ without null nodes. However, we can always add null nodes to $T^*$ while increasing both the depth and the length of every rib by only one. Therefore, $T^*$ is guaranteed to exist with null nodes as well.
From now on, we relate the splay tree to this nearly optimal tree $T^*$.

The need for riblessness is explained by the following lemma, which we later use to bound the number of operations on the splay tree.

\begin{corollary}\label{cor:depth-by-riblessness}
	The depths of  $x$, $\perp_{x-1}$, and $\perp_x$,
	in $T^*$ differ only by $O(1)$ for any $x$.
\end{corollary}

\begin{proof}
		Fix $x$, and let $x^L$ be the left child of $x$.
	    If $x^L = \perp_{x-1}$, the claim follows immediately. Otherwise, $\perp_{x-1}$ is the largest element in the subtree of $x^L$, so there can be no left edge on the path from $x^L$ to $\perp_{x-1}$. Hence, this path contains only right edges and therefore forms a right rib or a spine.
	    Since $x$ lies to the right of this rib, it cannot be a spine to the maximum and hence is a proper rib. By riblessness, the path to the parent of $\perp_{x-1}$ cannot be longer than $O(1)$. Therefore, the depths of $x$ and $\perp_{x-1}$ differ by at most $O(1)$.
		The argument for $\perp_x$ is analogous: it is the smallest element in the subtree of the right child of $x$, so the path to it uses only left edges and can be only an additive $O(1)$ deeper.
\end{proof}

\subsection{Heavy paths and lazy intervals}

In order to prove the competitiveness of splay trees, we need several combinatorial structures defined on the splay tree with respect to the state of the optimal tree at any given moment.
The principal structure is a decomposition of the splay tree into heavy paths, where each heavy path always leads to the element closest to the root in the optimal tree (\Cref{sec:heavy-light}). The heavy paths are used to define gaps, which are at the core of explaining why splay operations can be paid for.
Since these gaps can change for many heavy paths upon a rotation in the optimal tree,
lazy intervals are defined to handle these global changes (\Cref{sec:lazy-intervals}).
Lazy intervals are central to our bound on the operations performed by a splay tree. We define allowed operations on lazy intervals (\Cref{sec:lazy-interval-operations}) and, in later sections, use them to show that every zig-zig performs at least one such operation and that the total number of these operations is small.

\subsubsection{Decomposition into heavy paths and light edges}\label{sec:heavy-light}

For every external node $v$ of $T$, we define its \emph{rank} as its depth in the nearly optimal tree $T^*$, i.e., $\rank(v) = \depth_{T^*}(v)$.
In particular, since the depth of $T^*$ is logarithmic by \Cref{lemma:opt-constraints}, the ranks themselves are in $O(\log n)$.
Based on the ranks, we classify all edges as heavy or light. An edge $uv$ is \emph{heavy} if the minimum rank in the subtree of $v$ is the same as the minimum rank in the subtree of $u$; otherwise, it is \emph{light}. We define \emph{heavy paths} as the connected components that remain after all light edges are removed. The heavy path that contains the root is called the \emph{root heavy path}. For an illustration of ranks and heavy paths, see \Cref{fig:overview:heavy-light}.
The following folklore observation is crucial to ensure that heavy paths are indeed paths.

\begin{lemma}\label{lemma:unique-depth-minimum}
	For any tree $T^*$ and any consecutive interval of external nodes of the splay tree in the symmetric order $[a, b]$, there is a unique node $v \in [a, b]$ with minimum rank.
\end{lemma}

\begin{proof}
	The external nodes of the splay tree correspond to all nodes of the optimal tree.
    We prove the claim by induction on the number of elements in the optimal tree. 
    Consider some tree $T^*$ and an interval $[a, b]$.
	    If $T^*$ has only a single vertex, or if the root of $T^*$ is contained in the interval, then it is the unique minimizer. Otherwise, the interval is contained entirely in either the left subtree or the right subtree.
	    Let $T'$ be the subtree.
	    The ranks of nodes in $T'$ are exactly one larger than they would be if $T'$ were the whole tree, which preserves uniqueness of the minimum.
	    By the induction hypothesis, any interval in $T'$ has a unique minimizer, so the corresponding interval in the respective subtree does as well.
\end{proof}

\begin{corollary}\label{cor:heavy-paths}
		For every internal node of $T$, there is exactly one outgoing heavy edge.
\end{corollary}

\begin{proof}
		The subtree of every internal node $z$ corresponds to an interval in the symmetric order. Hence, by \Cref{lemma:unique-depth-minimum}, there is exactly one external node minimizing the rank. Since the child subtrees of $z$ partition the interval into disjoint subintervals, exactly one of them contains the minimum-rank node. Thus, there is exactly one outgoing heavy edge, and it goes to this child.
\end{proof}

\Cref{cor:heavy-paths} has several important implications. First, heavy paths are indeed paths. Second, since the bottom-most node of a heavy path cannot be an internal node, the bottom-most node of every heavy path is an external node. This is the only external node on this path, since external nodes have no outgoing edges. Therefore, there is a bijection between external nodes and heavy paths.
For a heavy path $P$ ending in an external node $v$, we denote $\Path(v) \coloneq P$ and $\bottom(P) \coloneq v$.
When it is clear from the context, for convenience, we sometimes use the heavy path and the bottom-most vertex interchangeably.
Furthermore, let $\top(P)$ denote the topmost node of $P$. We say that the rank of a heavy path is the rank of its bottom-most node.
Note that some heavy paths may contain no internal nodes and consist solely of their corresponding external node.
The concept of heavy paths and light edges is illustrated in \Cref{fig:overview:heavy-light}.

\paragraph{Relations between heavy paths.}
Every light edge connects two different heavy paths with bottom-most vertices $v_\text{top}$ and $v_\text{bottom}$, with the top one having strictly smaller rank than the bottom one.
Thus, if we imagine another tree in which every light edge induces an edge from $v_\text{top}$ to $v_\text{bottom}$, this new tree is heap-ordered by the ranks of the heavy paths.
We call this tree the heap view of the heavy paths and it is illustrated in \Cref{fig:overview:heap-children-heap-parents}.
We say that $v_\text{top}$ is the \emph{heap-parent} of $v_\text{bottom}$ and that $v_\text{bottom}$ is a \emph{heap-child} of $v_\text{top}$.
	It is a \emph{left heap-child} if $v_\text{bottom}$ is left of $v_\text{top}$ in the symmetric order, and a right heap-child otherwise.
For simplicity, we also say that heavy paths are heap-children and heap-parents when we mean that their bottom-most nodes are heap-children and heap-parents.
Equivalently, left heap-children are those that are connected by a light edge to its heap-parent at an internal node, where the heavy path goes to the right or if the light edge is left in case the heavy path goes via the middle edge; vice versa for right heap-children.
As a heavy path is connected as a heap-child only from its topmost node, every heavy path has exactly one heap-parent except for the root heavy path, which has no heap-parent. A heavy path can have multiple heap-children.
We denote by $\heapparent(P)$ the heap-parent of the heavy path $P$, and by $\leftheapchildren(P)$, $\rightheapchildren(P)$, and $\heapchildren(P)$ the sets of its left, right, and all heap-children, respectively. We consider the left, right, and all heap-children to be ordered top-down by the order in which they hang off $P$, with heap-children connected through middle edges considered closer to the bottom than those connected through left or right edges.
On each side, top-down order is monotone in symmetric order: for left heap-children it agrees with increasing symmetric order, while for right heap-children it agrees with decreasing symmetric order. Hence a consecutive interval of left or right heap-children is still consecutive in symmetric order.

\begin{definition}[Gap]
	We define the \emph{gap} of a non-root heavy path $P$ as $\gap(P) = \rank(P) - \rank(\heapparent(P))$ and $\gap(P) = 0$ if $P$ is the root heavy path.
\end{definition}

Next, we define \emph{heap-descendants} of a node as itself, all its heap-children and all the heap-children of the heap-children and so on. The heap-descendants of $P$ are all the bottom-most vertices of heavy paths corresponding to the external nodes within the subtree rooted at  $\top(P)$. Likewise, we define \emph{heap-ancestors} of $P$ as all the heavy paths that have $P$ as their heap-descendant.
Heap-descendants are the descendants in the heap view, while heap-ancestors are the ancestors in this view; see \Cref{fig:overview:heap-children-heap-parents}.
There is one important property of heap-ancestors. A heavy path of rank $R$ has at most $R+1$ heap-ancestors, because moving from a heap-child to its heap-parent strictly decreases rank, and ranks are integral and nonnegative.

\paragraph{Bends.} The last structural elements of a heavy path that we need are bends. A node is called a \emph{bend} if it is not the top-most node of the heavy path on which it lies and either it is the right child of its parent and the heavy path continues to its left child, or it is the left child of its parent and the heavy path continues to its right child. Note that only internal nodes can be bends, as external nodes do not have children.
Later, we use the total number of bends ever created to argue that the number of zig-zags is small.

\begin{figure}[p]
\centering

\begin{tikzpicture}[
    vertex/.style={
        circle,
        draw,
        fill=white,
        minimum size=6mm,
        inner sep=0pt
    },
    childsquare/.style={
        rectangle,
        draw,
        fill=white,
        minimum size=5mm,
        inner sep=0pt
    },
    trig/.style={
        isosceles triangle,
        draw,
        fill=white,
        minimum height=8mm,
        minimum width=8mm,
        inner sep=0pt
    },
    left interval brace/.style={
        decorate,
        decoration={brace,mirror,amplitude=5pt,raise=2pt}
    },
    right interval brace/.style={
        decorate,
        decoration={brace,amplitude=5pt,raise=2pt}
    },
    interval label left/.style={
        midway,
        below=8pt,
        left=8pt
    },
    interval label right/.style={
        midway,
        below=8pt,
        right=8pt
    }
]


\node[vertex] (n1) at (-4.0,  0.0) {};
\node[vertex] (n2) at ( 3.4, -1.2) {};
\node[vertex] (n3) at ( 2.6, -2.6) {};
\node[vertex] (n4) at (-2.6, -4.0) {};
\node[vertex] (n5) at (-1.8, -5.4) {};
\node[vertex] (n6) at ( 1.4, -6.8) {};
\node[vertex] (n7) at ( 0.8, -8.2) {};
\node[vertex] (n8) at (-0.4, -9.6) {};


\draw[line width=2pt]
    (n1) -- (n2)
    -- (n3)
    -- (n4)
    -- (n5)
    -- (n6)
    -- (n7)
    -- (n8);


\foreach \i in {1,...,7}{
    \node[childsquare] (s\i) at ($(n\i)+(0,-1.0)$) {};
    \draw[dashed] (n\i) -- (s\i);
}

\node[childsquare] (s8) at ($(n8)+(0,-1.0)$) {};
\draw[line width=2pt] (n8) -- (s8);


\foreach \i in {1,4,5}{
    \node[trig,shape border rotate=90] (l\i)
        at ($(n\i)+(-1.4,-1.3)$) {};

    \draw[dashed] (n\i) -- (l\i.apex);

    \draw[line width=1.4pt]
        ($(l\i.center)+(0.15,-0.18)$)
        -- ++(-0.12,0.12)
        -- ++( 0.12,0.12)
        -- ++(-0.12,0.12)
        -- (l\i.apex);
}


\foreach \i in {2,3,6,7}{
    \node[trig,shape border rotate=90] (r\i)
        at ($(n\i)+(1.4,-1.3)$) {};

    \draw[dashed] (n\i) -- (r\i.apex);

    \draw[line width=1.4pt]
        ($(r\i.center)+(0.15,-0.18)$)
        -- ++(-0.12,0.12)
        -- ++( 0.12,0.12)
        -- ++(-0.12,0.12)
        -- (r\i.apex);
}


\node[trig,shape border rotate=90] (l8)
    at ($(n8)+(-1.4,-1.3)$) {};

\draw[dashed] (n8) -- (l8.apex);

\draw[line width=1.4pt]
    ($(l8.center)+(0.15,-0.18)$)
    -- ++(-0.12,0.12)
    -- ++( 0.12,0.12)
    -- ++(-0.12,0.12)
    -- (l8.apex);

\node[trig,shape border rotate=90] (r8)
    at ($(n8)+(1.4,-1.3)$) {};

\draw[dashed] (n8) -- (r8.apex);

\draw[line width=1.4pt]
    ($(r8.center)+(0.15,-0.18)$)
    -- ++(-0.12,0.12)
    -- ++( 0.12,0.12)
    -- ++(-0.12,0.12)
    -- (r8.apex);


\draw[left interval brace]
    ($(s1.south west)+(-2.4,1.75)$)
    --
    ($(l1.south west)+(-0.2,-0.5)$)
    node[interval label left] {$I^L_{2}$};

\draw[left interval brace]
    ($(l4.south west)+(-0.8,1.25)$)
    --
    ($(s8.south west)+(-2,-0.75)$)
    node[interval label left] {$I^L_{1}$};

\draw[right interval brace]
    ($(r2.south east)+(0.8,1.25)$)
    --
    ($(s3.south east)+(2.2,-0.75)$)
    node[interval label right] {$I^R_{2}$};

\draw[right interval brace]
    ($(r6.south east)+(0.8,1.25)$)
    --
    ($(s8.south east)+(2.6,-0.75)$)
    node[interval label right] {$I^R_{1}$};

\draw[right interval brace]
    ($(n2.north east)+(2.4,1.4)$)
    --
    ($(n2.north east)+(2.4,0.6)$)
    node[midway,right=12pt] {$I^R_{3}$};

\end{tikzpicture}

\caption{A possible partition of heap-children into lazy intervals. Each lazy interval covers a contiguous segment of either the left heap-children (intervals denoted by $L$) or the right heap-children (intervals denoted by $R$). $I^L_2$ and $I^R_3$ are the growing lazy intervals; note that $I^R_3$ is empty. All other intervals are shrinking. $I^L_1$ is the only broken lazy interval. There are multiple valid decompositions into lazy intervals.}
\label{fig:lazy-intervals}

\end{figure}

\subsubsection{Lazy intervals}\label{sec:lazy-intervals}

Consider a heavy path $P$. A \emph{left lazy interval} is a subset $\fI \subseteq \leftheapchildren(P)$ that is consecutive in the top-down ordering. $\fI$ is consecutive within $\leftheapchildren(P)$ but might be non-consecutive in $\heapchildren(P)$ if there is a right heap-child between two left heap-children in $\fI$. If this is the case, we call the lazy interval \emph{broken}. A \emph{right lazy interval} is defined analogously.
Although a lazy interval may be broken in the full top-down order of all heap-children, it is still consecutive in symmetric order: all left heap-children of $P$ lie left of $\bottom(P)$, all right heap-children lie right of it, and within each side the top-down order is monotone in symmetric order.
Therefore, as heap-descendants of a heavy path are a consecutive interval in the symmetric order, the union of heap-descendants over a lazy interval also always forms a consecutive interval in symmetric order.

Each heavy path $P$ maintains a partition of its heap-children into lazy intervals. Denote their set by $\intervals(P)$. For each $\fI \in \intervals(P)$, we call $P$ its \emph{owner}. Each non-root heavy path $P$ belongs to exactly one lazy interval, denoted by $\fI(P)$. We have $\fI(P) \in \intervals(\heapparent(P))$.

The top-most left and top-most right intervals on each heavy path are called \emph{growing} and they can be empty. All other lazy intervals are called \emph{shrinking}. We additionally maintain the invariant that there is no heap-child in a growing lazy interval that would appear lower in the top-down ordering of heap-children than some other heap-child from a shrinking interval.

The concept of lazy intervals is illustrated in \Cref{fig:lazy-intervals}.

\begin{definition}[Decomposition of gaps]
	The \emph{interval gap} of a shrinking lazy interval $\fI$ is defined as $\lazygap(\fI) = \min_{P \in \fI} \gap (P)$. The interval gap of a growing lazy interval is always $0$.
	The \emph{point gap} of a heavy path $P$ is defined as $\nonlazygap(P) = \gap(P) - \lazygap(\fI(P))$.
    For the root heavy path, the point gap is defined to be $0$.
    This decomposes each gap into two components.
    
\end{definition}

Throughout the analysis, we use a contracted version of the point gap, called the \emph{contracted point gap}.

\begin{definition}[Contracted point gap]
		We define the \emph{contracted point gap} of a heavy path $P$ as $\lceil1 + \log (1 + \nonlazygap(P))\rceil$.\footnote{All logarithms in this work are binary.}
\end{definition}

Furthermore, we define $\contractedgapbound = \max_t \lceil1 + \log (2 +\depth(T_t^*)) \rceil$, which is an upper bound on any contracted point gap at any time.

\subsubsection{Operations on lazy intervals}\label{sec:lazy-interval-operations}

After every splay and every rotation of the optimal tree, we are free to choose the precise structure of the intervals, as long as it is consistent with the definitions.

We do not, however, want to recreate the lazy intervals. The strategy of our proof relies on the fact that each zig-zig modifies the structure in a particular way and on bounding the zig-zigs by the total number of modifications to the structure.
For this, we need the changes of intervals to be controllable. If an interval has the same heap-children with the same heap-parent before and after operations, we simply say that this is the same lazy interval. If the lazy intervals are not the same, we argue that we can get from one setup to the other via the following operations.
Before presenting the allowed operations on lazy intervals, we define \emph{pairings}, the central operation: we later show that every zig-zig can be decomposed into two pairings.

\begin{definition}[Pairing]
	    A \emph{pairing} takes two right or two left heap-children with the same heap-parent and different ranks that appear consecutively in the top-down ordering of the heap-children. The heap-child $v$ with the smaller rank remains in its lazy interval, while the heap-child $w$ with the higher rank leaves its lazy interval and joins the left growing lazy interval of $v$ if it is to the left in the symmetric order; otherwise, it joins the right growing interval of $v$.
    $v$ is called the winner of the pairing, while $w$ is the loser.

We further define the following classification of pairings:
\begin{itemize}
    \item If $v$ and $w$ originally belonged to the same lazy interval, i.e., $\fI(v) = \fI(w)$,
          the pairing is called \emph{internal}.
          \begin{itemize}
	              \item If the contracted point gaps of $v$ and $w$ were the same, the internal pairing is called \emph{good}.
              \item Otherwise, the internal pairing is called \emph{bad}.
          \end{itemize}
	    \item If $v$ and $w$ originally belonged to different lazy intervals, i.e., $\fI(v) \neq \fI(w)$,
          the pairing is called \emph{boundary}.
          \begin{itemize}
	              \item If the loser's lazy interval is not broken and is large enough, i.e., $|\fI(w)| \geq \kappa \contractedgapbound\log(\contractedgapbound)$, where $\kappa$ is a constant that we set later, the boundary pairing is called \emph{unimportant}.
              \item Otherwise, the boundary pairing is called \emph{important}.
          \end{itemize}
\end{itemize}
\end{definition}

Note that the contracted point gap of the winner of any pairing never changes. The important property of internal pairings is that the contracted point gap of the loser does not increase. Furthermore, if the pairing is good, the loser's contracted point gap decreases, which will later allow us to bound good pairings by the total increase in contracted point gaps.

\begin{lemma}\label{lemma:internal-pairings-gap}
		Let $(v, w)$ be an internal pairing. Then the contracted point gap of $w$ does not increase. If the internal pairing $(v, w)$ is good, the contracted point gap of $w$ decreases by at least $1$.
\end{lemma}

\begin{proof}

	    Denote by $\fI$ the lazy interval in which the internal pairing happens, and let $u$ be the owner of $\fI$, hence the heap-parent of $v$ and $w$ before the pairing.
	    Denote by $x$ and $x'$ the point gaps of $w$ before and after the pairing, respectively, and let $y$ be the point gap of $v$, which does not change.
	    Then, since a growing lazy interval has zero interval gap, the following identity holds:
    \[
		 x'                                                               
		 = \rank(w) - \rank(v)                                                  
		 = \rank(w) - \rank(u) - \lazygap(\fI) - (\rank(v)  - \rank(u) - \lazygap(\fI))
		 = x -  y                     
	\]

	    It follows that
        \[
		 \lceil 1 + \log (1 + x') \rceil = \lceil 1 + \log (1 + x -  y)\rceil
	\]
		and since $y$ is non-negative,
            \[
		 \lceil 1 + \log (1 + x') \rceil = \lceil 1 + \log (1 + x -  y)\rceil \leq \lceil 1 + \log (1 + x)\rceil
	\]
    which shows that the contracted point gap of $w$ cannot increase.

		If $(v, w)$ is a good pairing, let
	$k=\left\lceil 1+\log(1+x)\right\rceil
		=\left\lceil 1+\log(1+y)\right\rceil$ their contracted point gap before the pairing.

	    We must have $k \geq 2$, as otherwise $x=y=0$, contradicting the distinct ranks of $w$ and $v$.
		Next, since \(x,y\in\mathbb Z\),
	\[
		2^{k-2}\leq x,y\leq 2^{k-1}-1
	\]
	\[
		1+x-y\leq 1+(2^{k-1}-1)-2^{k-2}
			=2^{k-2}
	\]
	\[
		1+\log(1+x-y)\leq k-1
	\]
	    By the identity at the beginning of the proof, $x' = x-y$, and
	\[
		1+\log(1+x')=1+\log(1+x-y)\leq k-1
	\]
	    Since $k - 1$ is an integer, the ceiling operation cannot increase the value above it.
	\[
		\lceil1+\log(1+x')\rceil\leq k-1.
	\]
	    This shows that for good internal pairings, the contracted point gap decreases from $k$ to at most $k-1$.
\end{proof}

With pairings in hand, we define the following allowed operations on lazy intervals.

\begin{itemize}
	    \item \textit{Pair-up} performs the effect of an internal pairing on a shrinking interval.
    \item \textit{Delete} removes a heap-child from a shrinking lazy interval.
    \item \textit{Insert} inserts a heap-child on top of a growing lazy interval. The inserted heap-child has to be higher in the top-down ordering of heap-children than all other heap-children from shrinking lazy intervals.
    \item \textit{Split} splits a shrinking lazy interval into two lazy intervals with consecutive heap-children.
	    \item \textit{Convert} converts a growing lazy interval into a shrinking interval and immediately creates a new growing interval on the same side.
    \item \textit{Transfer} changes the owner of a shrinking lazy interval (all heap-children in it change their heap-parent). All the heap-children in the transferred lazy interval have to end up higher on the heavy path than the heap-children that were there before -- it is possible to transfer multiple lazy intervals at the same time so that the transferred lazy intervals do not have to satisfy this condition relative to one another; and a transfer of one or multiple lazy intervals immediately converts the growing lazy intervals to shrinking and creates new growing lazy intervals above the transferred lazy intervals.
	    \item The last operation is \textit{heap-child-exchange}. The bottom-most vertex of a heap-child heavy path may change, thereby changing the identity of the heap-child while the upper part of the heavy path remains unchanged. In this case, we allow a heap-child-exchange, in which one heap-child is deleted from a lazy interval and another is simultaneously inserted into it. A valid heap-child-exchange requires that the exchanged heap-child be in a singleton shrinking lazy interval and that the light edge connecting the heap-child to the heap-parent heavy path have the same endpoints before and after the exchange.
\end{itemize}

\begin{table}[t]
\centering
\scriptsize
\setlength{\tabcolsep}{3pt}
\renewcommand{\arraystretch}{1.12}
\begin{tabular}{lccccccc}
\toprule
 & Pair-up & Delete & Insert & Split & Convert & Transfer & Heap-child-exchange \\
\midrule
Paid operation? 
  & no & no & yes & no & no & no & yes \\
Creates new lazy interval(s)? 
  & no & no & no & yes & yes & yes & no \\
Can increase contracted point gaps? 
  & no & no & by $\le\contractedgapbound$ & no & no & no & no \\
Can add heap-child below shrinking intervals? 
  & no & no & no & no & no & no & yes, but keeps structure \\
\bottomrule
\end{tabular}
\caption{Summary of allowed operations on lazy intervals and their properties.}
\label{tab:lazy-interval-operations}
\end{table}

Throughout the proof, we ensure that we do not modify lazy intervals differently than through these operations.
Furthermore, we ensure that no contracted point gap of any heap-child changes other than through these operations. For instance, if we need to change contracted point gaps of a few heap-children in some lazy interval due to their changes in gaps, we need to resolve it by the above operations, for instance, by splitting the lazy interval so that the heap-children for which the gaps change are in their own lazy interval so that the interval gap captures the difference.

With pair-ups, we will be able to do internal pairings. For boundary pairings, there is no explicit operation, hence we will have to cover their effects with a delete followed by an insert.
Whenever pair-up or delete removes the last heap-child of a shrinking interval, we simply delete that interval.

 Note that, except for heap-child-exchange, all operations ensure that new heap-children appear only above shrinking lazy intervals, so they do not interfere with the subpaths corresponding to shrinking intervals.
 Heap-child-exchange is the only operation that allows adding new heap-children into shrinking lazy intervals, but due to its requirements, it is guaranteed that it does not affect the structure of the heavy path.

 Next, as elements are always inserted only to growing lazy intervals whose interval gap is defined to be zero and lazy intervals are never merged, no operation can increase the point gap of any existing heap-child.
 In particular, splits and converts only decrease the contracted point gaps of heap-children in the lazy intervals and transfer changes the rank of the new heap-parent, but this gets absorbed in the interval gap, so the point gap and contracted point gap remain the same.
 This is true also for heap-child-exchange as it happens only on singleton shrinking lazy intervals, so the new heap-child has the lowest possible contracted point gap of $1$.

\paragraph{Operations on lazy intervals requiring attention.}
 There will be several operations and properties that we need to track.
 In \Cref{tab:lazy-interval-operations}, we summarize all these operations, including their properties that we care about.
 
 First, we need to bound the total number of inserts and heap-child-exchanges. This not only gives an upper bound on the number of boundary pairings but, as we later show, also bounds the number of pair-ups and hence the number of bad internal pairings. Furthermore, since inserts are the only operations that can increase a contracted point gap, we use their number to bound both the total increase in contracted point gaps and the total number of good internal pairings.
 We say that inserts and heap-child-exchanges are \emph{paid operations}, while all other operations are free.

 Second, we need to track the total number of lazy intervals created by converts and splits. This will be crucial for bounding the total number of inserts, as we will need to perform a certain number of inserts per lazy interval.

Finally, we need to keep track of every bend that we create on heavy paths, as we will use the total number of bends created to bound the total number of zig-zags.

\subsubsection{Adjusting lazy intervals to small changes in heavy paths}\label{sec:heavy-path-junction}

\usetikzlibrary{calc,arrows.meta,decorations.pathreplacing}

\begin{figure}[t]
\centering

\begin{tikzpicture}[
    vertex/.style={
        circle,
        draw,
        fill=white,
        minimum size=6mm,
        inner sep=0pt
    },
    childsquare/.style={
        rectangle,
        draw,
        fill=white,
        minimum size=5mm,
        inner sep=0pt
    },
    mainpath/.style={
        line width=2pt
    },
    subtreepath/.style={
        line width=1.4pt
    },
    dashededge/.style={
        dashed,
        line width=1.2pt
    },
    transferarrow/.style={
        -{Latex[length=3mm,width=2mm]},
        line width=1.4pt
    },
    lazy interval brace/.style={
        decorate,
        decoration={brace,mirror,amplitude=5pt,raise=2pt}
    },
    lazy interval label/.style={
        midway,
        below=8pt,
        xshift=-6pt
    },
    left interval brace/.style={
        decorate,
        decoration={brace,mirror,amplitude=5pt,raise=2pt}
    },
    right interval brace/.style={
        decorate,
        decoration={brace,amplitude=5pt,raise=2pt}
    },
    interval label left/.style={
        midway,
        left=8pt,
        yshift=-6pt
    },
    interval label right/.style={
        midway,
        right=8pt,
        yshift=-6pt
    }
]


\begin{scope}[shift={(-4.2,0)}]


\node[childsquare] (LvL) at (-1.4,-7.3) {};
\node[childsquare] (LvR) at ( 1.4,-6.7) {};

\node[vertex] (Lx2) at (0.0,-3.6) {};
\node[childsquare] (Lu4) at (2,-0.8) {};

\node[below=3pt] at (LvL.south) {$v'$};
\node[below=3pt] at (LvR.south) {$v$};
\node[below=3pt] at (Lx2.south) {$z$};
\node[below=3pt] at (Lu4.south) {$u$};


\draw[mainpath]
    (LvL.north)
    -- ++(0,0.55)
    -- ++(0.75,0.8)
    -- ++(-1.25,0.4)
    -- (Lx2.south west);


\coordinate (LrbranchA) at ($(Lx2)+(1.0,-0.7)$);

\draw[dashededge]
    (Lx2.south east) -- (LrbranchA);

\draw[mainpath]
    (LrbranchA)
    -- ++(1.2,-0.9)
    -- ($(LvR.north)+(0,0.55)$)
    -- (LvR.north);


\coordinate (LubranchA) at (-3.5,-0.2);
\coordinate (LubranchB) at (0.2,0.675);

\draw[mainpath]
    (Lx2.north east)
    -- ++(1.2,1.7)
    -- ++(-4.2,0.8)
    -- (LubranchA);

\draw[dashededge]
    (LubranchA) -- (LubranchB);

\draw[mainpath]
    (Lu4.north)
    -- ++(0,0.55)
    -- (LubranchB)
    -- ++(-2.2, 0.825);


\draw[lazy interval brace]
    (-1.85,1.30)
    --
    (1.75,0.05)
    node[lazy interval label] {$I_u$};


\draw[left interval brace]
    (-3.5,-0.5)
    --
    (-1.8,-6.8)
    node[interval label left] {$I_L$};

\draw[right interval brace]
    (1.6,-2)
    --
    (-1,-6.8)
    node[interval label right] {$I_R$};

\end{scope}


\begin{scope}[shift={(4.2,0)}]


\node[childsquare] (RvL) at (-1.4,-7.3) {};
\node[childsquare] (RvR) at ( 1.4,-6.7) {};

\node[vertex] (Rx2) at (0.0,-3.6) {};
\node[childsquare] (Ru4) at (2,-0.8) {};

\node[below=3pt] at (RvL.south) {$v'$};
\node[below=3pt] at (RvR.south) {$v$};
\node[below=3pt] at (Rx2.south) {$z$};
\node[below=3pt] at (Ru4.south) {$u$};


\coordinate (RlbranchA) at ($(Rx2)+(-0.95,-0.75)$);

\draw[dashededge]
    (Rx2.south west) -- (RlbranchA);

\draw[mainpath]
    (RvL.north)
    -- ++(0,0.55)
    -- ++(0.75,0.8)
    -- ++(-1.25,0.4)
    -- (RlbranchA);


\coordinate (RrbranchA) at ($(Rx2)+(1.0,-0.7)$);

\draw[mainpath]
    (Rx2.south east)
    --
    (RrbranchA)
    -- ++(1.2,-0.9)
    -- ($(RvR.north)+(0,0.55)$)
    -- (RvR.north);


\coordinate (RubranchA) at (-3.5,-0.2);
\coordinate (RubranchB) at (0.2,0.675);

\draw[mainpath]
    (Rx2.north east)
    -- ++(1.2,1.7)
    -- ++(-4.2,0.8)
    -- (RubranchA);

\draw[dashededge]
    (RubranchA) -- (RubranchB);

\draw[mainpath]
    (Ru4.north)
    -- ++(0,0.55)
    -- (RubranchB)
    -- ++(-2.2, 0.825);


\draw[lazy interval brace]
    (-1.85,1.30)
    --
    (-0.75,0.92)
    node[lazy interval label] {$I''_u$};

\draw[lazy interval brace]
    (-0.60,0.87)
    --
    (0.55,0.48)
    node[lazy interval label] {$I'_u$};

\draw[lazy interval brace]
    (0.70,0.43)
    --
    (1.75,0.05)
    node[lazy interval label] {$I_u$};


\draw[left interval brace]
    (-3.5,-0.5)
    --
    (-2.75,-3.25)
    node[interval label left] {$I_L'$};

\draw[left interval brace]
    (-2.50,-4.05)
    --
    (-1.8,-6.8)
    node[interval label left] {$I_L$};

\draw[right interval brace]
    (1.6,-2)
    --
    (0.95,-3.25)
    node[interval label right] {$I_R'$};

\draw[right interval brace]
    (0.25,-4.55)
    --
    (-1,-6.8)
    node[interval label right] {$I_R$};

\end{scope}


\draw[transferarrow]
    ($(LvR.east)+(0.9,3.1)$)
    --
    ($(RubranchA)+(-0.3,-3.4)$);

\end{tikzpicture}

\caption{All the lazy interval splits associated with a heavy path junction at node $z$.
The lazy interval $I_u$ containing $v'$ is split so that $v'$ is in a singleton lazy interval $I_u'$, making heap-child-exchange possible. $I_L$ and $I_R$ -- the left and right intervals of $v'$ that in top-down order contain the internal node $z$ are split at the height of $z$. The upper halves $I_L'$ and $I_R'$ along with all other lazy intervals of $v'$ above $z$ are transferred to $v$. The other halves along with the remaining lazy intervals stay with $v'$. Except for these splits, the three heap-children that descend from $z$, including $v'$ and $v$, are resolved by inserts, deletes and the heap-child-exchange of $v'$ for $v$.}
\label{fig:heavy-path-junction}

\end{figure}
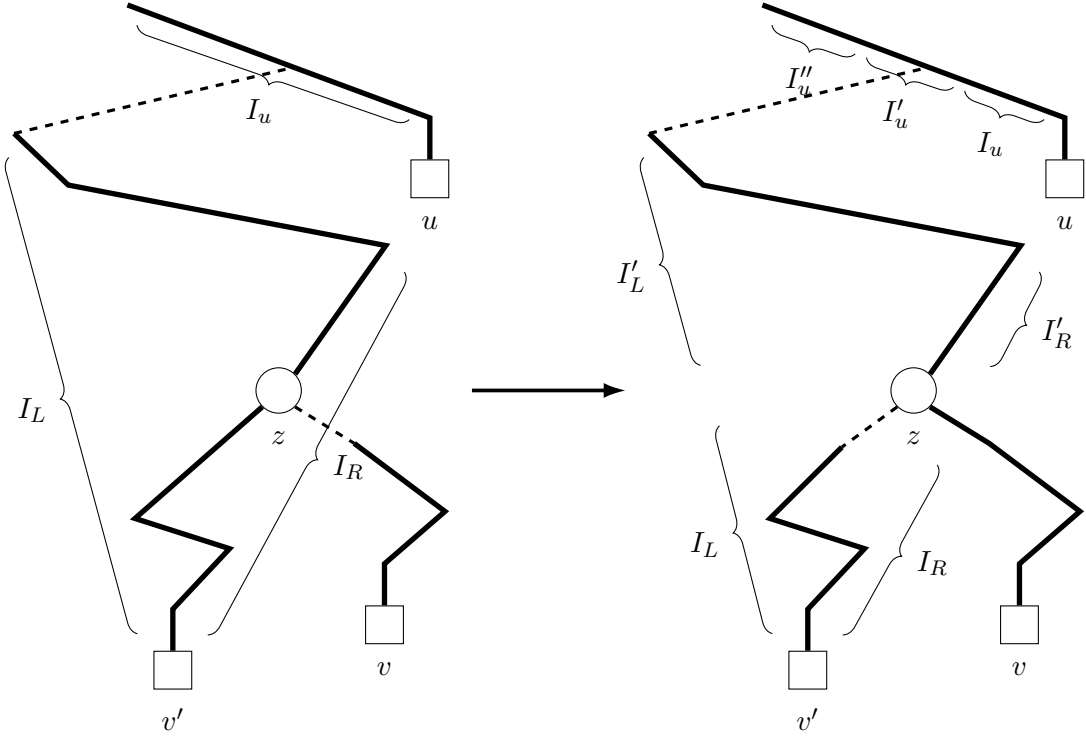

When the structure of the heavy paths needs to be modified only in a few places, we can describe the process in terms of the differences between the heavy-path structures before and after the modification. At certain internal nodes, we need to split the old heavy path and concatenate two pieces of the new heavy path. We call such an event a \emph{heavy path junction}.
Formally, a junction occurs at an internal node $z$ whose outgoing heavy edge changes, equivalently, when the minimum-rank external node in the subtree of $z$ changes from $v'$ to $v$; see \Cref{fig:heavy-path-junction}.

A junction might require modifications of lazy intervals so that it is true that a lazy interval contains only heap-children with the same heap-parent and that each heavy path has two growing lazy intervals.
We resolve this by the following sequence of operations on lazy intervals.

\begin{enumerate}
    \item The growing lazy intervals of $v'$ are converted.
	    \item The new lowest-rank node $v$ in the subtree of $z$ is deleted from the lazy intervals of the former lowest-rank node $v'$ in the subtree of $z$. The other heap-child that was connected via a light edge to $z$ is also deleted.
    \item We split the lazy intervals with owner $u$, the former heap-parent of $v'$, that contained $v'$ by two splits into a singleton lazy interval, converting it to shrinking if needed.
	    \item We do a heap-child-exchange from $v'$ to $v$. Since the lazy interval containing $v'$ is a singleton after the previous step and the top nodes of the heavy paths of $v'$ before the junction and of $v$ after it are the same, this satisfies the requirements of heap-child-exchange.
	    If $v'$ was the root heavy path, this step is skipped, as $v$ will now be the root heavy path and thus correctly belongs to no lazy interval.
	    \item The growing lazy intervals of $v$ are converted.
	    \item The node $v'$ and the other heap-child that was connected via a light edge to $z$ are inserted into the appropriate growing intervals of $v$ in the appropriate order. After each insertion, the growing lazy interval is converted into a shrinking interval, ensuring that it is a singleton.
	    \item The at most two lazy intervals of $v'$ that contained heap-children both above and below $z$ in the top-down order are split at $z$.
    \item All lazy intervals of $v'$ above $z$ are transferred to $v$.
\end{enumerate}

This procedure correctly resolves the relationships between $v$, $v'$, and $u$, and transfers all heap-children above $z$ from $v'$ to $v$ as this part of the heavy path is now part of the heavy path of $v$.
Updates on lazy intervals related to a heavy path junction are depicted in \Cref{fig:heavy-path-junction}.

There is one more important aspect to this junction resolution and that is that the former heap-parent $v'$ ends up in a singleton shrinking lazy interval.
This is important since when we need to apply these modifications of lazy intervals we will not be able to atomically ensure that at the same time we perform these junctions and change the ranks. Due to this, after the junction, the gap of $v'$ could temporarily be negative. However, its contracted point gap is always $1$ as the interval gap resolves this, so this will be correct.

\begin{lemma}\label{lemma:heavy-path-junction}

	    A heavy-path junction creates $O(1)$ bends; the lazy intervals can be adjusted using $O(1)$ paid operations and creating $O(1)$ new lazy intervals.
    
\end{lemma}

\begin{proof}
	    New bends can be created only at the node $z$ and at the node immediately below $z$ on the path to the new minimum $v$.
    At every other internal node, there is no new heavy edge going to or from it.
	    At most $4$ lazy intervals are split and $O(1)$ new growing intervals are created; this totals $O(1)$ new lazy intervals.
	    We perform one heap-child-exchange and $O(1)$ inserts, for a total of $O(1)$ paid operations.
\end{proof}

\subsection{Handling rotations of the nearly optimal tree}

Whenever the nearly optimal tree $T^*$ rotates, the structure of heavy paths changes, which in turn necessitates a modification of lazy intervals.
We first show that, after a rotation of the optimal tree near its root, the structure of the heavy paths changes in only a constant number of places, creating $O(1)$ bends.
Next, we show that we can adjust the lazy intervals to these changes using only $O(1)$ paid operations and creating $O(1)$ new lazy intervals.

\subsubsection{Changes in heavy paths due to rotations in the nearly optimal tree}

In this part, assume that a node $x$ in $T^*$ is rotated. 
Denote by $y$ the parent of $x$. Additionally, denote by $p^L$ the lowest ancestor of $y$ with $y$ in its right subtree, and analogously by $p^R$ the lowest ancestor with $y$ in its left subtree.

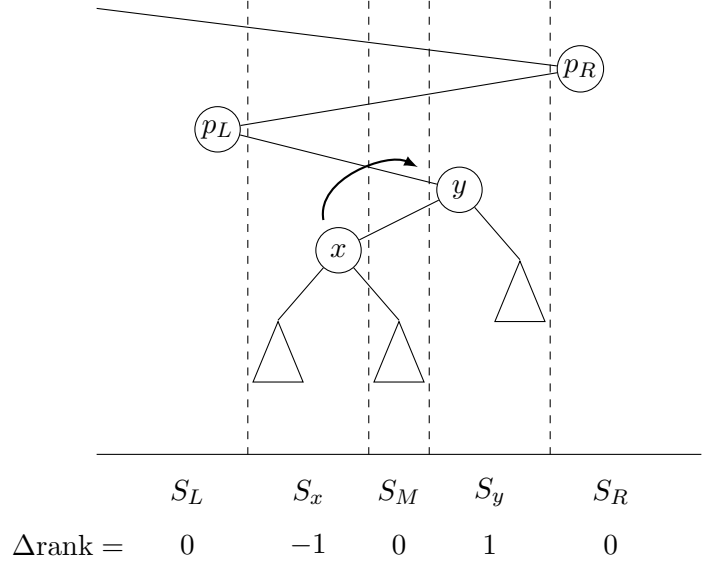
\begin{figure}[t]
\centering
\begin{tikzpicture}[
  level distance=9mm,
  sibling distance=12mm,
  subtree/.style={
    draw,
    isosceles triangle,
    shape border rotate=90,
    minimum height=8mm,
    minimum width=6mm,
    yshift=-7.5mm
  },
  subtree child/.style={
    edge from parent path={
      (\tikzparentnode) -- (\tikzchildnode.north)
    }
  },
  internal/.style={circle, draw, minimum size=6mm, inner sep=1.2pt},
  >=latex
]

\coordinate (SL) at (-1.8,-3.5);
\coordinate (Sx) at (-0.2,-3.5);
\coordinate (SM) at ( 1.0,-3.5);
\coordinate (Sy) at ( 2.2,-3.5);
\coordinate (SR) at ( 3.8,-3.5);

\node[internal] (pL) at (-1.4, 1.3) {$p_L$};
\node[internal] (pR) at ( 3.4, 2.1) {$p_R$};

\node[internal] (x) at (0.2,-0.3) {$x$};
\node[internal]   (y) at (1.8, 0.5) {$y$};

\node[subtree] (Ax) at (-0.6,-1.1) {};
\node[subtree] (Bx) at ( 1.0,-1.1) {};
\node[subtree] (Cy) at ( 2.6,-0.3) {};

\draw (-3, 2.9) -- (pR);
\draw (pL) -- (pR);
\draw (pL) -- (y);

\draw (x) -- (Ax.north);
\draw (x) -- (Bx.north);
\draw (x) -- (y);
\draw (y) -- (Cy.north);

\draw[->, thick] (0,0.1) .. controls (-0.1,0.65) and (0.85,1.0) .. (1.25,0.8);

\draw[dashed] (-1, 3) -- (-1,-3);
\draw[dashed] (0.6, 3)  -- (0.6,-3);
\draw[dashed] (1.4, 3)  -- (1.4,-3);
\draw[dashed] (3, 3) -- (3,-3);

\draw (-3,-3) -- (5,-3);

\node at (SL) {$S_L$};
\node at (Sx) {$S_x$};
\node at (SM) {$S_M$};
\node at (Sy) {$S_y$};
\node at (SR) {$S_R$};

\node[left] at (-2.5,-4.2) {$\Delta\rank =$};

\node at (SL |- 0,-4.2) {$0$};
\node at (Sx |- 0,-4.2) {$-1$};
\node at (SM |- 0,-4.2) {$0$};
\node at (Sy |- 0,-4.2) {$1$};
\node at (SR |- 0,-4.2) {$0$};

\end{tikzpicture}
\caption{Changes in ranks due to a left-edge rotation of $x$ in $T^\text{OPT}$.}
\label{fig:rank-changes-to-rotation}
\end{figure}

We define the following \emph{segments} of nodes that are consecutive in the symmetric order of nodes; see also \Cref{fig:rank-changes-to-rotation}.
\begin{itemize}
	    \item $S_L$ consists of all nodes that are not in the subtree of $y$ and are left of $y$ in the symmetric order. This corresponds to the segment $[\perp_0, p^L]$ and may be empty if $p^L$ does not exist. Nodes in $S_L$ do not change their rank during the rotation.
	    \item Analogously, $S_R$ consists of all nodes that are not in the subtree of $y$ and are right of $y$ in the symmetric order. This corresponds to the segment $[p^R, \perp_n]$ and may be empty if $p^R$ does not exist. As with $S_L$, nodes in $S_R$ do not change their rank.
    \item $S_y$ consists of $y$ and the subtree of $y$ not containing $x$. If this is a rotation on a left edge, these are $[y, p^R)$, otherwise $(p^L, y]$.
    In case $p^R$ or $p^L$ do not exist, these are instead $[y, \perp_n]$, respectively $[\perp_0, y]$.
    These nodes increase their rank by $1$.
    \item $S_x$ consists of $x$ and the subtree of $x$ that remains a subtree of $x$ after the rotation. If this is a rotation on a left edge, these are $(p^L, x]$, otherwise $[x, p^R)$.
    Again, in case $p^L$ or $p^R$ do not exist, these are instead $[\perp_0, x]$, respectively $[x, \perp_n]$.
    These nodes decrease their rank by $1$.
	    \item $S_M$ consists of the nodes in the other subtree of $x$.
    If this is a rotation on a left edge, these are $(x, y)$, otherwise $(y, x)$. Their rank is not modified.
\end{itemize}

A pair of consecutive vertices belonging to different segments is called a \emph{border}. The node of a border with the lower rank is called a \emph{border node}.

\begin{lemma}\label{lemma:ranks-of-border-nodes}
    Every border node has rank $O(1)$.
\end{lemma}

\begin{proof}
	    There are at most $4$ borders, and each of them contains at least one node from $\{x, y, p^L, p^R\}$.
	    The node $x$ has rank $O(1)$ by \Cref{lemma:opt-constraints}, as rotations occur only at constant depth from the root. Moreover, $y$, $p^L$, and $p^R$ have smaller ranks because they are ancestors of $x$ in the optimal tree. Thus, all these nodes have rank $O(1)$.
    Since every border contains a node with rank $O(1)$, the border node corresponding to that border has to have rank $O(1)$.
\end{proof}


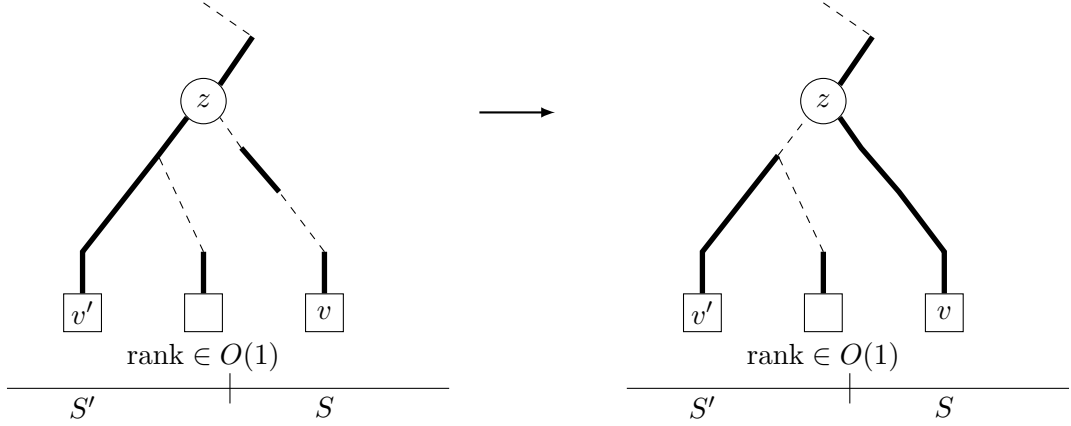
\begin{figure}[t]
\centering
\begin{tikzpicture}[
  level distance=9mm,
  sibling distance=12mm,
  subtree/.style={
    draw,
    isosceles triangle,
    shape border rotate=90,
    minimum height=8mm,
    minimum width=6mm,
    yshift=-7.5mm
  },
  subtree child/.style={
    edge from parent path={
      (\tikzparentnode) -- (\tikzchildnode.north)
    }
  },
    mainpath/.style={
        line width=2pt
    },
  internal/.style={circle, draw, inner sep=1.2pt, minimum width=6mm},
  value/.style={rectangle, draw, inner sep=2pt, minimum size=5mm},
  >=latex
]

\begin{scope}[xshift=-4.1cm]

\node[internal] (z) at (0,0) {$z$};
\node[value] (u) at (-1.6,-2.8) {$v'$};
\node[value] (r) at ( 0.0,-2.8) {};
\node[value] (v) at ( 1.6,-2.8) {$v$};

\draw[mainpath] (z.north east) -- (0.65,0.85);
\draw[dashed] (0, 1.3) -- (0.65,0.85);

\draw[mainpath] (u.north) -- ++(0, 0.55) -- (-0.6, 
-0.72) -- (z.south west);

\draw[mainpath] (r.north) -- ++(0, 0.55);
\draw[dashed] ($(r.north) + (0, 0.55)$)-- (-0.6, 
-0.72) ;

\draw[mainpath] (v.north) -- ++(0, 0.55);
\draw[dashed] ($(v.north) + (0, 0.55)$) -- (1,-1.2) ;
\draw[mainpath] (0.5, -0.62) -- (1,-1.2) ;
\draw[dashed] (0.5, -0.62) -- (z.south east) ;

\draw (-2.6,-3.80) -- (3.25,-3.80);

\draw (0.35,-3.6) -- (0.35,-4);
\node at (-1.6,-4.05) {$S'$};
\node at ( 1.60,-4.05) {$S$};

\node[below] at (r.south) {$\rank \in O(1)$};

\end{scope}

\draw[->, thick] (-0.45,-0.15) -- (0.55,-0.15);

\begin{scope}[xshift=4.1cm]

\node[internal] (z) at (0,0) {$z$};
\node[value] (u) at (-1.6,-2.8) {$v'$};
\node[value] (r) at ( 0.0,-2.8) {};
\node[value] (v) at ( 1.6,-2.8) {$v$};

\draw[mainpath] (z.north east) -- (0.65,0.85);
\draw[dashed] (0, 1.3) -- (0.65,0.85);

\draw[mainpath] (u.north) -- ++(0, 0.55) -- (-0.6, 
-0.72);
\draw[dashed] (-0.6, -0.72) -- (z.south west);

\draw[mainpath] (r.north) -- ++(0, 0.55);
\draw[dashed] ($(r.north) + (0, 0.55)$)-- (-0.6, 
-0.72) ;

\draw[mainpath] (v.north) -- ++(0, 0.55) -- (1,-1.2) -- (0.5, -0.62) -- (z.south east) ;

\draw (-2.6,-3.80) -- (3.25,-3.80);

\draw (0.35,-3.6) -- (0.35,-4);
\node at (-1.6,-4.05) {$S'$};
\node at ( 1.60,-4.05) {$S$};

\node[below] at (r.south) {$\rank \in O(1)$};
\end{scope}

\end{tikzpicture}
\caption{A heavy path junction at node $z$ due to optimal tree rotations. $v'$ is the node containing $z$ in its heavy path before and $v$ is the node after. Since their relative ranks changed, they must come from different segments $S'$ and $S$. Before rotation $v'$ was a heap-ancestor of $v$ and a heap-ancestor of a border node with rank $O(1)$; after $v$ is a heap-ancestor of both $v'$ and the border node.}
\label{fig:junction-due-to-rotation}
\end{figure}

\begin{lemma}\label{lemma:heavy-paths-rotation}
	    There are $O(1)$ heavy-path junctions due to a rotation of the nearly optimal tree $T^*$. As a consequence, each such rotation creates $O(1)$ bends.
\end{lemma}

\begin{proof}
    We count the number of internal nodes $z$ where a junction happens. This happens whenever there was a node $v'$ containing $z$ in its heavy path before the rotation and a different node $v$ containing $z$ in its heavy path after the rotation.
	    Each such pair $(v',v)$ uniquely determines the internal node $z$ where it can cause a junction: the point where their paths to the root diverge, namely, the lowest common ancestor of $v'$ and $v$.
	    Above the lowest common ancestor, the heavy path would lead to the same node, so this pair cannot cause a junction there. Below it, one of the two nodes is not in the subtree, so the pair cannot cause a junction there either.

    As $v$ is in the subtree of $z$, prior to the rotation $v'$ was a heap-ancestor of $v$; and after the rotation,
    $v$ is a heap-ancestor of $v'$. In particular, this means that $v'$ and $v$ must have been in different segments as otherwise their ranks would change in the same way. Therefore the subtree of $z$ contains a border and hence at least one border node. Therefore, $v'$ must have been a heap-ancestor of the heavy path of a border node before the rotation and $v$ a heap-ancestor after the rotation. The situation is depicted in \Cref{fig:junction-due-to-rotation}.

	    There are at most $4$ different border nodes, and each has rank $O(1)$ before the rotation by \Cref{lemma:ranks-of-border-nodes}.
    Since a node of rank $R$ can have at most $R+1$ heap-ancestors,
    there are only $O(1)$ possible nodes $v'$ whose heavy paths were heap-ancestors of a border node before the rotation.
    Similarly, after the rotation each border node has rank $O(1)$ as ranks change only by a constant during the rotation. Therefore, there are only $O(1)$ possible nodes $v$ whose heavy paths are heap-ancestors of a border node after the rotation.
	    Hence, there are $O(1) \times O(1) = O(1)$ pairs of vertices $(v',v)$ and thus only $O(1)$ heavy-path junctions.
    The two $O(1)$ factors come from the fact that the nearly optimal tree $T^*$ does rotations only at constant depth.
	    Each junction creates $O(1)$ bends by \Cref{lemma:heavy-path-junction}.
\end{proof}

\subsubsection{Adjusting lazy intervals to rotations}

Since the structure of the heavy paths changes, the lazy intervals must be updated. However, there is one more issue to settle. We cannot directly modify the contracted point gaps of heap-children. Instead, we must either ensure that the contracted point gaps do not change or absorb the changes through operations on lazy intervals, so that the contracted point gaps remain unchanged once these operations and the rank updates are complete.

\begin{lemma}\label{lemma:bounding-border-intervals}
    There are $O(1)$ lazy intervals containing heap-children from at least two different segments.
\end{lemma}

\begin{proof}
    We know that all heap-descendants in any lazy interval are consecutive in the symmetric order.
    This implies that if a lazy interval contains heap-children from two different segments, then it contains a border node as a heap-descendant.
	    In particular, the owner of the lazy interval is a heap-ancestor of the border node.
    
	    There are $4$ border nodes, each having $O(1)$ heap-ancestors, and for each such heap-ancestor, there is at most one lazy interval containing the border node as a heap-descendant.
    Therefore there are $O(1)$ such lazy intervals.
\end{proof}

Our strategy for updating the lazy intervals to handle a rotation is as follows. We first perform the necessary updates to the lazy intervals so that the structure of the heavy paths can change. We then perform certain splits and converts. Finally, we observe that, because of how we split the intervals, no heap-child changes its contracted point gap due to the rank changes.

\begin{lemma}\label{lemma:rotation-lazy-intervals}
    Whenever the nearly optimal tree $T^*$ rotates, the lazy intervals can be updated by $O(1)$ paid operations and by creating $O(1)$ new lazy intervals.
\end{lemma}

\begin{proof}
	    First, we adjust the lazy intervals to the changes in the heavy paths. By \Cref{lemma:heavy-paths-rotation}, changes occur in only $O(1)$ places.
	    The heavy-path junctions can be handled using $O(1)$ paid operations and creating $O(1)$ new lazy intervals by \Cref{lemma:heavy-path-junction}.
    Note that we first perform the changes due to junctions and then apply the rank changes. For a short while the former heap-parents could have negative gaps, but since they are in shrinking singleton lazy intervals, their contracted point gaps are always $1$, so this is correct.

	    If a shrinking lazy interval lies entirely within a segment, the gaps of all heap-children in it change by the same amount. Therefore, the interval gap also changes by exactly this amount, and the point gaps remain unchanged.
    
	    By \Cref{lemma:bounding-border-intervals}, there are only $O(1)$ intervals not fully contained in a single segment. We can split them into $O(1)$ parts, each fully contained in one segment. We do this using free splits, preceded by converts if the lazy interval being split was growing. This creates only $O(1)$ new lazy intervals.

    Next, we need to bound non-empty growing lazy intervals that are fully within a single segment. If the owner of the lazy interval is in the same segment, then all ranks change by the same amount, so all the gaps and also contracted point gaps of heap-children within the lazy interval remain the same.
	    The problematic intervals are those for which all heap-children of a growing lazy interval are in the same segment but the owner is in a different segment. Since the interval gap of a growing lazy interval is defined to be zero, it does not capture this shift. We resolve the issue by converting the interval into a shrinking interval, which absorbs the rank update so that the contracted point gaps do not change.
    However, we need to argue that this cannot happen too often.
    Consider a heap-child $v$ in the lazy interval and its heap-parent $u$. Since they are in different segments, it means that
	    there is a border node between $v$ and $u$ in the symmetric order; in particular, $u$ is a heap-ancestor of this border node.
	    Each of the $O(1)$ border nodes has only $O(1)$ such heap-ancestors, and each heap-ancestor has only $2$ growing lazy intervals. Therefore, this creates only $O(1)$ new lazy intervals.

	    In total, by creating $O(1)$ new lazy intervals and using $O(1)$ paid operations, we ensure that no contracted point gap changes due to the rank update, except for decreases already caused by free interval operations.
\end{proof}
\subsection{Handling splays}

To be able to argue about the cost of splaying, we first observe how the structure of the splay tree changes when we splay some node $x$ to the root.
The main idea is to use these structural changes to bound zig-zigs and zig-zags. Roughly speaking, we bound the total number of changes needed to maintain the lazy intervals and show that each zig-zig makes at least one such change. For zig-zags, we observe that each removes a bend, which allows us to relate their number to the zig-zigs that create bends.

We analyze these changes in three phases. First, we observe how splaying changes the structure of heavy paths. We then show how to adjust the lazy intervals while creating only a few new ones. Finally, we show how to perform these changes on the interface with only a few operations; the crucial step is to decompose the effects of zig-zigs into pairings.

\subsubsection{Structural changes to heavy paths due to splaying}

Because the accessed element is at the root of the nearly optimal tree $T^*$ before the splay tree changes, it has rank $0$ and is always the bottom-most node on the root heavy path.

When the parent internal node of $x$, denoted by $x^\circ$, is splayed, $x$ and $x^\circ$ are progressively moved towards the root by zig-zigs and zig-zags, possibly ending with a zig. Each zig-zig and zig-zag removes two internal nodes from the heavy path of $x$, while a zig removes one; we call these nodes \emph{hit} by the corresponding operation. Except for the hit nodes, the structure of the heavy path of $x$ remains exactly the same as before the operation.
Additionally, we say that a heap-child is \emph{hit} by an operation if the light edge connecting it to its heap-parent is incident to a hit internal node.

Our aim is to determine what happens during a splay to heavy paths other than the root heavy path.
Towards that goal, let $x^-$ denote the node with the lowest rank in the subtree of the left child $x^L$ of $x^\circ$ before the splay. Since $\perp_{x-1}$ is in this subtree, $\rank(x^-) \leq \rank(\perp_{x-1})$. By \Cref{cor:depth-by-riblessness}, the rank of $\perp_{x-1}$ is $O(1)$, and hence so is the rank of $x^-$.
Analogously, $x^+$ is defined as the node with the lowest rank in the subtree of $x^R$, and its rank is also $O(1)$ because the rank of $\perp_{x}$ is $O(1)$.

To be able to reason about the effects of a splay on the structure of heavy paths, we do the following trick.
We imagine that before the splay, the ranks of both $x^-$ and $x^+$ lie between $x$'s rank of $0$ and the ranks of all other nodes; think of them as having rank $0.5$. Before the splay, this does not affect the structure of the heavy paths: both $x^+$ and $x^-$ are heap-children of the lower-rank node $x$, and they are already the least-rank nodes within their respective subtrees.
Although general artificial rank changes could cause a subtree to contain two different least-rank nodes, this particular change is safe: after the splay, $x^-$ and $x^+$ are the least-rank vertices in the left and right subtrees, respectively, and subtrees containing neither of them are unaffected.

In the following lemma, we analyze how the structure of heavy paths changes after a splay under these assumptions. Later, we show that omitting these assumptions requires only a small adjustment to the analysis.

\begin{lemma}\label{lemma:heavy-paths-after-splay}
	    Assume that $x^-$ and $x^+$ are the least-rank nodes to the left and right of $x$, respectively, in symmetric order.
	    After a splay of $x^\circ$ with value node $x$, the structure of the heavy paths changes as follows.
    \begin{itemize}
        \item After the splay, $x^-$ is the only left heap-child of $x$ and $x^+$ is the only right heap-child of $x$. Furthermore, $x$ contains only a single internal node, $x^\circ$, in its heavy path.
        \item All internal nodes except for $x^\circ$ on the former heavy path of $x$ are hit by some zig-zig/zig-zag/zig and so are all former heap-children of $x$ except for $x^-$ and $x^+$. No internal node and no heap-child are hit by more than one operation.
	        \item For every zig-zig, among its four hit heap-children, the bottom-most one in the top-down order is skipped and the remaining three form a consecutive triplet. Denote by $w$ the heap-child in the triplet with the least rank and call it the \emph{winner}. The other two heap-children in the triplet, \emph{the losers}, become heap-children of $w$, and $w$ acquires the upper internal node hit by that zig-zig into its heavy path.
        \item 
	        All hit right heap-children that are not losers of a triplet---that is, heap-children touched by zig-zags or zigs, heap-children skipped by a zig-zig, and winners of a triplet---become right heap-children of $x^+$. All hit left heap-children that are not losers become left heap-children of $x^-$.
	        All internal nodes that are either hit by a zig-zag or zig, or are the lower node hit by a zig-zig, become internal nodes of the heavy path of either $x^-$ or $x^+$.
    \end{itemize}
\end{lemma}

\begin{proof}
    Once $x$ is splayed, it still has rank $0$ so it is the root heavy path and $x^\circ$ is the root and the only internal node on the heavy path.
	    The node $x^-$ is assumed to be the lowest-rank node in the left subtree of the root, so it is still a left heap-child of $x$. Since there is only one internal node on the heavy path of $x$, it is the only left heap-child. Similarly, $x^+$ is the lowest-rank node in the right subtree of the root, and for the same reason, it is the only right heap-child of $x$.

	    For the second property, every operation hits two internal nodes, which are removed from the heavy path of $x$. Therefore, no internal node can be hit twice, and neither can any heap-child. Since, after the splay, the heavy path of $x$ has no internal node other than $x^\circ$, every other internal node on its former heavy path must have been hit by some operation.
    
    Next, consider a particular zig-zig and the structure preceding this zig-zig. At this point, $x$ is still the bottom-most node of the root heavy path with everything above as before the splay, whereas everything below is exactly as after the splay, since rotations only affect the local neighborhood.
	     All three cases distinguished below are depicted in \Cref{fig:splay-heavy-paths:zig-zig}.
     
	    Let $y^\circ$ be the parent of $x^\circ$ and $z^\circ$
	    the parent of $y^\circ$---these are the two internal nodes hit by this zig-zig. Let $y$ and $z$ denote the value children of $y^\circ$ and $z^\circ$, respectively. First, we consider the case where $x^\circ$ and $y^\circ$ are both left children.
    Let $y^+$ be the lowest-rank node in the right subtree of $y^\circ$ and similarly define $z^+$.
    Prior to the zig-zig, $x$ had four heap-children that were connected via light edges to either $y^\circ$ or $z^\circ$ -- these were $y$, $y^+$, $z$, and $z^+$. These are exactly the heap-children hit by the zig-zig and they appear consecutively in the top-down ordering of heap-children.
    After the zig-zig, $y^\circ$ and $z^\circ$ are no longer on the heavy path of $x$, with $y^\circ$ being a right child of $x^\circ$. Since $x^+$ is the least-rank node right from $x$, $y^\circ$ is on the heavy path of $x^+$ and hence all $y$, $y^+$, $z$, and $z^+$ are no longer heap-children of $x$. $y$ becomes a heap-child of $x^+$ since it is a child of an internal node on its heavy path. Similarly, the node that contains in its heavy path $z^\circ$ becomes a heap-child of $x^+$ -- this is the one from $y^+$, $z$, and $z^+$ with the least rank. The other two are connected via a light edge to $z^\circ$, so they become heap-children of the least-rank one from $y^+$, $z$, and $z^+$.  Hence, $y^+$, $z$, and $z^+$ are the consecutive triplet considered in the statement and $y$ along with the least-rank node from the triplet become heap-children of $x^+$. No other heap-children are changed, since these are the only nodes hit by the zig-zig. The situation when $x^\circ$ and $y^\circ$ are right children is exactly the same, just with $x^-$ instead of $x^+$.

	    To see the last part of the claim, we first need to show that all other operations move former heap-children of $x$ to become heap-children of $x^-$ or $x^+$. For zig-zigs, we have already shown that the remaining heap-children become heap-children of either $x^-$ or $x^+$, so we need to show this for zigs and zig-zags.

    For a simple rotation of $x$, first consider that $x^\circ$ is a left child.
	    As above, $y^\circ$ becomes part of the heavy path of $x^+$, so both $y$ and $y^+$ become heap-children of $x^+$. The case where $x^\circ$ is a right child is analogous, with $x^-$ in place of $x^+$.
    Zigs are one such rotation and zig-zags are two such rotations one after another. Therefore the heap-children and internal nodes hit by them behave as described above. The effect of the rotation of $x$ on the heavy paths is illustrated in \Cref{fig:splay-heavy-paths:rotation}.

	    This shows that every remaining former heap-child of $x$ becomes a heap-child of either $x^-$ or $x^+$. All former right heap-children of $x$ are in the right subtree, so they must become heap-children of $x^+$. Furthermore, they lie to the right of $x^+$ in the symmetric order, as $x^+$ is the bottom-most and hence left-most right heap-child of $x$; thus, they all become right heap-children of $x^+$. By the same reasoning, all left heap-children become left heap-children of $x^-$.
\end{proof}

See also \Cref{fig:splay} for an illustration of how the heavy paths change after splaying.

However, as mentioned above, the assumption that $x^-$ and $x^+$ have the next-lowest ranks need not hold. We therefore need to consider what changes when we adjust their ranks to their true values.
We now want to argue that this requires only $O(1)$ heavy path junctions.
First, consider $x^-$. Prior to its rank adjustment, it was the bottom of the root heavy path of the left subtree of $x^\circ$. In the left subtree of $x^\circ$, increasing the rank of $x^-$ can create junctions only on this heavy path, as no other internal node contains $x^-$ and no other external node's rank changes. After the update to its correct rank, its rank is $O(1)$ by the riblessness of the optimal tree, as proved in \Cref{cor:depth-by-riblessness}.
In particular there can be at most $O(1)$ light edges on the path from $x^-$ to the root as every light edge decreases the rank by at least $1$. 
Since each heavy-path junction on the heavy path of $x^-$ creates one such light edge, there can be only $O(1)$ heavy-path junctions.
The same holds for $x^+$ as the situation is identical.

Since we now know how the structure of heavy paths changes, we can bound the number of bends created and moreover show that each zig-zag removes at least one bend.

\begin{lemma}\label{lemma:bends-by-splay}
    The number of bends created by a splay is at most $O(1)+\text{\#zig-zigs}$;
    the number of bends removed by a splay is  at least \#zig-zags.
\end{lemma}

\begin{proof} 
	    Let $x^R$ be the top-most node of the heavy path of $x^+$ before the splay, i.e., the former right child of $x^\circ$.
    By \Cref{lemma:heavy-paths-after-splay}, $x^+$ is the sole right heap-child of $x$. Since the subtree of $x^R$ is the left-most one right of $x$, all edges on the path to $x^R$ (which is the new prefix of the heavy path to $x^+$) are left and hence no internal node above $x^R$ on the heavy path of $x^+$ becomes a bend. The heavy path below $x^R$ remains the same, so possibly only $x^R$ could have become a new bend on the heavy path of $x^+$.
    Similarly, all edges to analogously defined $x^L$ are only right edges, hence $x^L$ is the only node that could have become a bend on the heavy path of $x^-$.

	    Next, by \Cref{lemma:heavy-paths-after-splay}, for every zig-zig there is one node $w$---the winner of the triple---whose heavy path acquires a new internal node. Consequently, the former top-most internal node on the heavy path of $w$ could become a new bend, but all its other internal nodes are unaffected. This creates at most \#zig-zigs new bends.
    
    The heavy path of $x$ after the splay contains only a single internal node, hence it has no bends. Prior to the splay it could have had many bends, which got removed.
	    In particular, we show that there were at least \#zig-zags bends.
	    In a zig-zag, the edge from the upper hit node to the lower hit node and the heavy edge from the lower hit node toward $x^\circ$ go in opposite directions, so the lower hit node must have been a bend. The hit nodes for all zig-zags are disjoint; thus, there were \#zig-zags such bends, and all were removed.

    By \Cref{lemma:heavy-paths-after-splay}, no other heavy paths were modified under the assumption of low ranks of $x^-$ and $x^+$, hence no other bends could have been created.

	    Restoring the true ranks of $x^-$ and $x^+$ requires $O(1)$ heavy-path junctions, each creating $O(1)$ bends by \Cref{lemma:heavy-path-junction}.
    
    This totals at most $O(1)+\text{\#zig-zigs}$ created bends and at least \#zig-zags  removed bends as desired.  
\end{proof}


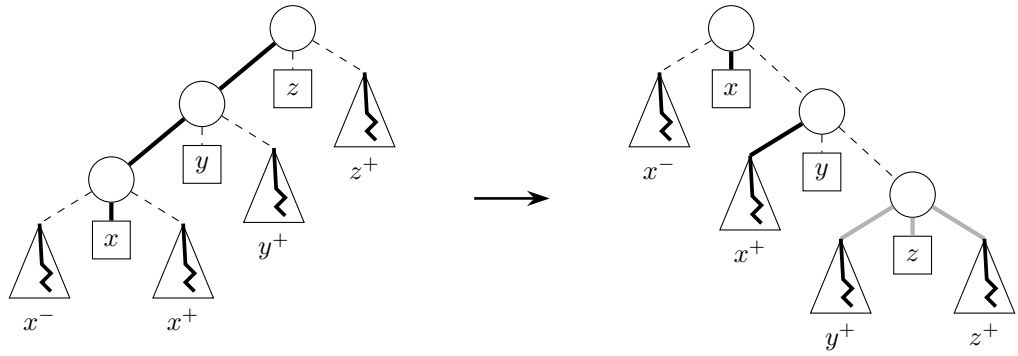
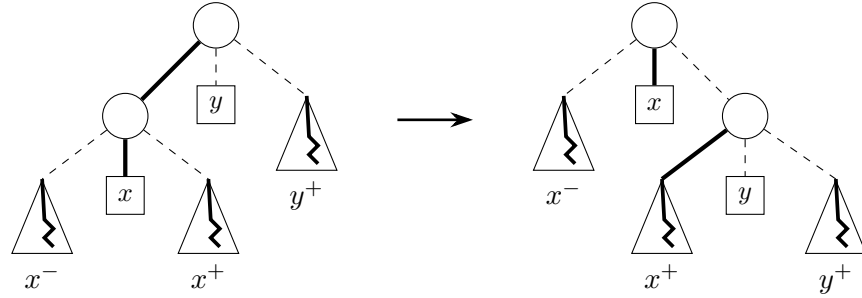
\begin{figure}[p]
\centering

\begin{subfigure}{\textwidth}
\centering

\begin{tikzpicture}[
    vertex/.style={
        circle,
        draw,
        fill=white,
        minimum size=6mm,
        inner sep=0pt
    },
    childsquare/.style={
        rectangle,
        draw,
        fill=white,
        minimum size=5mm,
        inner sep=0pt,
        font=\small
    },
    trig/.style={
        isosceles triangle,
        draw,
        fill=white,
        minimum height=8mm,
        minimum width=8mm,
        inner sep=0pt,
        shape border rotate=90
    },
    subtree/.style={
        trig,
        yshift=-11mm
    },
    dashed edge/.style={
        dashed
    },
    solid edge/.style={
        line width=1.6pt
    },
    transform arrow/.style={
        -{Stealth[length=3mm,width=2mm]},
        line width=1pt
    },
    every label/.style={
        font=\small
    }
]


\node[vertex] (Az) at (-3.8,  1.0) {};
\node[vertex] (Ay) at (-5.0,  0.0) {};
\node[vertex] (Ax) at (-6.2, -1.0) {};

\draw[solid edge] (Az) -- (Ay) -- (Ax);

\node[childsquare] (Asx) at ($(Ax)+(0,-0.8)$) {$x$};
\node[childsquare] (Asy) at ($(Ay)+(0,-0.8)$) {$y$};
\node[childsquare] (Asz) at ($(Az)+(0,-0.8)$) {$z$};

\node[subtree,label=below:{$x^{-}$}] (Axm)
    at ($(Ax)+(-0.95,-0.20)$) {};
\node[subtree,label=below:{$x^{+}$}] (Axp)
    at ($(Ax)+( 0.95,-0.20)$) {};
\node[subtree,label=below:{$y^{+}$}] (Ayp)
    at ($(Ay)+( 0.95,-0.20)$) {};
\node[subtree,label=below:{$z^{+}$}] (Azp)
    at ($(Az)+( 0.95,-0.20)$) {};

\draw[dashed edge] (Ax) -- (Axm.apex);
\draw[solid edge]  (Ax) -- (Asx);
\draw[dashed edge] (Ax) -- (Axp.apex);

\draw[dashed edge] (Ay) -- (Asy);
\draw[dashed edge] (Ay) -- (Ayp.apex);

\draw[dashed edge] (Az) -- (Asz);
\draw[dashed edge] (Az) -- (Azp.apex);


\draw[transform arrow] (-1.4,-1.25) --(-0.4,-1.25);

\node[vertex] (Bx) at ( 2.0,  1.0) {};
\node[vertex] (By) at ( 3.2, -0.1) {};
\node[vertex] (Bz) at ( 4.4, -1.2) {};

\draw[dashed edge] (Bx) -- (By) -- (Bz);

\node[childsquare] (Bsx) at ($(Bx)+(0,-0.8)$) {$x$};
\node[childsquare] (Bsy) at ($(By)+(0,-0.8)$) {$y$};
\node[childsquare] (Bsz) at ($(Bz)+(0,-0.8)$) {$z$};

\node[subtree,label=below:{$x^{-}$}] (Bxm)
    at ($(Bx)+(-0.95,-0.20)$) {};
\node[subtree,label=below:{$x^{+}$}] (Bxp)
    at ($(By)+(-0.95,-0.20)$) {};
\node[subtree,label=below:{$y^{+}$}] (Byp)
    at ($(Bz)+(-0.95,-0.20)$) {};
\node[subtree,label=below:{$z^{+}$}] (Bzp)
    at ($(Bz)+( 0.95,-0.20)$) {};

\draw[dashed edge] (Bx) -- (Bxm.apex);
\draw[solid edge]  (Bx) -- (Bsx);

\draw[solid edge]  (By) -- (Bxp.apex);
\draw[dashed edge] (By) -- (Bsy);

\draw[solid edge, color=black!30] (Bz) -- (Byp.apex);
\draw[solid edge, color=black!30] (Bz) -- (Bsz);
\draw[solid edge, color=black!30] (Bz) -- (Bzp.apex);


\foreach \T in {Axm,Axp,Ayp,Azp,Bxm,Bxp,Byp,Bzp}{
    \draw[line width=1.4pt]
        ($(\T.center)+(0.15,-0.18)$)
        -- ++(-0.12,0.12)
        -- ++( 0.12,0.12)
        -- ++(-0.12,0.12)
        -- (\T.apex);
}

\end{tikzpicture}

\caption{The structure of heavy paths after a zig-zig on $x$ ($x^+$ is assumed to be the second least-rank node). The grey solid edges indicate which edges can be either light or heavy, depending on the exact ranks. The edge to the lowest-rank node among $y^+$, $z$, and $z^+$ is heavy.}
\label{fig:splay-heavy-paths:zig-zig}

\end{subfigure}

\begin{subfigure}{\textwidth}
\centering
\begin{tikzpicture}[
    vertex/.style={
        circle,
        draw,
        fill=white,
        minimum size=6mm,
        inner sep=0pt
    },
    childsquare/.style={
        rectangle,
        draw,
        fill=white,
        minimum size=5mm,
        inner sep=0pt,
        font=\small
    },
    trig/.style={
        isosceles triangle,
        draw,
        fill=white,
        minimum height=8mm,
        minimum width=8mm,
        inner sep=0pt,
        shape border rotate=90
    },
    subtree/.style={
        trig,
        yshift=-3mm
    },
    dashed edge/.style={
        dashed
    },
    solid edge/.style={
        line width=1.6pt
    },
    transform arrow/.style={
        -{Stealth[length=3mm,width=2mm]},
        line width=1pt
    }
]


\node[vertex] (Lroot) at (-5.0,  0.0) {};
\node[vertex] (Lmid)  at (-6.2, -1.2) {};

\node[childsquare] (Ly) at (-5.0, -1.05) {$y$};

\node[subtree,label=below:{$y^{+}$}] (Lyp)
    at (-3.8, -1.35) {};

\node[subtree,label=below:{$x^{-}$}] (Lxm)
    at (-7.3, -2.45) {};

\node[childsquare] (Lx) at (-6.2, -2.25) {$x$};

\node[subtree,label=below:{$x^{+}$}] (Lxp)
    at (-5.1, -2.45) {};

\draw[solid edge] (Lroot) -- (Lmid);
\draw[dashed edge] (Lroot) -- (Ly);
\draw[dashed edge] (Lroot) -- (Lyp.apex);

\draw[dashed edge] (Lmid) -- (Lxm.apex);
\draw[solid edge] (Lmid) -- (Lx);
\draw[dashed edge] (Lmid) -- (Lxp.apex);


\draw[transform arrow] (-2.6,-1.25) -- (-1.6,-1.25);


\node[vertex] (Rroot) at ( 0.8,  0.0) {};
\node[vertex] (Rmid)  at ( 2.0, -1.2) {};

\node[subtree,label=below:{$x^{-}$}] (Rxm)
    at (-0.4, -1.35) {};

\node[childsquare] (Rx) at (0.8, -1.05) {$x$};

\node[subtree,label=below:{$x^{+}$}] (Rxp)
    at (0.9, -2.45) {};

\node[childsquare] (Ry) at (2.0, -2.25) {$y$};

\node[subtree,label=below:{$y^{+}$}] (Ryp)
    at (3.2, -2.45) {};

\draw[dashed edge] (Rroot) -- (Rxm.apex);
\draw[solid edge] (Rroot) -- (Rx);
\draw[dashed edge] (Rroot) -- (Rmid);

\draw[solid edge] (Rmid) -- (Rxp.apex);
\draw[dashed edge] (Rmid) -- (Ry);
\draw[dashed edge] (Rmid) -- (Ryp.apex);


\foreach \T in {Lyp,Lxm,Lxp,Rxm,Rxp,Ryp}{
    \draw[line width=1.4pt]
        ($(\T.center)+(0.15,-0.18)$)
        -- ++(-0.12,0.12)
        -- ++( 0.12,0.12)
        -- ++(-0.12,0.12)
        -- (\T.apex);
}

\end{tikzpicture}

\caption{The effect on heavy paths after rotating $x$ (under the assumption that $x^+$ is the next least-rank node). A zig consists of one such rotation, while a zig-zag consists of two.}
\label{fig:splay-heavy-paths:rotation}

\end{subfigure}

\caption{The changes to the structure of heavy paths after zig/zig-zag/zig-zig of the splayed node.}

\end{figure}
\usetikzlibrary{patterns.meta}
\newcommand{\treezig}[1]{%
    \draw[line width=1.4pt]
        ($(#1.center)+(0.15,-0.18)$)
        -- ++(-0.12,0.12)
        -- ++( 0.12,0.12)
        -- ++(-0.12,0.12)
        -- (#1.apex);
}

\begin{figure}[t]
\centering

\begin{tikzpicture}[
    scale=0.70,
    transform shape,
    >=Stealth,
    vertex/.style={
        circle,
        draw,
        fill=white,
        minimum size=6mm,
        inner sep=0pt
    },
    childsquare/.style={
        rectangle,
        draw,
        fill=white,
        minimum size=5mm,
        inner sep=0pt
    },
    trig/.style={
        isosceles triangle,
        draw,
        fill=white,
        minimum height=8mm,
        minimum width=8mm,
        inner sep=0pt
    },
    hatch horizontal/.style={
        pattern={Lines[angle=0, distance=2.5pt]}, 
        pattern color=DarkGreen
    },
    hatch left/.style={
        pattern=north west lines,
        pattern color=blue
    },
    hatch right/.style={
        pattern=north east lines,
        pattern color=red
    },
    main edge/.style={
        line width=2pt
    }
]


\begin{scope}


\node[vertex,hatch horizontal] (n1) at (-4.0, -1.4) {};
\node[vertex,hatch horizontal] (n2) at ( 2.6, -2.6) {};

\node[vertex,hatch left] (n3) at (-2.6, -4.0) {};
\node[vertex,hatch left] (n4) at (-1.8, -5.4) {};

\node[vertex,hatch right] (n5) at ( 1.4, -6.8) {};
\node[vertex,hatch right] (n6) at ( 0.8, -8.2) {};

\node[vertex] (n7) at (-0.4, -9.6) {$x^\circ$};


\draw[main edge]
    (n1) -- (n2)
    -- (n3)
    -- (n4)
    -- (n5)
    -- (n6)
    -- (n7);


\node[childsquare,hatch horizontal] (s1) at ($(n1)+(0,-1.0)$) {};
\draw[dashed] (n1) -- (s1);

\node[childsquare,hatch horizontal] (s2) at ($(n2)+(0,-1.0)$) {};
\draw[dashed] (n2) -- (s2);

\node[childsquare,hatch left] (s3) at ($(n3)+(0,-1.0)$) {};
\draw[dashed] (n3) -- (s3);

\node[childsquare,hatch left] (s4) at ($(n4)+(0,-1.0)$) {};
\draw[dashed] (n4) -- (s4);

\node[childsquare,hatch right] (s5) at ($(n5)+(0,-1.0)$) {};
\draw[dashed] (n5) -- (s5);

\node[childsquare,hatch right] (s6) at ($(n6)+(0,-1.0)$) {};
\draw[dashed] (n6) -- (s6);

\node[childsquare] (s7) at ($(n7)+(0,-1.0)$) {$x$};
\draw[main edge] (n7) -- (s7);


\node[trig,shape border rotate=90,hatch horizontal] (l1)
    at ($(n1)+(-1.4,-1.3)$) {};
\draw[dashed] (n1) -- (l1.apex);
\treezig{l1}

\node[trig,shape border rotate=90,hatch left] (l3)
    at ($(n3)+(-1.4,-1.3)$) {};
\draw[dashed] (n3) -- (l3.apex);
\treezig{l3}

\node[trig,shape border rotate=90,hatch left] (l4)
    at ($(n4)+(-1.4,-1.3)$) {};
\draw[dashed] (n4) -- (l4.apex);
\treezig{l4}


\node[trig,shape border rotate=90,hatch horizontal] (r2)
    at ($(n2)+(1.4,-1.3)$) {};
\draw[dashed] (n2) -- (r2.apex);
\treezig{r2}

\node[trig,shape border rotate=90,hatch right] (r5)
    at ($(n5)+(1.4,-1.3)$) {};
\draw[dashed] (n5) -- (r5.apex);
\treezig{r5}

\node[trig,shape border rotate=90,hatch right] (r6)
    at ($(n6)+(1.4,-1.3)$) {};
\draw[dashed] (n6) -- (r6.apex);
\treezig{r6}


\node[trig,shape border rotate=90,label=below:{$x^-$}] (l7)
    at ($(n7)+(-1.4,-1.3)$) {};
\draw[dashed] (n7) -- (l7.apex);
\treezig{l7}

\node[trig,shape border rotate=90,label=below:{$x^+$}] (r7)
    at ($(n7)+(1.4,-1.3)$) {};
\draw[dashed] (n7) -- (r7.apex);
\treezig{r7}

\end{scope}


\draw[->,line width=1.4pt] (4.6,-6.1) -- (5.9,-6.1);


\begin{scope}[xshift=11.0cm]


\node[vertex] (m0) at (0.0,-1.4) {$x^\circ$};
\node[childsquare] (mx) at (0.0,-2.7) {$x$};

\draw[main edge] (m0) -- (mx);


\node[vertex,hatch horizontal] (a1) at (-2.4,-3.6) {};
\node[vertex,hatch horizontal] (b2) at ( 2.4,-3.6) {};

\draw[dashed] (m0) -- (a1);
\draw[dashed] (m0) -- (b2);


\node[childsquare,hatch horizontal] (sa1) at (-2.4,-4.9) {};
\node[trig,shape border rotate=90,hatch horizontal] (ta1)
    at (-3.6,-5.3) {};

\draw[dashed] (a1) -- (sa1);
\draw[dashed] (a1) -- (ta1.apex);
\treezig{ta1}


\node[vertex,hatch left] (a3) at (-1.25,-5.6) {};
\draw[main edge] (a1) -- (a3);

\node[childsquare,hatch left] (sa3) at (-1.25,-6.9) {};
\node[trig,shape border rotate=90,label=below:{$x^-$}] (txminus)
    at (-0.4,-7.65) {};

\draw[dashed] (a3) -- (sa3);
\draw[main edge] (a3) -- (txminus.apex);
\treezig{txminus}

\node[vertex,hatch left] (a4) at (-2.75,-7.45) {};
\draw[dashed] (a3) -- (a4);

\node[childsquare,hatch left] (sa4) at (-2.75,-8.8) {};
\node[trig,shape border rotate=90,hatch left] (ta4L)
    at (-4.05,-9.25) {};
\node[trig,shape border rotate=90,hatch left] (ta4R)
    at (-1.50,-9.25) {};

\draw[main edge] (a4) -- (sa4);
\draw[dashed] (a4) -- (ta4L.apex);
\draw[dashed] (a4) -- (ta4R.apex);

\treezig{ta4L}
\treezig{ta4R}


\node[childsquare,hatch horizontal] (sb2) at (2.4,-4.9) {};
\node[trig,shape border rotate=90,hatch horizontal] (tb2)
    at (3.6,-5.3) {};

\draw[dashed] (b2) -- (sb2);
\draw[dashed] (b2) -- (tb2.apex);
\treezig{tb2}


\node[vertex,hatch right] (b5) at (1.30,-5.6) {};
\draw[main edge] (b2) -- (b5);

\node[childsquare,hatch right] (sb5) at (1.30,-6.9) {};
\node[trig,shape border rotate=90,label=below:{$x^+$}] (txplus)
    at (0.55,-7.65) {};

\draw[dashed] (b5) -- (sb5);
\draw[main edge] (b5) -- (txplus.apex);
\treezig{txplus}

\node[vertex,hatch right] (b6) at (2.75,-7.35) {};
\draw[dashed] (b5) -- (b6);

\node[childsquare,hatch right] (sb6) at (2.75,-8.7) {};
\node[trig,shape border rotate=90,hatch right] (tb6L)
    at (1.55,-9.25) {};
\node[trig,shape border rotate=90,hatch right] (tb6R)
    at (4.00,-9.25) {};

\draw[dashed] (b6) -- (sb6);
\draw[main edge] (b6) -- (tb6L.apex);
\draw[dashed] (b6) -- (tb6R.apex);

\treezig{tb6L}
\treezig{tb6R}

\end{scope}

\end{tikzpicture}

\caption{The effect of splaying node $x$ (under the assumption that $x^-$ and $x^+$ are the next least-rank nodes). Nodes hit by the first zig-zig are red and right-diagonally hatched, nodes hit by the second zig-zig are blue and left-diagonally hatched, and nodes hit by the zig-zag are green and horizontally hatched. Note that the structure below the heavy paths to $x^-$ and $x^+$ depends on the least-rank node in the triplet and can be any configuration from \Cref{fig:splay-heavy-paths:zig-zig}.
Observe also that zig-zags happened exactly when the bottom-most hit internal node was a bend.}
\label{fig:splay}

\end{figure}
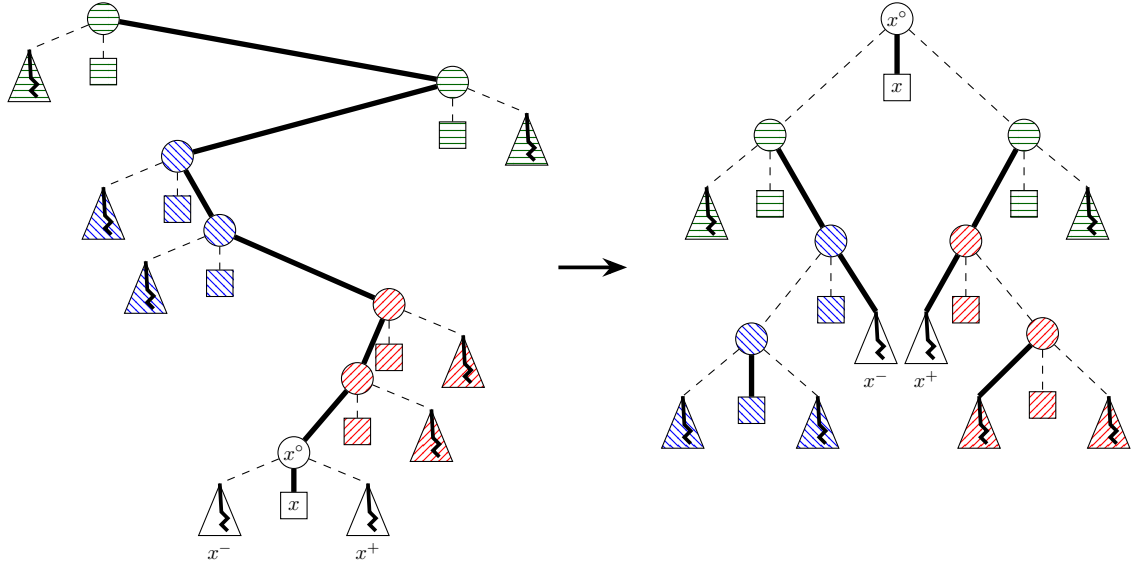

\subsubsection{Adjusting lazy intervals to changes after splaying}

\Cref{lemma:heavy-paths-after-splay} tells us how the structure of heavy paths changes. This change necessitates a change in the lazy intervals.

As in the heavy-path analysis, we first assume that $x^-$ and $x^+$ have the next-lowest ranks and therefore become the sole heap-children of $x$; after the splay, we remove this assumption.
To make this assumption, we split the lazy intervals containing $x^-$ and $x^+$ to make them singleton lazy intervals, converting them if necessary. Then, changes to the ranks of $x^-$ and $x^+$ do not affect the contracted point gaps of any other heap-children of $x$. If we ensure that all non-empty lazy intervals of $x^-$ and $x^+$ are converted to shrinking, no heap-children of $x^-$ or $x^+$ change their contracted point gaps, as every change is absorbed by the interval gaps.

Under this assumption, these two singleton lazy intervals remain the only two shrinking lazy intervals of $x$.

Next, by \Cref{lemma:heavy-paths-after-splay}, we have two more tasks. First, we need to account for the effects of zig-zigs. Second, we need to ensure that every remaining heap-child of $x$ changes its heap-parent to $x^-$ or $x^+$.
The second task is straightforward: once all heap-children moved by zig-zigs have been moved accordingly within the lazy intervals, all remaining right heap-children of $x$ except for $x^+$ become right heap-children of $x^+$.
Since $x^+$ is already in a separate lazy interval, all other right lazy intervals of $x$ are transferred to $x^+$.
Similarly, all heap-children in left lazy intervals of $x$ except for the singleton one that contains $x^-$ become left heap-children of $x^-$, which on the level of lazy intervals is handled by transferring the left lazy intervals of $x$ to $x^-$.
These are all valid lazy-interval transfers, since the heap-children in the transferred intervals appear above the lazy intervals that were already owned by $x^-$ and $x^+$.

The remaining part that we need to handle is the effect of zig-zigs. We show that the effect of a zig-zig corresponds to exactly two pairings and then explain how to perform them on lazy intervals using pair-ups, inserts, and deletes. This is important because we will later bound the total number of pairings by bounding the total number of these operations on lazy intervals.

\begin{lemma}\label{lemma:zig-zig-to-pairings}
    The effect of a zig-zig of $x$ corresponds to two pairings assuming that afterwards the lazy intervals containing all the hit heap-children of $x$ get transferred to $x^-$/$x^+$.
\end{lemma}
\begin{proof}
First, consider a right zig-zig. Let $x^\circ$ be the splayed internal node, with $x$ as its corresponding value node; let $y^\circ$ be the parent of $x^\circ$ and $z^\circ$ the parent of $y^\circ$, with $y$ and $z$ as their respective value nodes. Let $x^-$ and $x^+$ be the least-rank nodes in the left and right subtrees of $x^\circ$, respectively, and let $y^+$ and $z^+$ be the least-rank nodes in the right subtrees of $y^\circ$ and $z^\circ$, respectively.
After the zig-zig, $x^-$ and $x^+$ remain the bottom-most heap-children of $x$, while $y$
is transferred to $x^+$ since $y^\circ$ lies on the heavy path of $x^+$. The lowest-rank node among $y^+$, $z$, and $z^+$ contains $z^\circ$ in its heavy path and is transferred to $x^+$ as a heap-child. The other two become top-most heap-children of the heavy path containing $z^\circ$. Therefore, we only need to show that this change can be decomposed into two pairings.

The proof of this depends on which of  $y^+$, $z$, and $z^+$ has the lowest rank.
Note that initially they appear in this order consecutively as the right heap-children of $x$.
If $z$ is the lowest-rank node, then $y^+$ and $z^+$ become the new top-most heap-children of $z$, with $y^+$ becoming the left heap-child and $z^+$ becoming the right heap-child. Since both pairs $(z,z^+)$ and $(y^+,z)$ are consecutive, this can be performed with two pairings in either order.

If $y^+$ has the lowest rank, then $z^+$ becomes the new top-most right heap-child of $y^+$ and $z$ becomes the second top-most right heap-child of $y^+$. Therefore, we can first pair $y^+$ and $z$, which are consecutive, and then $y^+$ and $z^+$, which become consecutive when $z$ is removed. This correctly puts $z^+$ above $z$ among the right heap-children of $y^+$.

Finally, if $z^+$ has the lowest rank, then $y^+$ becomes the new top-most left heap-child of $z^+$ and $z$ becomes the second top-most left heap-child of $z^+$. Therefore, we can first pair $z$ and $z^+$, which are consecutive, and then $y^+$ and $z^+$, which become consecutive when $z$ is removed. This correctly puts $y^+$ above $z$ among the left heap-children of $z^+$.
All three cases are illustrated in \Cref{fig:splay-heavy-paths:zig-zig}.

After the splay, the losers are heap-children of the winner, so their ranks are strictly larger than the winner's rank. Therefore, the requirement that the paired vertices have different ranks is satisfied.

The situation for left zig-zigs is analogous: $x^-$, $y^-$, and $z^-$ replace $x^+$, $y^+$, and $z^+$, respectively, and the roles of left and right heap-children are swapped.

\end{proof}

Hence the effect of each zig-zig can be decoupled into two pairings that happen one after the other. We call the first of these a \emph{first pairing} and the second a \emph{second pairing}.
We perform internal pairings using pair-ups. Boundary pairings cannot be performed by pair-ups, so we implement them by deleting each loser and inserting it into the winner's lazy interval.

\begin{lemma}\label{lemma:pairings-in-parallel}
    For a given splay, no two first pairings share a heap-child; additionally, no two second pairings share a heap-child.
\end{lemma}

\begin{proof}
    By \Cref{lemma:heavy-paths-after-splay}, the sets of heap-children hit by individual zig-zigs are disjoint, therefore the pairings are disjoint as each comes from a different zig-zig. Analogously for second pairings.
\end{proof}

From \Cref{lemma:pairings-in-parallel}, we can perform the pair-ups / deletes and inserts for first pairings in any order or in parallel; the same for second pairings. This will be later important to allow us to bound the number of pair-ups.

These operations on lazy intervals capture the changes to heavy paths under the assumption that $x^-$ and $x^+$ are the next-lowest-rank vertices. This need not be the case, so we must also update the intervals at the $O(1)$ heavy-path junctions required to remove the assumption. By \Cref{lemma:heavy-path-junction}, we know how to do this using $O(1)$ paid operations and creating $O(1)$ new lazy intervals.
Since all transferred lazy intervals are shrinking, if any transferred element in fact has lower rank than $x^-$ or $x^+$, this difference is absorbed by the interval gap, and all point gaps remain non-negative throughout.

Overall, the following sequence of operations correctly updates the lazy intervals for a splay.

\begin{enumerate}
    \item Convert the growing lazy intervals of $x$, $x^-$, and $x^+$ to shrinking.
    \item Split the lazy intervals containing $x^-$ and $x^+$ so that $x^-$ and $x^+$  are in singleton lazy intervals.
	    \item For every first pairing $(u,v)$ during a splay, in any order, do the following:
    \begin{enumerate}
        \item If $(u,v)$ is a boundary pairing, delete $v$ from its former lazy interval and insert it into the appropriate growing lazy interval of $u$.
        \item If $(u,v)$ is an internal pairing, perform a pair-up.
    \end{enumerate}
    \item For every second pairing, do steps 3a and 3b.
	    \item Transfer all non-empty right intervals of $x$, except for the singleton interval containing $x^+$, to $x^+$, and all non-empty left intervals of $x$, except for the singleton interval containing $x^-$, to $x^-$. All the non-empty intervals are shrinking due to Step 1.
    \item For every junction that happens due to $x^-$ and $x^+$ not having the next-lowest rank, update the lazy intervals as described in \Cref{sec:heavy-path-junction}.
\end{enumerate}

We now analyze the number of paid operations needed for a splay and the number of lazy intervals created.

An important step is to show that the number of internal pairings within a lazy interval after the first splay over it---meaning a splay for which the lazy interval was owned by the splayed vertex immediately before the splay---is always a constant fraction of the number of heap-children in it. This ensures that lazy-interval sizes decay exponentially with the number of splays over them, which helps us bound the boundary pairings.

\begin{lemma}\label{lemma:unbroken-lazy-intervals}
	    Consider a lazy interval $\fI$. After the first splay over $\fI$, the interval is never broken again; nor is any interval obtained by splitting $\fI$.
\end{lemma}

\begin{proof}
    Consider the lazy interval $\fI$. After the first splay, $\fI$ is either the singleton interval containing $x^-$/$x^+$, or it is one of the transferred intervals to $x^-$/$x^+$  by \Cref{lemma:heavy-paths-after-splay}.
	    The newly added prefix of the heavy path to $x^-$ or $x^+$ contains no bends, as it consists only of left or right edges, respectively.
    Hence $\fI$
    is not broken.
	    The only way this lazy interval can become broken is for a new internal node to appear on the heavy path between the internal nodes to which the top-most and bottom-most heap-children are attached by their light edges. This is because removing internal nodes cannot create bends.

	    However, adding an internal node to the heavy path that owns $\fI$ requires adding a new heap-child to the owner of $\fI$. Since all operations except heap-child-exchange add new heap-children only above shrinking intervals, and $\fI$ is shrinking after the splay, these operations cannot add an internal vertex below the top-most heap-child of $\fI$.
	    Heap-child-exchange is the only operation without this guarantee; however, because it preserves the light edge, it also preserves the structure of the heavy path itself and cannot add bends.

	    Therefore, since no operation can add bends to a shrinking lazy interval, the interval is never broken again after a splay over it.
	    In particular, splitting an interval does not create any new bends, so neither of the two resulting lazy intervals is broken.
\end{proof}

\begin{lemma}\label{lemma:internal-pairings-fraction}
	    For a fixed lazy interval $\fI$, unless $\fI$ is broken, which can happen only during the first splay over it,
    the set of all internal pairings $M$ from that splay has size at least
    $|M| \geq (|\fI| - 6)/2$, where $|\fI|$ is measured immediately before the splay.
	    In particular, for the set of first internal pairings $M_1$ and the set of second internal pairings $M_2$,
    $|M_1|, |M_2| \geq (|\fI| - 6)/4$.
\end{lemma}

For the next two lemmata, we extend the classification of internal and boundary pairings to zig-zigs. We say that a zig-zig is internal to $\fI$ if all heap-children hit by it are from $\fI$, and that a zig-zig is a boundary zig-zig for $\fI$ if at least one, but not all four, of the heap-children hit by it belong to $\fI$.
Internal zig-zigs correspond only to internal pairings.

\begin{proof}
	    After the first splay, $\fI$ is not broken by \Cref{lemma:unbroken-lazy-intervals}.
	    This implies that, except for the two top-most and the two bottom-most heap-children in $\fI$, none can be hit by a zig-zag.
    This holds because a zig-zag hits four consecutive heap-children, two of them left and two of them right. If other than the two top-most or bottom-most were hit, it would imply that there are heap-children on the other side in between in top-down order, contradicting non-brokenness.
	    A zig can hit only the two top-most heap-children, as it always hits only the two top-most heap-children on the whole heavy path.
	    A boundary zig-zig can hit at most the $3$ top-most and $3$ bottom-most heap-children of $\fI$, as it must contain a heap-child outside $\fI$ and the heap-children hit by it must be consecutive.
    
    Hence, heap-children except for these at most $6$ heap-children are hit by internal zig-zigs as they are not hit by zig-zags, zigs, nor boundary zig-zigs.
	    Each internal zig-zig hits exactly $4$ heap-children and corresponds to one first pairing and one second pairing.
	    Hence, $|M_1| \geq (|\fI| - 6)/4$ and
	    $|M_2| \geq (|\fI| - 6)/4$, so $|M| = |M_1|+|M_2| \geq (|\fI| - 6)/2$.
\end{proof}

\begin{lemma}\label{lemma:outer-ops-block}
	    Consider a lazy interval $\fI$. Over the lifetime of $\fI$, there are $O(\log \contractedgapbound)$ important boundary pairings
    with $\fI$ containing the loser of the pairing.
\end{lemma}

\begin{proof}

    Consider $\fI$.
	    In one splay, there can be at most $2$ boundary zig-zigs, one on each side of $\fI$, as the hit heap-children are consecutive.
	    Each boundary zig-zig consists of at most $2$ boundary pairings, and every zig-zig containing at least one boundary pairing must be a boundary zig-zig. Hence, during one splay, there can be at most $4$ boundary pairings with a loser in $\fI$.

    We now distinguish three cases depending on the status and size of $\fI$.

    \begin{itemize}
        \item $\fI$ is broken. Then after a single splay it is no longer broken by \Cref{lemma:unbroken-lazy-intervals}, so there is at most one such splay.
        \item $\fI$ has size at most $12$. Then there can be at most $12$ splays with at least one boundary pairing having a loser in $\fI$, since the loser leaves the lazy interval, so after every such splay, the size decreases by at least $1$.
	        \item In all other cases, $\fI$ is unbroken and $|\fI| \geq 12$. Moreover, $|\fI| \leq \kappa \contractedgapbound \log \contractedgapbound$, as otherwise the boundary pairings would not be important. By \Cref{lemma:internal-pairings-fraction}, $\fI$ has at least $(|\fI|-6)/2$ internal pairings, which for lazy intervals of size at least $12$ is at least $|\fI|/4$. For each such internal pairing, the loser is removed from $\fI$, so the size of $\fI$ decreases by at least $|\fI|/4$. Thus, there are $O(\log \contractedgapbound)$ such splays.
    \end{itemize}

	    Therefore, there are $O(\log \contractedgapbound)$ splays with at least one important boundary pairing whose loser is in $\fI$, and there are $O(1)$ such boundary pairings per splay.
	    Thus, over all splays, there are $O(\log \contractedgapbound)$ boundary pairings whose loser is in $\fI$.
\end{proof}

Now, we bound the total number of paid operations and the total number of lazy intervals created during splaying.

\begin{lemma}\label{lemma:blocks-by-splay}
    Every splay creates $O(1)$ new lazy intervals.
\end{lemma}

\begin{proof}
	    Step 1 creates $6$ new lazy intervals by conversion, and Step 2 creates $O(1)$ more by splitting.
    $4$ lazy intervals are created in total for the $2$ transfers in Step 5.
	    Finally, Step 6 creates $O(1)$ lazy intervals, corresponding to $O(1)$ heavy-path junctions, each of which creates $O(1)$ new lazy intervals by \Cref{lemma:heavy-path-junction}.
\end{proof}

\begin{lemma}\label{lemma:indels-by-splay}
	    Let $\ell$ be the total number of lazy intervals ever created and $u$ the total number of unimportant boundary pairings performed. The total number of paid operations over all splays is $u + O(\ell \log \contractedgapbound + \cost(T^*))$.
\end{lemma}

\begin{proof}
	    In Step 3(a) and the analogous case in Step 4, there is one insert per boundary pairing.
	    There are $O(\log \contractedgapbound)$ important boundary pairings per lazy interval, so there are $O(\ell \log \contractedgapbound)$ corresponding inserts over all $\ell$ lazy intervals. There are $u$ unimportant boundary pairings in total and hence $u$ corresponding inserts.

	    Step 6 uses $O(1)$ inserts, as there are $O(1)$ heavy-path junctions, each requiring $O(1)$ paid operations by \Cref{lemma:heavy-path-junction}.
	    The number of splays is at most $\cost(T^*)$, since the nearly optimal tree $T^*$ pays $1$ for each access.
\end{proof}

\subsection{Bounding the total number of bad internal pairings}\label{app:interface}

The goal of this section is to bound the total number of bad internal pairings by the total number of paid operations on the lazy intervals and the total number of good internal pairings.

\begin{lemma}
\label{lemma:bounding-bad-pairings}
    Let $n$ be the number of elements, $g$ be the total number of good internal pairings, $a$ the total number of internal pairings (both good and bad), and $p$ the total number of paid operations performed.
    Then
    \[
    a \in O((p+n) \cdot \contractedgapbound \log \contractedgapbound + g\log \contractedgapbound).
    \]
\end{lemma}

We prove \Cref{lemma:bounding-bad-pairings} in two steps. First, we give an interface and prove a lemma on its properties that is an abstraction of the core argument in the analyses of multi-pass and pure pairing heaps~\cite{Sinnamon2025,PairingHeaps2026}.
We then show how to translate operations on lazy intervals into this interface and how to derive a bound on the number of internal pairings.

The interface result isolates the properties needed for splay trees. It does not use the full strength of the proof method of \cite{PairingHeaps2026}: in particular, in that method a \ValueChange{} operation incurs a charge logarithmic only in the element's current value rather than in the maximum value. On the other hand, our interface is more general in several respects: it includes a new \Split{} operation, a more general \Pass{} operation, and no requirement that the elements form a tree. Adapting the proof method to these changes does not require major new ideas.

\subsubsection{An interface abstraction of pairing-heap analysis}

Throughout, we fix the following:
\begin{itemize}
    \item A set of $n$ elements $V$.
    \item Integer values $C_{\text{growing}}, C_{\text{exempt}} \in O(1)$ and $G \in \mathbb{Z}_{\geq 2}$.
    \item A real value $\gamma \in [\frac{1}{O(1)}, \frac{1}{2}]$.
\end{itemize}

\paragraph{State}
The state of the interface consists of the following:
\begin{itemize}
	\item A collection of \textit{blocks} $B \in \mathcal{B}$, each containing an ordered sequence of elements. At any time, each element belongs to at most one block.
	      \begin{itemize}
		      \item Each block has a \textit{creator element} $u \in V$ and a \textit{pass count} $p \in \mathbb{Z}_{\geq 0}$.
		      \item Each block is either \textit{growing} or \textit{shrinking}.
		      \item At any time, for any element $v \in V$, there are at most $C_{\text{growing}}$ growing blocks with creator $v$.
	      \end{itemize}
	\item A value function $f : V \rightarrow \{1, 2, \dots, \contractedgapbound\}$.
	\item Initially, there are no blocks, and the values $f(v)$ are arbitrary.
\end{itemize}

\paragraph{Operations}
The state can be modified by the following operations on the interface:
\begin{itemize}
	\item $\Convert(B)$: turn a growing block $B$ into a shrinking block.
	\item $\Create(v)$: create a new empty growing block with creator element $v$ and pass count $0$.\\
		      \textit{\small (Afterward, there must be at most $C_{\text{growing}}$ growing blocks with creator $v$).}
      \item $\Insert(B, v)$: given an element $v$ not in any block and a growing block $B$ with creator $u$ satisfying $f(v) \geq f(u) + 1$, add $v$ as the leftmost element of $B$.
      \item $\Delete(v)$: given an element $v$ in a shrinking block $B$, remove $v$ from $B$.
      \item $\Pass(B, M)$: $B$ is a shrinking block. The \textit{matching} $M$ is a set of \emph{pairing tuples} $(v, w, B_v)$, where $v$ and $w$ are \emph{adjacent} elements of $B$ in either order and $B_v$ is a growing block of $v$. The matching $M$ must satisfy the following:
	      \begin{itemize}
		      \item For every pairing $(v, w)$, $f(w) \geq f(v) + 1$.
		      \item Each element of $B$ appears in at most one pairing in $M$.
              \item Either $|M| \geq \gamma |B|$, or the pass count $p$ of $B$ satisfies $p \leq C_{\text{exempt}}$.
	      \end{itemize}
			The pass has the following effects:
        \begin{itemize}
            \item The pass count $p$ of $B$ is incremented by $1$.
            \item For every pairing $(v, w, B_v) \in M$, the element $w$ is removed from $B$ and added as the leftmost element of the indicated growing block $B_v$ with creator $v$.
        \end{itemize}
        \item $\Split(B,k)$: given a shrinking block $B$, split it into two new shrinking blocks, one consisting of the prefix of length $k$ and the other of the suffix of length $|B| - k$. The block $B$ becomes empty. 
        Both blocks created by a \Split{} inherit the creator of the split block and its current pass count.
	\item $\ValueChange(v, \Delta)$: given an element $v \in V$ and $\Delta \in \{-1, 1\}$, assign $f(v) \gets f(v) + \Delta$. \\
	    \textit{\small ($1 \leq f(v) \leq \contractedgapbound$ must still hold afterward).}
\end{itemize}

\begin{lemma}\label{lemma:interface}
	For any sequence of interface operations, define the \textit{total activity} $A$ as the sum of matching sizes $|M|$ over all \Pass{} operations. Then,
    \begin{equation*}
        A \leq \#\Insert{} \cdot O(\contractedgapbound) + \#\ValueChange{} \cdot O(\log \contractedgapbound),
    \end{equation*}
    where $\#\Insert{}$ and $\#\ValueChange{}$ are, respectively, the total numbers of $\Insert{}$ and $\ValueChange{}$ operations in the sequence.
\end{lemma}

To prove \Cref{lemma:interface}, we need one helper lemma that allows us to ignore activity from \Pass{} operations on blocks with a high pass count. Intuitively, this holds because a block shrinks by at least a constant factor every time \Pass{} is called on it, except possibly during the first $C_{\text{exempt}} = O(1)$ such calls.
We defer the proof of \Cref{lemma:insignificant-passes} to the end of this subsection.

\begin{lemma}\label{lemma:insignificant-passes}
    For a threshold $\tau$, call a \Pass{} on a block $B$ \textit{significant} if the pass count of the block was strictly less than $\tau$ before the \Pass{}. Otherwise, call the \Pass{} \textit{insignificant}.

    Let $A_{\text{insig}}$ be the sum of matching sizes $|M|$ over all \emph{insignificant} \Pass{} operations. Then, there is a threshold $\tau
    = O(\log \contractedgapbound)$ such that, for any sequence of operations,
    \begin{equation*}
        A_{\text{insig}} \leq \frac{1}{4 \contractedgapbound} \left(A + \#\Insert{}\right).
    \end{equation*}
\end{lemma}
With this in hand, we now prove \Cref{lemma:interface}.

\begin{proof}
    Fix the threshold $\tau = O(\log \contractedgapbound)$ given by \Cref{lemma:insignificant-passes}. We first introduce some notation and terminology.

    We call a block $B$ a \textit{root block} if it was created by a \Create{} operation and a \textit{split block} otherwise, that is, if it was created by a \Split{} operation. For each split block $\tilde{B}$, there is a unique root block $B$ from which a sequence of splits (possibly interleaved with other operations) resulted in the creation of $\tilde{B}$. We call $B$ the \emph{root block} of $\tilde{B}$.
    
    At any point during the process, we define the \textit{current time} to be the total number of interface operations completed so far. For an element $v$ that belonged to a root block $B$ at some point, we write $t_{vB}$ for the time when $v$ joined $B$, $t_{B+}$ for the time when $B$ was created, and $t_{B-}$ for the time when $B$ was converted to a shrinking block, with $t_{B-} = \infty$ if this never occurs. We use the following facts about join times:
    \begin{enumerate}
        \item The time $t_{vB}$ at which an element $v$ joins a root block $B$ is well-defined and satisfies $t_{B+} < t_{vB} < t_{B-}$, because elements can only join growing blocks, can only leave shrinking blocks, and shrinking blocks cannot become growing blocks again. \label{list:one}
        \item Since split blocks are always shrinking and elements can only join growing blocks, any element $v \in \tilde{B}$ in a split block $\tilde{B}$ must have joined the root block $B$ of $\tilde{B}$ while $B$ was growing. It is therefore well-defined to write $t_{v\tilde{B}} = t_{vB}$ for a split block $\tilde{B}$ with root block $B$. \label{list:two}
        \item At any time, for any block $B$, the times $t_{vB}$ of the elements $v$ in $B$ are strictly decreasing from left (newest element) to right (oldest element). For a growing block, this follows because elements are always added as the leftmost element; for a shrinking block, it follows because no operation can reorder elements within a block. \label{list:three}
        \item At any time, for any two distinct blocks $\tilde{B}_1, \tilde{B}_2$ with the same root block $B$, the intervals of join times $[\min \{v \in \tilde{B} : t_{vB}\}, \max \{v \in \tilde{B} : t_{vB}\}]$ of $\tilde{B}_1$ and $\tilde{B}_2$ are disjoint. Before \Convert{} is applied to $B$, this holds because there is no split block with root block $B$. After \Convert{} is applied, the property is preserved when a \Split{} operation is performed or when elements leave a block with root block $B$ through a \Delete{} or \Pass{} operation. No further elements can then enter a block with root block $B$. \label{list:four}
        \item For any time $t$ and element $u$, there are at most $C_{\text{growing}}$ root blocks $B$ with creator $u$ satisfying $t \in [t_{B+}, t_{B-}]$, since otherwise this would violate the definition of $C_{\text{growing}}$. \label{list:five}
    \end{enumerate}
     We write $f_t$ for the value function at time $t$ and abbreviate $f_{vB} := f_{t_{vB}}$ for the state of the value function at the moment when $v$ joined $B$ (or its root block).
     
    For a pairing $(v, w, B')$, we refer to $v$ (the element that remains in the block) as the \textit{winner} of the pairing and to $w$ (the element that joins the growing block $B'$ of the winner) as the \textit{loser} of the pairing.
    
    For a pairing $(v, w, B')$ done during a \textsc{Pass} in a block $B$ with creator $u$, define $\text{pair-rank}(v, w) := f_{vB}(v) - f_{vB}(u)$. Note that every \text{pair-rank} is at least $1$, as whether $v$ entered the block $B$ (or its root block if $B$ is a split block) by being \textsc{Insert}ed or through losing a pairing to $u$, at that moment $t_{vB}$, it must have held that $f(v) \geq f(u) + 1$. Thus,
    \begin{equation*}
        A = \sum_{M} |M| \leq \sum_{M} \sum_{(v, w, B') \in M} \text{pair-rank}(v, w)
    \end{equation*}
    where the sums are over all matchings $M$ made during the sequence of operations. It thus suffices to bound the sum of all pair-ranks. We split each pair-rank term into three terms, defined as follows, where $u$ is again the creator of the block $B$:
    \begin{align*}
        \text{v-shift}(v, w) &:= f_{vB}(v) - f_{wB'}(v)\\
        \text{w-shift}(v, w) &:= f_{wB'}(v) - f_{wB}(u)\\
        \text{u-shift}(v, w) &:= f_{wB}(u) - f_{vB}(u)\\
    \end{align*}
    By definition, $\text{pair-rank}(v, w) = \text{v-shift}(v, w) + \text{w-shift}(v, w) + \text{u-shift}(v, w)$, as all terms but $f_{vB}(v)$ and $-f_{vB}(u)$ are both added and subtracted once. We now separately bound the sums of each shift type over all pairings.
     
    \begin{itemize}
        \item \textbf{v-shifts}: The v-shift $f_{vB}(v) - f_{wB'}(v)$ of the pairing $(v, w, B')$ done inside a block $\tilde{B}$ with root block $B$ is the change in the value $-f(v)$ in the time interval $[t_{vB}, t_{wB'})$ from the time $v$ joined $B$ to the time $w$ loses the pairing (and joins block $B'$).

        We say that a v-shift is significant if the pairing it corresponds to was in a significant \Pass{} and we call it insignificant otherwise. Since each v-shift has value at most $\contractedgapbound$, the contribution from insignificant v-shifts is at most $A_{\text{insig}} \cdot \contractedgapbound$.

        To bound the sum of significant v-shifts, consider a $\ValueChange$ operation to change the value $f(v)$ at time $t$. We want to bound the number of pairings $(v, w, B')$ for which $t \in [t_{vB}, t_{wB'})$, where $B$ is the root block of the block $\tilde{B}$ the pairing is done in.
        
        First, we argue that if any such pairing exists, the block $B$ is unique, and specifically, it must be the root block of the block $v$ is in at time $t$. This must hold, as $v$ is in a block with root block $B$ at both time $t_{vB}$ and time $t_{wB'}$, thus by \hyperref[list:one]{(1)} it must be in a block with root block $B$ in the whole time interval $[t_{vB}, t_{wB'}]$, and $t$ is in this time interval.

        Then, with a fixed element $v$ and root block $B$, if we further fix the pass count $i$ of the block $\tilde{B}$ right before the pass, the pairing $(v, w, B')$ must be unique if it exists. This holds as the pass count of the block $v$ is in from the time $t_{vB}$ to the time $v$ is no longer in a block with root block $B$ can only change when a \Pass{} operation is performed on the block $v$ is in, and this operation increments the pass count. 
        
        Thus, each \ValueChange{} contributes to at most $\tau$ significant v-shifts, and we obtain
        \begin{equation*}
            \sum_{M} \sum_{(v, w, B') \in M} \text{v-shift}(v, w) \leq A_{\text{insig}} \cdot \contractedgapbound + \#\ValueChange{} \cdot \tau.
        \end{equation*}

        \item \textbf{u-shifts}: The u-shift $f_{wB}(u) - f_{vB}(u)$ of the pairing $(v, w, B')$ done inside a block $\tilde{B}$ with root block $B$ is the change in the value $f(u)$ in the time interval $[\min\{t_{vB}, t_{wB}\}, \max\{t_{vB}, t_{wB}\})$ between $v$ and $w$ joining $B$, multiplied by $-1$ if $w$ joined the block first.

        As with v-shifts, we say that a u-shift is significant if the pairing it corresponds to was in a significant \Pass{} and we call it insignificant otherwise. Since each $u$-shift has value at most $\contractedgapbound$, the contribution from insignificant u-shifts is at most $A_{\text{insig}} \cdot \contractedgapbound$.

        To bound the sum of significant u-shifts, consider a $\ValueChange$ operation to change the value $f(u)$ at time $t$. We want to bound the number of pairings $(v, w, B')$ for which $t \in [\min\{t_{vB}, t_{wB}\}, \max\{t_{vB}, t_{wB}\})$. 
        
        For this to hold, $B$ had to be a growing block of $u$ at time $t$ (as $t_{B+} < t_{vB}, t_{wB} < t_{B-}$ 
        ), thus by \hyperref[list:five]{(5)} there are at most $C_{\text{growing}} = O(1)$ possible blocks $B$.

        Fix one such block $B$ and, furthermore, a pass count $i$. We now claim that there is at most one pairing of the desired kind that is both done inside a block with root block $B$ and in a \Pass{} immediately before which that block $\tilde{B}$ had pass count $i$.

        For this, recall that by \hyperref[list:three]{(3)}, the elements in a block $B$ have strictly decreasing join times $t_{vB}$. Thus, for any time $t$, there is at most one pair of adjacent elements $v, w \in B$ such that $t_{vB} \leq t < t_{wB}$. Thus, there cannot be two pairings of the desired kind in a single \Pass{}.
        
        Since each \Pass{} increases the pass count of a block, if two pairings of the desired kind exist, they must be in distinct split blocks $\tilde{B}_1, \tilde{B}_2$ of $B$. This, however, cannot occur either, as by \hyperref[list:four]{(4)}, at any time, the join time interval of only one of those blocks can contain $t$, and the block containing the time cannot change if it exists, as elements cannot join shrinking blocks.

        Thus, each \ValueChange{} contributes to at most $O(\tau)$ significant u-shifts, and we obtain
        \begin{equation*}
            \sum_{M} \sum_{(v, w, B') \in M} \text{u-shift}(v, w) \leq A_{\text{insig}} \cdot \contractedgapbound + \#\ValueChange{} \cdot O(\tau).
        \end{equation*}
        
        \item \textbf{w-shifts}: The w-shift $f_{wB'}(v) - f_{wB}(u)$ of the pairing $(v, w, B')$ done inside a block with root block $B$ is the difference between the value of the creator of the block $B$ at the moment $w$ joined it and the value of the creator of the next block $B'$ that $w$ joined, at the moment $w$ joined it.

        Note that for any block $B$ with creator $u$ that $w$ joins at some point, the term $f_{wB}(u)$ appears in a $w$-shift positively if and only if $w$ joined $B$ by losing a pairing (as opposed to being \textsc{Insert}ed), and appears in a $w$-shift negatively if and only if $w$ left a block with root block $B$ by losing a pairing (as opposed to being \textsc{Delete}d or remaining in such a block at the end of the process).
        
        Thus, a necessary condition for the term $f_{wB}(u)$ to appear with multiple $1$ in the sum of w-shifts is that the element $w$ was either \textsc{Delete}d out of a block with root block $B$ or remained in $B$ at the end of the process. Each such event contributes to the multiple of a unique term $f_{wB}(u)$, and each term $f_{wB}(u)$ is nonnegative and at most $\contractedgapbound$. Furthermore, the number of \textsc{Delete}s plus the number of elements in a block at the end of the process is precisely the number of \textsc{Insert}s.
        
        Thus, we obtain the bound
        \begin{equation*}
            \sum_{M} \sum_{(v, w, B') \in M} \text{w-shift}(v, w) \leq \text{\#\textsc{Insert}} \cdot \contractedgapbound.
        \end{equation*}
    \end{itemize}
    With these bounds on the sums of the individual shift types, we obtain
    \begin{align*}
        A   &\leq \sum_{M} \sum_{(v, w, B') \in M} \text{v-shift}(v, w) + \text{u-shift}(v, w) + \text{w-shift}(v, w)\\
            &\leq 2 A_{\text{insig}} \cdot \contractedgapbound + \#\Insert{} \cdot O(\contractedgapbound) + \#\ValueChange{} \cdot O(\tau)\\
            &\leq \frac{1}{2} A + \#\Insert{} \cdot O(\contractedgapbound) + \#\ValueChange{} \cdot O(\log \contractedgapbound).
    \end{align*}
    Subtracting $A / 2$ from both sides and multiplying by 2 yields the desired bound.
\end{proof}

It remains to prove \Cref{lemma:insignificant-passes}, which we do now.
\begin{proof}
    First, note that since $A + \#\Insert{}$ is the total number of times elements enter blocks, it is an upper bound on the sum of sizes $|B|$ of blocks over all \Convert{} calls.  

    Consider some block $B$ at the moment \Convert{} is called on it, and let $b = |B|$ at that moment. We will bound the contribution to $A_{\text{insig}}$ from insignificant passes on blocks with root block $B$. To do this, we will specifically bound the total number of elements that at some point appear in some block with root block $B$ that has pass count $i$. Denote this value by $c_i$. We will show by induction on increasing $i$ that
    \begin{equation*}
        c_i \leq b \cdot (1 - \gamma)^{\max\{0, i - C_{\text{exempt}}\}}.
    \end{equation*}
    For $i \leq C_{\text{exempt}}$, the claim is simply $c_i \leq b$, which holds as $b$ is the size of the block $B$ when it is converted, at which point it becomes shrinking and new elements can no longer enter it.

    Suppose now that for some $i > C_{\text{exempt}}$ the claim holds for $i' < i$. We will show it holds for $i$. Consider some block $\tilde{B}$ with root block $B$ on which $\Pass{}$ was called while it had pass count $i - 1$, and let $S_{\tilde{B}}$ be the set of elements in $\tilde{B}$ right before that \Pass{} and $M$ the matching of the \Pass{}. We have $|M| \geq \gamma |S_{\tilde{B}}|$ as $i - 1 \geq C_{\text{exempt}}$, thus at least a $(1 - \gamma)$-fraction of the elements in $S_{\tilde{B}}$ left a block with root block $B$ as a result of the \Pass{}, and thus cannot contribute to $c_{i}$. Additionally, every element that contributed to $c_{i - 1}$ appears in at most one set $S_{\tilde{B}}$, and every element that contributes to $c_{i}$ has to appear in some set $S_{\tilde{B}}$. Thus, $c_{i} \leq c_{i - 1} (1 - \gamma) \leq b (1 - \gamma)^{\max\{0, i - C_{\text{exempt}}\}}$, where the last inequality uses the induction assumption.

    Now, selecting $\tau = C_{\text{exempt}} + \left\lceil\log_{1 / (1 - \gamma)}(4\contractedgapbound)\right\rceil = O(\log \contractedgapbound)$, we have $c_\tau \leq \frac{1}{4G} \cdot b$. Thus, for this threshold $\tau$, as each pairing in an insignificant pass in a block with root block $B$ results in an element leaving the block, the total number of these pairings is at most $\frac{1}{4G} \cdot b$.

    Summing over all root blocks $B$ at the moments \Convert{} was called on them and applying this inequality, we obtain $A_{\text{insig}} \leq \frac{1}{4G} (A + \#\Insert{})$, as desired.
\end{proof}

\subsubsection{Translating operations on lazy intervals into the interface}

Here, we show how to map lazy intervals onto the blocks of the interface, how to translate operations on lazy intervals into the interface, and how to derive \Cref{lemma:bounding-bad-pairings} from \Cref{lemma:interface}.

We maintain the following invariants. The interface elements are the heavy paths in the splay tree, the blocks correspond to the lazy intervals, the creator of a growing block is its owner, and the function $f(v)$ corresponds to the contracted point gap of $v$.

We maintain the invariant that there will be no more than
$C_{\text{growing}}=2+1$ growing blocks for each creator at a time. Two correspond to the fact that each owner has one left and one right growing lazy interval, while we need one additional growing block to perform heap-child-exchange. Note that a lazy interval can change its owner only if it is shrinking, so the creator correctly never changes.

We further set $\contractedgapbound' = \contractedgapbound + 1$, one more than the maximum contracted point gap ever.

We now show how to map individual operations on lazy intervals while maintaining these invariants.
For clarity, we use $\ValueIncrease$ for $\ValueChange$ with $\Delta=1$ and $\ValueDecrease$ for $\ValueChange$ by $-1$.

\begin{itemize}
    \item \textbf{Delete}. We perform a delete by \Delete{}.
    If the heap-child with minimum gap was deleted from a lazy interval, the block gap might increase, which decreases point gaps and might also decrease contracted point gaps. In such a case, the \Delete{} is followed by \ValueDecrease{}s.
    \item \textbf{Insert}. To satisfy the conditions of \Insert{}, we might need to do multiple \ValueIncrease{}s of the value of the inserted element so that it is larger than the value of the creator of the growing block it is inserted into.
    If this would increase the value over $G$, we \ValueDecrease{} by one the value of the element to which block is being \Insert{}ed, after which we \ValueIncrease{} back.
    Then we use \Insert{}, followed by possibly multiple \ValueIncrease{}s or \ValueDecrease{}s to match the new contracted point gap.
    \item \textbf{Split}. We use \Split{}. Afterwards, we need to update the values of elements in the respective blocks to match their new contracted point gaps, but since these values only decrease, we use multiple \ValueDecrease{}s.
    \item \textbf{Convert}. We use \Convert{}. This decreases the contracted point gaps of the heap-children in the lazy interval, which we track by \ValueDecrease{}s. Then \Create{} is called to create the block corresponding to the new growing lazy interval.
    \item \textbf{Transfer}. Since the blocks do not care about the owners except for their creators, no operation on the block is needed for the transfer itself. However, we need to track the conversion of the lazy intervals by \Convert{}s, \ValueDecrease{}s, and \Create{}s.
    \item \textbf{Heap-child-exchange}. The interface does not allow insertion into shrinking blocks. Therefore, we \Delete{} the previous heap-child, \Create{} a new growing block, and do an \Insert{} of the new heap-child. As with inserts, we might need to do multiple \ValueIncrease{}s to satisfy the conditions of \Insert{}.
    Again, the increases are correct, as we never increase to more than $\contractedgapbound' = \contractedgapbound+1$.
    Afterwards, the new block is \Convert{}ed, and we apply \ValueDecrease{} to the value of its element so that it corresponds to the new contracted point gap.
    \item \textbf{Pair-up}. We perform pair-ups corresponding to internal pairings in batches. We know from \Cref{lemma:pairings-in-parallel} that no two first pairings share a node and, similarly, no two second pairings share a node. All internal pair-ups for one lazy interval during one splay can be split into two disjoint groups in which pair-ups can be performed simultaneously. Hence, we perform them by two \Pass{}es.
    Since they are disjoint, the condition of \Pass{} that an element appears in at most one pairing is satisfied.
    From \Cref{lemma:internal-pairings-fraction}, we know that, for an unbroken lazy interval $\fI$, if $|\fI| \geq 10$, then there are at least $(|\fI|-6)/4 \geq |\fI|/10$ internal first pairings and internal second pairings for $\fI$.
    If $|\fI|\leq 10$, then either there are no internal pairings or at least one, hence at least $|\fI|/10$.
    After the first splay over the lazy interval, which translates to $2$ \Pass{}es over the corresponding block, by \Cref{lemma:unbroken-lazy-intervals}, the lazy interval is not broken.
    Therefore, after the first $2$ \Pass{}es over a block, the set of pairings has size a constant fraction of the block, so this requirement of the \Pass{} is satisfied.
    Finally, the requirement that, for every pairing, the value of the loser is greater than that of the winner is satisfied for bad internal pairings. This is not true for good internal pairings, where the values are equal---for each good internal pairing, before the \Pass{}, we need to do one \ValueDecrease{} on the winner, followed by one \ValueIncrease{} after the \Pass{}.
    After the \Pass{}, the values of the losers need to be updated, but these can be handled by \ValueDecrease{}s as the contracted point gap after an internal pairing always decreases by \Cref{lemma:internal-pairings-gap}.
\end{itemize}

\begin{lemma}\label{lemma:bounding-inserts}
    The total number of \Insert{}s on the interface is at most $p + n$, where $p$ is the total number of paid operations on the lazy intervals.
\end{lemma}

\begin{proof}
    For every insert and heap-child-exchange, we do one \Insert{}. No other operation does \Insert{}s. At the very beginning, every heap-child needs to be \Insert{}ed into the block corresponding to its lazy interval, which necessitates at most $n$ \Insert{}s.
\end{proof}

\begin{lemma}\label{lemma:bounding-increases}
    The total number of \ValueIncrease{}s on the interface is $O(g + p\contractedgapbound)$, where $p$ is the total number of paid operations on the lazy intervals and $g$ is the total number of good internal pairings.
\end{lemma}

\begin{proof}
    For every good internal pairing, we do one \ValueIncrease{}.
    For every insert and heap-child-exchange, we need to do $O(\contractedgapbound)$ \ValueIncrease{}s to satisfy the conditions of \Insert{}. No other operation does \ValueIncrease{}s.
\end{proof}

\begin{lemma}\label{lemma:bounding-decreases}
    Let $r$ be the total number of \ValueIncrease{}s on the interface and $s$ be the total number of \ValueDecrease{}s on the interface.
    Then $s \leq r + n\contractedgapbound$.
\end{lemma}

\begin{proof}
    The values $f$ can never be negative. They start at values of at most $\contractedgapbound$, so their sum is at most $n\contractedgapbound$ and is always nonnegative. Every \ValueDecrease{} decreases the sum by one and every \ValueIncrease{} increases the sum by one, so there can be no more than $r + n\contractedgapbound$ \ValueDecrease{}s.
\end{proof}

Now we are ready to bound the total number of internal pairings by proving \Cref{lemma:bounding-bad-pairings}.

\begin{proof}[Proof of \Cref{lemma:bounding-bad-pairings}.]
    Since every internal pairing is performed by a pair-up, and pair-ups are translated into \Pass{}es, the total activity $A$ from \Cref{lemma:interface} is exactly the total number of internal pairings.

    This leaves us with bounding the cost $C$ from \Cref{lemma:interface} as $A \in O(C)$.
    The total number of \ValueIncrease{}s is $O(g + p\contractedgapbound)$ from \Cref{lemma:bounding-increases} and by \Cref{lemma:bounding-decreases}, the total number of \ValueDecrease{}s is $O(g + (p + n)\contractedgapbound)$. Therefore the total number of \ValueChange{}s is also  $O(g + (p + n)\contractedgapbound)$. As the total number of inserts is $O(p + n)$ by \Cref{lemma:bounding-inserts}, the total cost is $O((p + n)\contractedgapbound\log \contractedgapbound + g\log \contractedgapbound)$. This is also asymptotically the total activity and the total number of internal pairings, which concludes the proof.
\end{proof}

\subsection{Putting everything together}

In this section, we bound the total number of pairings that can occur and use it to bound the work of the splay tree.

\begin{lemma}\label{lemma:number-of-blocks}
    Let $n$ be the number of elements. 
    The number of lazy intervals ever created is $O(n + \cost(T^*))$.
\end{lemma}

\begin{proof}
	    At the very beginning, we need to create two lazy intervals for each of the $2n + 1$ heavy paths, giving $O(n)$ lazy intervals.
	    For every splay, we create $O(1)$ lazy intervals by \Cref{lemma:blocks-by-splay}. Since there are $m$ splays, this gives $O(m)$ lazy intervals, which is bounded by $O(\cost(T^*))$.
	    Finally, for every rotation of the nearly optimal tree $T^*$, we create $O(1)$ new lazy intervals by \Cref{lemma:rotation-lazy-intervals}. There are at most $\cost(T^*)$ rotations, so this gives $O(\cost(T^*))$ lazy intervals.
    Summing these terms gives the desired bound.
\end{proof}

\begin{lemma}\label{lemma:cost-of-everything}
	   The total number of paid operations is at most $u + O((n + \cost(T^*))\log \contractedgapbound)$.
\end{lemma}

\begin{proof}
For every rotation of the nearly optimal tree $T^*$, there are $O(1)$ paid operations by \Cref{lemma:rotation-lazy-intervals}. This gives $O(\cost(T^*))$ paid operations over all rotations.

Let $\ell$ be the total number of lazy intervals ever created. The total number of paid operations for all splays is $u + O(\ell \log \contractedgapbound + \cost(T^*))$ by \Cref{lemma:indels-by-splay}.
By \Cref{lemma:number-of-blocks}, $\ell \in O(n + \cost(T^*))$, which gives
$u + O((n+\cost(T^*)) \log \contractedgapbound)$ paid operations over all splays and also overall.
\end{proof}

\begin{lemma}\label{lemma:total-good-pairings}
	Let $u$ be the number of unimportant boundary pairings performed. There are at most $u \contractedgapbound + O((n + \cost(T^*))\contractedgapbound\log \contractedgapbound)$ good internal pairings.
\end{lemma}

\begin{proof}
    We investigate which operations change the total contracted point gap of all heap-children.

	    Deletes, splits, converts, transfers, and heap-child-exchanges either decrease the contracted point gaps of individual heap-children or leave them unchanged.
	    Every pair-up corresponds to an internal pairing and, by \Cref{lemma:internal-pairings-gap}, does not increase contracted point gaps.
	    Inserts are the only operations that can increase contracted point gaps, and every insert increases them by at most $O( \contractedgapbound)$.
	    Therefore, the total increase in the sum of contracted point gaps over the whole access sequence is $u\contractedgapbound + O((n + \cost(T^*))\contractedgapbound\log \contractedgapbound)$.
	    At the very beginning, this sum cannot be larger than $O(n \contractedgapbound)$.

	    Every good pairing decreases the total contracted point gap by at least one by \Cref{lemma:internal-pairings-gap}.
	    Since the total contracted point gap is always non-negative, the number of good pairings is at most $u\contractedgapbound + O((n + \cost(T^*))\contractedgapbound\log \contractedgapbound) + O(n\contractedgapbound)$, which is $u\contractedgapbound + O((n + \cost(T^*))\contractedgapbound\log \contractedgapbound)$.
\end{proof}

In particular, substituting the number of good internal pairings into \Cref{lemma:bounding-bad-pairings} gives the following bound, where $a$ is the total number of internal pairings and $p$ is the total number of paid operations:
\[
a \in O(p\contractedgapbound \log \contractedgapbound + (n +\cost(T^*))\contractedgapbound\log^2 \contractedgapbound)
\]

\begin{lemma}\label{lemma:bound-unimportant-pairings-by-internal}
		    Let $n$ be the number of elements, $g$ the total number of good internal pairings, $a$ the total number of internal pairings (both good and bad), and $p$ the total number of paid operations.
	    Let $\alpha$ be the proportionality constant from \Cref{lemma:bounding-bad-pairings}. We rewrite the bound as $a \leq \alpha \cdot (p\contractedgapbound \log \contractedgapbound + (n +\cost(T^*))\contractedgapbound\log^2 \contractedgapbound) + O(1)$.
    
	    Let $u$ be the total number of unimportant pairings. Then, for $\kappa= \max(16, 32\alpha)$ in the definition of unimportant pairings, we have $\alpha u\contractedgapbound\log \contractedgapbound \leq a/2$ and $u \leq a$.
\end{lemma}

\begin{proof}
	    Consider an unimportant boundary pairing, and let $\fI$ be the lazy interval of the loser.
	    By the definition of unimportance, the lazy interval had at least $\kappa \contractedgapbound \log \contractedgapbound$ heap-children. This is at least $16\contractedgapbound\log \contractedgapbound \geq 16 \cdot 2 \log 2 > 12$ heap-children. Therefore, by \Cref{lemma:internal-pairings-fraction}, the number of internal pairings on $\fI$ in this splay is at least $(\kappa \contractedgapbound \log \contractedgapbound - 6) / 2 \geq \kappa \contractedgapbound \log \contractedgapbound / 4$.
	    Note that, by the definition of unimportance, $\fI$ is not broken, so the assumptions of \Cref{lemma:internal-pairings-fraction} are satisfied.
	    There are at most $4$ different unimportant pairings whose loser is in $\fI$ and that correspond to the same splay, as there are at most $2$ zig-zigs on the boundary of $\fI$, each consisting of at most $2$ boundary pairings.
	    Summing over all lazy intervals and all splays, we get $a \geq \kappa \contractedgapbound\log \contractedgapbound \cdot u / 16 = \max(16, 32\alpha) \contractedgapbound\log \contractedgapbound \cdot u / 16$.
	    Hence, $a/2 \geq \alpha u\contractedgapbound\log \contractedgapbound$ and $a \geq u$.
\end{proof}

\begin{lemma}\label{lemma:total-zig-zigs}
    Let $n$ be the number of elements.
	    The total number of zig-zigs performed by the splay tree over the whole access sequence is $O((n + \cost(T^*)) \cdot \log\log n \cdot \log^2\log\log n)$.
\end{lemma}

\begin{proof}
	    Every zig-zig corresponds to two pairings; therefore, it is sufficient to bound the total number of pairings.
	    We need to bound three types of pairings: important boundary pairings, unimportant boundary pairings, and internal pairings.

	    There are only $O(\log \contractedgapbound)$ important boundary pairings per lazy interval by \Cref{lemma:outer-ops-block}, and there are $O(n + \cost(T^*))$ lazy intervals ever created by \Cref{lemma:number-of-blocks}. Thus, there are $O((n + \cost(T^*))\log \contractedgapbound)$ important boundary pairings.

	    To bound the number of internal pairings, let $\alpha$ be the proportionality constant from \Cref{lemma:bounding-bad-pairings}, i.e.,
    \(a \leq \alpha \cdot (p\contractedgapbound \log \contractedgapbound + (n +\cost(T^*))\contractedgapbound\log^2 \contractedgapbound) + O(1) \).
	    By \Cref{lemma:cost-of-everything}, $p$, the total number of paid operations, is at most $u + O((n + \cost(T^*))\log \contractedgapbound)$. Substituting this bound gives
    \[a \leq \alpha \cdot (u\contractedgapbound\log \contractedgapbound + O((n + \cost(T^*))\contractedgapbound \log^2 \contractedgapbound)) + O(1) \]
	    Furthermore, by \Cref{lemma:bound-unimportant-pairings-by-internal},
    $\alpha u\contractedgapbound\log \contractedgapbound \leq a/2$, hence
    \[ a \leq \alpha\left(\frac{a}{2\alpha} + O((n + \cost(T^*))\contractedgapbound\log^2 \contractedgapbound)\right) + O(1)\]
    \[ a/2 \leq \alpha\left(O((n + \cost(T^*))\contractedgapbound\log^2 \contractedgapbound)\right) + O(1)\]
    \[ a \in O((n + \cost(T^*))\contractedgapbound\log^2 \contractedgapbound)\]
    This bounds $a$, the total number of internal pairings.

	    By \Cref{lemma:bound-unimportant-pairings-by-internal}, the number of unimportant boundary pairings satisfies the same bound, as it is at most the number of internal pairings.

	    Summing gives $O((n + \cost(T^*))\contractedgapbound\log^2 \contractedgapbound)$ pairings, and thus also zig-zigs. Since the nearly optimal tree $T^*$ has depth $O(\log n)$ by \Cref{lemma:opt-constraints}, there are at most $O((n +  \cost(T^*)) \cdot \log\log n \cdot \log^2\log\log n)$ zig-zigs.
\end{proof}

\begin{lemma}\label{lemma:total-zig-zags}
    Let $n$ be the number of elements and $m$ the length of the access sequence.
	    The total number of zig-zags performed by the splay tree over the whole access sequence is $O((n + \cost(T^*)) \cdot \log\log n \cdot \log^2\log\log n)$.
\end{lemma}

\begin{proof}
	    By \Cref{lemma:bends-by-splay}, every splay creates $O(1)$ + \#zig-zigs bends and removes \#zig-zags bends. Every rotation of the nearly optimal tree $T^*$ creates $O(1)$ bends by \Cref{lemma:heavy-paths-rotation}.
	    Over all splays, the total number of bends created is the total number of zig-zigs plus $O(\cost(T^*))$, and initially there can be at most $O(n)$ bends. Therefore, the total number of bends created is $O((n + \cost(T^*)) \cdot \log\log n \cdot \log^2\log\log n)$ by \Cref{lemma:total-zig-zigs}.
    Since every zig-zag removes one bend and the number of bends is non-negative, there can be at most $O((n + \cost(T^*)) \cdot \log\log n \cdot \log^2\log\log n)$ zig-zags.
\end{proof}

We have now bounded the numbers of zig-zigs and zig-zags, so we can put everything together to prove the competitiveness of splay trees.

\begin{proof}[Proof of \Cref{thm:main}]
		    There are at most $m \in O(\cost(T^*))$ zigs, as they occur only once per access, and each performs $1$ rotation. There are also $O((n + \cost(T^*)) \cdot \log\log n \cdot \log^2\log\log n)$ zig-zigs and zig-zags, each performing $O(1)$ rotations.
	    Hence, the overall cost of the splay tree is $O((n +  \cost(T^*)) \cdot \log\log n \cdot \log^2\log\log n)$.
	    This relates the cost of the splay tree to that of the nearly optimal tree $T^*$.
	    To relate it to the cost of the optimal tree $T^\text{OPT}$,
	    we recall that by \Cref{lemma:opt-constraints}, $\cost(T^*) \in O(\cost (T^\text{OPT}))$.
    Therefore, the overall cost of the splay tree is $O((n +  \cost(T^\text{OPT})) \cdot \log\log n \cdot \log^2\log\log n)$,
    which proves the main theorem.
\end{proof}

\section{Conclusions}

In this work we prove that splay trees are $\tilde{O}(\log \log n)$-competitive, giving the first $o(\log n)$ bound for a binary search tree that is conjectured to be dynamically optimal.
Among other consequences, this implies that splay trees satisfy the unified bound and the lazy finger bound with a $\tilde{O}(\log\log n)$ multiplicative factor and an additive $\tilde{O}(n\log\log n)$ term; see \Cref{sec:related-work}.
We also show that our techniques can be used to nearly prove the split conjecture: splitting a splay tree costs $\tilde{O}(n\log\log n)$, improving on the previously best trivial bound of $O(n \log n)$; see \Cref{app:split-conjecture}.
Our proof of competitiveness is novel in that we charge the cost of the splay tree directly to the optimal tree, circumventing known gaps between lower bounds and the optimum.

Our work used several novel techniques that we believe can be useful in other contexts.
First, the normal form we used to assume the optimal tree is ribless, has logarithmic depth, and does restricted rotations (\Cref{S:normal-form}) can be used as a black box whenever a tree is required to have these properties.
Second, we abstracted the core of the complexity proof of pairing heaps~\cite{Sinnamon2025,PairingHeaps2026} into an interface that can be used in the analysis of other data structures beyond splay trees and multiple types of pairing heaps.
Third, our result that all zig-zags ever done can be asymptotically bounded by the total number of zig-zigs can be used as a black box.
Therefore, future work can focus on bounding the number of zig-zigs only.

\ifsubmission\else
\paragraph{Acknowledgements.}
BH, RH, and AR were partially funded by the Ministry of Education and Science of Bulgaria's support for INSAIT as part of the Bulgarian National Roadmap for Research Infrastructure and partially funded through the European Research Council (ERC) under the European Union's Horizon Europe research and innovation program (ERC Advanced grant agreement 101268062 and ERC Starting grant agreement 949272). 
PC, MK and VR were partially supported by Charles Univ. project UNCE 24/SCI/008.
PC was partially supported by Charles Univ. GA UK project No. 70924.
VR was supported by the Czech Science Foundation (GAČR), grant no. 26-23599M. 
RT's research at Princeton was partially supported by a gift from Microsoft.
Part of this work was done while most authors were visiting INSAIT.
\fi

\printbibliography

\appendix
\section{Related work}
\label{sec:related-work}

\subsection{Dynamic optimality and related conjectures and properties}

\paragraph{Dynamic optimality.}
A BST algorithm is dynamically optimal if, on every sufficiently long access sequence, its cost is within a constant factor of that of the optimal offline BST algorithm.
There are two main candidates which are conjectured to be dynamically optimal -- splay trees~\cite{SleatorTarjan85} and GREEDY~\cite{Lucas88Canonical,Munro00LinearSearch,DemaineHIKP09Geometry}; see the surveys~\cite{Iacono13InPursuit,Russo19Study,Kozma16Thesis}.
Dynamic optimality is an example of \emph{instance optimality}---the requirement that
an algorithm's cost on every input be within a constant factor of the best algorithm for
that particular input.
This notion was formalized by Fagin, Lotem, and
Naor~\cite{FaginLotemNaor03InstanceOptimality} and has since become a central paradigm in beyond-worst-case analysis~\cite{Roughgarden2021book}. Instance optimality can be seen as a subfield of parameterized algorithms~\cite{cygan2015parameterized}.
The closest classical analogue to BST dynamic optimality is the \emph{list update problem},
where Sleator and Tarjan~\cite{SleatorTarjan85ListUpdate} showed that the simple
Move-to-Front rule is constant-competitive against the offline optimum.
Instance-optimality and related notions have been established across many domains: geometric algorithms~\cite{AfshaniBarbayChan17InstanceOptimalGeometry}, distributed computing~\cite{HaeuplerWajcZuzic21UniversalDistributed,haeupler_racke_ghaffari2022hopconstrained_expanders,rozhon_grunau_haeupler_zuzic_li2022deterministic_sssp,10.1145/3519270.3538429}, sorting-related problems~\cite{HaeuplerEtAl24UniversalDijkstra,HoogRotenbergRutschmann25SimplerDijkstra,haeupler_hladik_iacono_rozhon_tarjan_tetek2025fast_supi}, and sublinear algorithms~\cite{HaeuplerEtAl25BidirectionalDijkstra,valiant2017automatic,valiant2016instance,dagum_karp_luby_ross2000optimal_sampling,narayanan_rozhon_tetek_thorup2024instance_optimal_sampling}.

\paragraph{Splay trees and conjectures related to dynamic optimality.}
In 1985, Sleator and Tarjan~\cite{SleatorTarjan85} introduced splay trees and stated the celebrated dynamic optimality conjecture for them.
In restricted settings for the optimal tree, certain results on dynamic optimality are known.
Georgakopoulos~\cite{Georgakopoulos04Reweighing} proved that splay trees are competitive
against a class of dynamically balanced BSTs using a reweighing argument.
Splay trees are also \emph{key-independently optimal}~\cite{Iacono05KeyIndependent} --
under a random relabeling of keys, they match the expected offline optimum up to a constant.
Furthermore, much of the evidence around dynamic optimality comes from restricted input families where splay trees are conjectured, expected, or proved to be optimal or near-optimal.

\begin{itemize}
\item \emph{static optimality}: Sleator and Tarjan~\cite{SleatorTarjan85} proved the \emph{access lemma} for splay trees, which implies that splay trees are within a constant factor of any static tree, i.e., a tree that is not allowed to change. This also proves $O(\log n)$ amortized access time.
\item \emph{sequential access}, where keys are accessed in sorted order;
Tarjan proved that sorted access takes
linear total time in splay trees, regardless of the initial tree~\cite{Tarjan85SequentialAccess},
and Elmasry later simplified the argument~\cite{Elmasry04SequentialDeque}.
\item \emph{traversal sequences}, typically preorder or postorder traversals of a BST; Levy and Tarjan obtained optimality of splaying for preorder and postorder sequences in
special settings~\cite{LevyTarjan19PrePost}.
\item \emph{deque conjecture}, where updates happen only at the two extremes. It is conjectured that a sequence of $m$ updates takes total time $O(m)$; Sundar showed that the total time is $O(m \alpha (m))$~\cite{Sundar92Deque}, and Pettie improved this to $O(m \alpha^* (m))$~\cite{Pettie08Deque}, where $\alpha$ is the inverse Ackermann function and $\alpha^*$ is the iterated inverse Ackermann function.
\item \emph{split conjecture}, where the tree is split into two at the accessed element. It is conjectured that splitting all $n$ elements takes total time $O(n)$. Prior to this work, no bound was known beyond the $O(n\log n)$ bound implied by the access lemma and the $O(n\alpha(n))$ bound of Lucas for the special case in which the initial tree is a path~\cite{Lucas91Split}.
\end{itemize}

\paragraph{The GREEDY algorithm.}
A second major candidate for dynamic optimality is the maximally-greedy BST strategy,
introduced independently by Lucas~\cite{Lucas88Canonical} and Munro~\cite{Munro00LinearSearch}
and reinterpreted via the geometric view in~\cite{DemaineHIKP09Geometry}.
Progress on GREEDY includes access-lemma-type analyses~\cite{Fox11MaximallyGreedy},
bounds for decomposable sequences~\cite{GoyalGupta19Decomposable},
near-optimality on deque sequences~\cite{ChalermsookG0MS15GreedyDeque},
improved bounds on pattern-avoidance access families~\cite{ChalermsookEtAl22PatternAvoidanceGreedy},
and group access bounds~\cite{ChalermsookEtAl24GroupAccess}.
However, recent work shows nontrivial limitations: the competitive ratio of GREEDY is at
least~2, and the additive gap from $\OPT$ can be $\Omega(m\log\log n)$~\cite{SadehKaplan23GreedyFuture}, where $m$ is the length of the accessed sequence.

\paragraph{Properties implied by dynamic optimality.}
There are several properties of binary search trees that hold for any dynamically optimal binary search tree. Furthermore, if the binary search tree is $r$-competitive up to an additional factor of $r'$, these hold amortized up to a multiplicative factor of $r$ and initialization cost $r'$.
All of these are implied by the existence of other BSTs with these properties, which on any sequence can be related to by the dynamically optimal BST. Note that these are implied only if there are no insertions and deletions.
\begin{itemize}
    \item \emph{working-set bound}: access to an element $x$ takes time $O(\log(2 + w_t))$, where $w_t$ is the number of distinct elements accessed since time $t$ and $t$ is the time of last access to $x$. This property is proved for splay trees by Sleator and Tarjan~\cite{SleatorTarjan85}.
    \item \emph{dynamic-finger bound}: access to $x_i$ takes time $O(\log(2 + |x_i - x_{i-1}|))$; this was also proved for splay trees~\cite{ColeMSS00,Cole00}.
    \item \emph{unified bound}~\cite{BadoiuCDI07Unified}: a combination of the working-set bound and dynamic-finger bound: access to an element $x_i$ takes time $O(\min_{t<i} \log(2 + w_t + |x_t - x_i|))$; skip-splay trees achieve it amortizedly with additive $O(\log\log n)$ per query~\cite{DerryberrySleator09SkipSplay}, and cache-splays achieve it amortizedly~\cite{Derryberry09Thesis}.
    \item \emph{lazy finger bound} and \emph{weighted finger bound}: both are achieved by the structure of Iacono and Langerman~\cite{IaconoLangerman16WeightedDynamicFinger}.
\end{itemize}
A systematic landscape of BST bounds is surveyed
in~\cite{ChalermsookEtAl16Landscape,ChalermsookEtAl24GroupAccess}.

\paragraph{$O(\log\log n)$-competitive BSTs.}
While dynamic optimality is open, online BSTs with provably sublogarithmic competitive
ratios do exist.
Tango trees were introduced in~\cite{DemaineHIP07Tango} and achieve an
$O(\log\log n)$ competitive ratio by charging cost to Wilber's interleave bound.
Wang, Derryberry, and Sleator~\cite{WangDS06MultiSplay,Sleator04MultiSplayTR} gave
\emph{multi-splay trees}, also $O(\log\log n)$-competitive, that additionally retain
several splay-like properties~\cite{DerryberrySW09PropertiesMST}.
Georgakopoulos~\cite{Georgakopoulos08ChainSplay} achieved $O(\log\log n)$-competitiveness
via \emph{chain-splay trees}, a splay-based variant that is the closest proven competitive
guarantee to the original splay conjecture among splay-like variants.
Zipper trees were introduced in~\cite{BoseDouiebDujmovicFagerberg10Zipper},
also $O(\log\log n)$-competitive, with the additional guarantee of $O(\log n)$
worst-case time per access.
A meta-transformation in~\cite{DemaineEtAl13CombiningBSTs} combines
the guarantees of multiple BST algorithms into a single structure within a constant factor
of the best.
Generic transformations can de-amortize essentially any BST algorithm while preserving
asymptotic costs~\cite{BoseCFLa12Deamortizing}, implying that if dynamic optimality is
achievable, it is achievable with worst-case $O(\log n)$ time per access.

\paragraph{Lower bounds and the geometric view.}
Wilber~\cite{Wilber89} gave two influential lower bounds on $\OPT(X)$.
The first (\emph{interleave} or \emph{alternation bound}, often called Wilber-1) underlies
Tango trees and most known $o(\log n)$-competitive BST algorithms.
The second (\emph{funnel bound}, Wilber-2) is traditionally viewed as a more promising
candidate for tightly characterizing $\OPT$.

The geometric interpretation of BST executions from~\cite{DemaineHIKP09Geometry}
maps an access sequence to a planar point set and expresses $\OPT$ as a minimum
\emph{arborally satisfied} completion, yielding additional lower bounds such as the
\emph{Independent Rectangle Bound} (IRB) that subsumes both Wilber bounds.
Kozma and Saranurak~\cite{KozmaSaranurak20SmoothHeaps,KozmaSaranurak16Rectangulations} connected
BST bounds to rectangulations and smooth heaps via a duality between BSTs and heaps.
Pattern-avoidance and decomposition frameworks explain broad subclasses of sequences where
BST costs are nearly linear~\cite{ChalermsookGKMS15PatternAvoiding,ChalermsookEtAl16Landscape}.

Recent advances clarified the relationships among these bounds.
Lecomte and Weinstein~\cite{LecomteWeinstein20Wilber} proved that the funnel bound dominates
the alternation bound for all sequences and exhibited a tight $\Theta(\log\log n)$
separation between them.
Chalermsook, Chuzhoy, and Saranurak~\cite{ChalermsookCS23StrongWilber1} showed that even a
strengthened ``Strong Wilber-1'' benchmark can be $\Omega(\log\log n / \log\log\log n)$
smaller than $\OPT$, while also giving $O(\log\log n)$-competitive and approximation
algorithms within their framework.
This is a barrier result: improvements in competitive ratio cannot come merely from
tracking the alternation bound; something fundamentally different
is required.
Direct-sum theorems for Wilber's bounds were proved in~\cite{JiangLWY24HardnessAmplification},
which also used them for hardness amplification, showing as a corollary that tango trees
are optimal among BST algorithms whose analyses charge cost to the alternation bound.

\subsection{Pairing heaps}

In this section we describe pairing heaps~\cite{Fredman1986PairingHeap} and recent breakthroughs in their analysis. While seemingly unrelated, we show that, when related to the optimal tree, the changes to the structure of splay trees due to splaying resemble what happens in pairing heaps when the root is deleted.

In general, pairing heaps maintain a heap structure with lower-keyed values closer to the root and implement four operations: (i) \emph{meld}, which takes two heaps and hangs the root with the higher key below the other root as its leftmost child; (ii) \emph{insert}, which creates a singleton heap and melds it with the target heap; (iii) \emph{decrease-key}, which takes a pointer to a node, cuts its subtree, and melds the subtree with the rest of the heap; and (iv) \emph{delete-min}, which removes the root and creates a single tree in two steps, first melding pairs of neighboring children and then assembling the remaining children by a series of melds. Different assembly orders yield different variants of pairing heaps.

In the following variants of pairing heaps, delete-min always takes time $O(\log n)$~\cite{Fredman1986PairingHeap}, and decrease-key requires $\Omega(\log \log n)$ amortized time on certain instances~\cite{Fredman1999EfficiencyPairingHeaps}.

\begin{itemize}
    \item \textbf{Standard pairing heaps.} The assembly is performed oldest-to-newest, with the current vertex always melded with the lowest-key child seen so far. Pettie~\cite{Pettie2005} proved that all other operations except for delete-min take amortized time $O(2^{2\sqrt{\log\log n}})$.
    \item \textbf{Pure pairing heaps.} The assembly identifies the globally minimum child after the pairing phase and then melds all other children with it. While this discards comparisons, it is easier to analyze, and it was shown in~\cite{PairingHeaps2026} that decrease-key takes $O(\log \log n \log\log\log n)$ amortized time.
    \item \textbf{Multi-pass pairing heaps.} Instead of an assembly similar to that in pure or standard pairing heaps, the pairing pass is repeated as long as more than one child remains. Sinnamon and Tarjan~\cite{Sinnamon2025} proved that decrease-key in multi-pass pairing heaps takes $O(\log\log n \log\log\log n)$ amortized time.
\end{itemize}

We extract from the analyses of pure and standard pairing heaps a common interface, proved in \Cref{app:interface}, that we also use in our analysis of splay trees.

We additionally remark that the $O(2^{2\sqrt{\log\log n}})$ analysis for standard pairing heaps~\cite{Pettie2005} could also simplify the proof of splay-tree competitiveness, at the cost of worsening the resulting competitive ratio from $\tilde{O}(\log\log n)$ to $2^{O(\sqrt{\log\log n})}$.

\section{Equivalence of the root BST model to the standard BST model}\label{app:bst-models}

In this appendix, we briefly note the equivalence of the root binary search tree (BST) model and the standard BST model of Sleator and Tarjan~\cite{SleatorTarjan85}.
Another widely used pointer-based BST model, due to Demaine et al.~\cite{DemaineHIP07Tango}, is also equivalent to the standard BST model~\cite{DemaineHIP07Tango}.

In all of these models, a tree $T$ and an access sequence $X = x_1,\dots,x_m$ are given.

\paragraph{Standard BST~\cite{SleatorTarjan85}.} 
The cost of each access $x_i$ is $1 + \depth_T(x_i)$, where the depth is measured at the time of access.
Between accesses, the tree can perform any number of rotations, and each rotation costs $1$.

\paragraph{Root BST (used in this work).}
For every access, the accessed element must be the root of $T$, and the cost of the access is $1$.
In addition, before each access, the tree can perform any number of rotations, and each rotation costs $1$.

\begin{lemma}
	    For every sequence $X$ and every binary search tree algorithm $T$ in the Standard BST model, there is a modified algorithm $T'$ in the Root BST model such that $\cost^{\text{root}}(T', X) \leq 2\cost^{\text{standard}}(T, X)$.
		    Furthermore, for any binary search tree $T'$ in the Root BST model,
		    $\cost^{\text{standard}}(T', X) = \cost^{\text{root}}(T', X)$.
		    Here, $\cost^{\text{standard}}$ denotes the cost in the Standard BST model and $\cost^{\text{root}}$ the cost in the Root BST model.
\end{lemma}

\begin{proof}
	    An algorithm that is valid in the Standard BST model might not be valid in the Root BST model. To fix this, we always rotate the accessed element to the root before the access in the Root BST model. This takes
	    $\depth_T(x_i)$ rotations.
	    After the access, we need to restore the state that the tree would have had in the Standard BST model. We recover this state by undoing the rotations in the reverse order from that in which they were performed, as a left rotation undoes a right rotation.
	    Therefore, in the Root BST model, we perform all rotations performed in the Standard BST model, as well as an additional $2\depth_T(x_i)$ rotations for every accessed element.
	    Hence, the cost in the Root BST model is
    \[
    \cost^{\text{root}}(T', X) = m + \#\text{rotations}_\text{T} + \sum_i 2\depth_T(x_i) \leq 2 \cdot (m + \#\text{rotations}_\text{T} + \sum_i \depth_T(x_i)) =\]
    \[=2 \cdot (\#\text{rotations}_\text{T} + \sum_i (1 +\depth_T(x_i))) = 2 \cdot \cost^{\text{standard}}(T, X)
    \]

	    Furthermore, every algorithm valid in the Root BST model is valid in the Standard BST model. Since each access to the root and each rotation costs $1$ in the Standard BST model, the total costs are identical.
\end{proof}

\section{A normal form for binary search tree algorithms}\label{S:normal-form}

We shall prove that any binary search tree algorithm (henceforth, BST algorithm) can be transformed into one with three properties: (i) the tree maintained by the algorithm always has logarithmic depth, (ii) the tree maintained by the algorithm is ribless, and (iii) the algorithm performs only restricted rotations.
By a restricted rotation, we mean a rotation that happens at depth $\OO(1)$.
The simulation incurs an $\OO(1)$ amortized slowdown factor.

\normalform*

Our result combines three known results: (i) any algorithm can be transformed into one that maintains logarithmic depth~\cite{BoseCFLa12Deamortizing}, (ii) there is a binary search tree implementation of a finger search tree~\cite{BrodalR25} that is fully ribless, and (iii) any $n$-node binary search tree can be converted into any other one on the same set of nodes in at most $4n - 8$ rotations.  We briefly discuss these results and how to combine them to obtain our result.

\subsection{Log-depth trees of Bose et al.}

Bose et al.~\cite{BoseCFLa12Deamortizing} give a construction that converts any binary search tree algorithm into an algorithm that maintains a tree of $\OO(\log n)$ depth, with constant-factor amortized slowdown.  
They assume the \emph{finger model} of a BST algorithm: the algorithm maintains a finger at some node, initially the root.
In one step, the algorithm can move the finger to the parent or a child of the indicated node, or do a rotation at the indicated node.  
The algorithm does an access by moving the finger to the node to be accessed.  
It is straightforward to verify that, up to constant factors, this model is equivalent to the rotation model and the subtree replacement model.

The key ingredient of their construction is an implementation of a \emph{biased stack} as a binary search tree.
They call this data structure a \emph{pop-tart}.
A pop-tart supports $\OO(1)$-amortized-time push and pop operations, as well as \emph{weight-biased access}: 
Each node $x$, optionally including null nodes, has a positive real-valued \emph{weight} $w(x)$; the depth of $x$ is $\OO(\log{\min}\{n, W/w(x)\})$, where $W$ is the total weight of all the nodes.  
The nodes must be pushed into the data structure in either decreasing or increasing order.
In the decreasing case, the pushed node must be smaller than all the nodes currently in the stack.
The node might also be accompanied by a subtree of nodes smaller than the pushed node.
The pop-tart operates on such a subtree as a unit, so it does not modify its internal structure and treats it as a null node.
The increasing case is symmetric.

To obtain their simulation, Bose et al. partition the nodes of the binary search tree maintained by the algorithm into \emph{heavy paths}.
They replace each heavy path by a single node and a pair of biased stacks.  
They also replace the path from the root to the finger node by the finger and a pair of pop-tarts (biased stacks).  
With an appropriate choice of node weights, the overall tree depth is logarithmic.  
Simulating a rotation or a movement of the finger takes $\OO(1)$ pushes and pops in the biased stacks.  
One can verify by inspecting the details of their construction that each such push or pop is at the top of a stack that has depth $\OO(1)$
and that each such push or pop takes $\OO(1)$ amortized time in the subtree model of BST algorithms.

The Bose et al. simulation is the first step of our transformation.

\subsection{Ribless trees of Brodal and Rysgaard}

We say that a non-root node of a binary tree is a {\em left-rib node} if it is the left child of its parent. The definition of {\em right-rib nodes} is symmetric.
A {\em left rib} is any maximal path in $T$ consisting only of left-rib nodes. A right rib is defined symmetrically.
A spine is a rib that starts at a child of the root of the tree; all other ribs are {\em proper ribs}.
When using the term rib, we will usually mean a proper rib unless specified otherwise.
For an integer $k$, the tree is {\em $k$-ribless} if all its proper ribs are of size less than $k$.
For an integer $k$, the tree is {\em $k$-fully-ribless} if all its ribs, including spines, are of size less than $k$.
When $k$ is some fixed constant, we might say that the tree is ribless or fully-ribless.

Brodal and Rysgaard~\cite{BrodalR25} give a fully-ribless binary search tree implementation of a finger search tree.  
Their implementation gives $\OO(1)$-time access from any node to its parent and children, to its predecessor and successor in symmetric order, and to the smallest and largest nodes in its subtree in symmetric order.  It also supports the insertion of a node before or after a given node, or the deletion of a given node, in $\OO(1)$ amortized time.  One can verify by inspecting their construction that their implementation is a tree of depth $\OO(\log n)$.  In our application, we only need to perform insertions and deletions at one end of the tree, so we do not need the full power of their construction.

We will obtain riblessness by replacing each rib in the tree maintained by the Bose et al. simulation by a Brodal and Rysgaard finger search tree.  We replace a left rib by a finger search tree in which its increasing symmetric order corresponds to decreasing depth along the rib, and a right rib by a finger search tree in which its increasing symmetric order corresponds to increasing depth along the rib.   This is the second step of our transformation.  The following lemma implies that the tree resulting from these replacements is fully-ribless:

\begin{lemma}\label{L:folding-ribless}
Let $T$ be any tree, and let $T'$ be a tree formed from $T$ by replacing each rib and spine of $T$ by a $k$-fully-ribless tree. Then $T'$ is $3k$-fully-ribless.
\end{lemma}
\begin{proof}
The lemma follows from a simple observation.  
Here, by a rib, we mean a rib or spine.
Let $v_1v_2$ be a non-rib edge in $T$, let $v_1'v_2'$ be the corresponding edge in $T'$, and let $R_1$ and $R_2$ be the ribs in $T$ containing $v_1$ and $v_2$, respectively. Let $R'_1$ and $R'_2$ be the replacements in $T'$ of $R_1$ and $R_2$, respectively.  If $v_1'v_2'$ is a rib edge of $T'$ and this rib contains the root of $R'_1$, then $v_1'$ is the smallest node in $R'_1$ if $v_1'v_2'$ is a left edge and the largest node in $R'_1$ if $v_1'v_2'$ is a right edge.
The observation is true because if $v_1'v_2'$ is a left edge, a traversal that starts at the root of $R_1'$ in $T'$ and follows left edges through descendants will reach the smallest node in $R'_1$ and then traverse $v_1'v_2'$.  The symmetric argument applies if $v_1'v_2'$ is a right edge.

We prove the lemma by a case analysis using the observation. Let $R$ be any rib of $T'$.  If $R$ does not contain an edge corresponding to a non-rib edge of $T$, then $R$ is entirely in the subtree replacing some rib of $T$, so $R$ contains at most $k$ nodes.  Suppose $R$ contains some edge $v_1'v_2'$ corresponding to a non-rib edge $v_1v_2$ in $T$.  Choose such an edge $v_1'v_2'$ with $v_2'$ deepest in $R$.   Without loss of generality, suppose $v_1v_2$ is a left edge.  Node $v_1$ cannot be the root of $T$, since then $v_1' = v_1$ is the root of $T'$, but the root of $T'$ is a one-node rib of $T'$. Let $v_0$ be the parent of $v_1$ in $T$, and let $v_0'v_1'$ be the edge in $T'$ corresponding to $v_0v_1$ in $T$.

Since $v_1v_2$ is a left non-rib edge in $T$, $v_0v_1$ is a right edge in $T$.  Let $R_0$, $R_1$, $R_2$, $R'_0$, $R'_1$, and $R'_2$ be the ribs in $T$ and the corresponding replacing subtrees in $T'$ containing  $v_0$, $v_1$, and $v_2$ respectively.  If $v_0v_1$ is a non-rib edge, $R_0 \neq R_1$; if $v_0v_1$ is a rib edge, $R_0=R_1$.    

The subtrees rooted at $v_2$ in $T$ and at $v_2'$ in $T'$ contain the same nodes, all of which are greater than $v_0$ and less than $v_1$.

Suppose $v_1'v_2'$ is a right edge. By the observation, $R$ cannot contain the root of $R'_1$, since if it did, $v_1'$ would be the largest node in $R'_1$ in symmetric order and $v_2'$ would be larger than $v_1'$ and hence larger than $v_1$, a contradiction.  Hence, $R$ contains only nodes in $R'_1$ and $R'_2$.  Since both $R'_1$ and $R'_2$ are $k$-fully-ribless, $R$ contains at most $2k$ nodes, making the lemma true.

Suppose, on the other hand, that $v_1'v_2'$ is a left edge.  By the observation, $R$ cannot contain the root of $R'_0$, since if it did, $v_0'$ would be the smallest node in $R'_0$ in symmetric order and $v_2'$ would be smaller than $v_0'$ and hence smaller than $v_0$, a contradiction.  Hence, $R$ contains only nodes in $R'_0$, $R'_1$, and $R'_2$.  Since $R'_0$, $R'_1$, and $R'_2$ are $k$-fully-ribless, $R$ contains at most $3k$ nodes, making the lemma true.
\end{proof}

The Brodal and Rysgaard construction replaces each rib (including spines) of a red-black tree by a tree that is $4$-fully-ribless.  By \Cref{L:folding-ribless}, the resulting tree is itself $12$-fully-ribless.  Again by Lemma~\ref{L:folding-ribless}, the tree formed by replacing each rib of a Bose et al. tree by a Brodal and Rysgaard tree is $36$-fully-ribless.

We also require that the rib replacement increase the tree depth by at most a constant factor.  In any $k$-fully-ribless tree, the $i$-th smallest and $i$-th largest nodes in symmetric order have depth at most $\OO(k i)$.  It follows that the tree $T'$ formed by replacing each rib (including spines) in a Bose et al. tree $T$ by a Brodal and Rysgaard tree has $\OO(\log n)$ depth: if $P$ is the path in $T$ from the root to some node $x$, $P$ alternates between subpaths of ribs and single non-rib edges.  The corresponding path $P'$ in $T'$ alternates between subpaths within the replacements of ribs and edges corresponding to the rib edges on $P$.  If a subpath of $P$ that is within a rib contains $i$ edges, then the corresponding subpath in $P'$ contains $\OO(i)$ edges.  Thus, the replacement of the ribs of $T$ increases the depth by at most a constant factor.

We conclude that if we replace each rib of the tree $T$ maintained by the Bose et al. simulation by a Brodal-Rysgaard finger search tree, the resulting tree $T'$ is $36$-ribless and has $\OO(\log n)$ depth.

\subsection{Restricted rotations of Cleary and Taback}

It remains to show that we can simulate the accesses performed by the Bose et al. algorithm using restricted rotations in $T'$ with an amortized $\OO(1)$ slowdown while maintaining $\OO(\log n)$ depth and riblessness. By a restricted rotation, we mean a rotation that happens at depth at most $C$ for some universal constant $C>0$.  For this, we use the result of Cleary and Taback~\cite{ClearyT03} that any $n$-node binary search tree can be converted into any other on the same set of nodes in at most $4n-8$ rotations.  We need the sequence of rotations to preserve riblessness, so we re-prove their result to show that there is a sequence of $4n-4$ restricted rotations that transforms any $n$-node tree into any other while maintaining riblessness.

\begin{lemma}\label{L:restricted-rotations}
Let $T$ and $T'$ be two trees on the same set of nodes.  Then $T$ can be transformed into $T'$ by doing at most $4n-4$ restricted rotations.  If $T$ and $T'$ are both $k$-ribless, each tree in the transformation sequence is $k+1$-ribless.
\end{lemma}
\begin{proof}

Since restricted rotations are reversible, it is enough to show that $2n-2$ restricted rotations are enough to transform any tree into the tree $R$ consisting only of the root and the right spine. Then, we can transform $T$ into $T'$ by transforming $T$ into $R$, transforming $T'$ into $R$, and reversing each rotation in the second transformation as well as their order.

To transform an arbitrary tree $T$ into $R$, (i) do left rotations at the right child of the root until no node on the right spine has a left child.  Then, repeat the following two steps until the tree is $R$: (ii) do right rotations at the left child of the root until the left child of the root has a right child, or there is no left child of the root; (iii) repeatedly do a left rotation at the right child of the left child of the root until there is no right child of the left child of the root.

After (i), no node on the right spine has a left child.  Each rotation in (ii) adds a node to the right spine and preserves the invariant that no node on this spine has a left child.  It follows that when the algorithm stops, the tree is $R$.  Each rotation in (i) and (iii) adds a different node to the left spine; each rotation in (ii) adds a different node to the right spine.  Thus the transformation of $T$ into $R$ takes at most $n-1$ restricted left rotations and at most $n-1$ restricted right rotations.

Suppose that $T$ and $T'$ are $k$-ribless.  To prove that the transformation from $T$ or $T'$ to $R$ maintains $k+1$-riblessness, we consider the effect of rotations in (i), (ii), and (iii) on the proper ribs.

Once a node is the root or is on a spine, it remains the root or on a spine.

A rotation in (ii) rotates at the node that is the left child of the root, say node $v$.  This has no effect on the proper ribs.

A rotation in (i) rotates at the right child of the root.  The rotation moves the left child of the right child of the root from the top of a left proper rib to the top of a right proper rib.  Before the rotation, the top of this rib was not a right child of a node on the left spine; after the rotation, it is.  The rotation has no other effect on the proper ribs.

A rotation in (iii) rotates at the right child of the left child of the root, say node $w$.  The rotation removes node $w$ from the top of a right proper rib and moves it to the left spine; it makes the previous right child of $w$ the right child of a node on the left spine; and it moves the previous left child of $w$ from the top of a left proper rib to the top of a right proper rib, at the same time making this node the right child of a node on the left spine.  Two proper ribs shrink. One right proper rib grows by one node, but at the same time its top changes from a node that was not a child of a node on the left spine to one that is. The rotation has no other effect on the proper ribs.

We conclude that only right proper ribs can grow, and a right proper rib can grow by one node only when its top becomes a child of a node on the left spine.  Once its top is a child of a node on the left spine, it can only shrink.  The lemma follows.
\end{proof}

We use \Cref{L:restricted-rotations} to prove that the Bose et al. simulation can be implemented using restricted rotations and, in turn, that it can still be implemented using restricted rotations if we replace each rib (including spines) by a finger search tree.  Furthermore, the latter simulation maintains $\OO(\log n)$ depth and $37$-riblessness and has an $\OO(1)$ amortized slowdown.

\newcommand{\opt}{\cost(T^\text{OPT}_{x_1,\dots,x_n})}

\subsection{Proof of the normal-form theorem}

\begin{proof}[Proof of \cref{lemma:opt-constraints}]
Consider the optimal BST serving the access sequence $x_1,\dots,x_n$.
After rotating $x_i$ to the root to serve the $i$-th request,
the tree performs a sequence of touches and rotations to get $x_{i+1}$ into the root.
We denote the tree that serves the $i$-th request with root $x_i$ by $T_i$.
Let $n_i$ be the number of touches and rotations performed after the $i$-th request is served and before the $(i+1)$-st request is served.
The cost of the BST is $\opt = \sum_{i=1}^{n-1}n_i$.  

Since the touched nodes are required to form a connected subtree of $T_i$ and all rotations involve only touched nodes,
$T_i$ and $T_{i+1}$ differ in a subtree of size at most $n_i$ containing the root.
By Lemma~\ref{L:restricted-rotations}, we can convert $T_i$ into $T_{i+1}$ using at most $4n_i - 4$ restricted rotations.
This gives a sequence $\mathcal{T}'$ of trees of length at most $4 \opt$ that serves $x_1,\dots,x_n$ and in which all rotations happen at depth at most $O(1)$.

This sequence of trees can be further represented by the finger trees of Bose et al.
Before each rotation, we place the finger at the root of the current tree, move the finger to the node we want to rotate,
perform the rotation, and move the finger back to the root.
As each rotation happens at constant depth, this creates a sequence of finger-move and rotation operations of length $\OO(\opt)$.
The finger trees of Bose et al. maintain the depth $\OO(\log n)$.

Each finger move or rotation on the Bose et al. tree translates into a constant number of push/pop operations on pop-tarts.
Each of the pop-tart operations is implemented using $O(1)$ amortized rotations within a subtree of amortized size $\OO(1)$ containing the root.
In the worst case, each pop-tart operation uses at most $\OO(\log n)$ rotations within a subtree of size $\OO(\log n)$ containing the root.
We simulate the changes done on the subtree using restricted rotations of Lemma~\ref{L:restricted-rotations}.
As each such subtree is of size $\OO(\log n)$, each intermediate tree has depth $\OO(\log n)$.

Since the amortized size of the modified subtree is $\OO(1)$ and the total number of operations is $\OO(\opt)$,
we obtain a sequence of trees $\mathcal{T}''$ of length $\OO(\opt)$ serving $x_1,\dots,x_n$,
where each tree has depth $\OO(\log n)$ and is obtained from the previous tree by a restricted rotation.

Finally, we want to make the trees in the sequence ribless.
We replace each rib (including spines) by the finger search tree of Brodal and Rysgaard with a finger at the top of the rib.
Simulating a restricted rotation changes the ribs of the tree being simulated.  
It is straightforward to check that a restricted rotation modifies at most four ribs, each having a top at depth $\OO(1)$, by adding one node to the top of at most two ribs and removing at most one node from the top of at most two ribs.  
These updates are pushes and pops on the ribs.  
Each such push or pop can be performed in the corresponding finger search tree in the subtree replacement model.
We extend the subtree to be replaced by adding the $\OO(1)$ nodes from the root of the finger search tree to the root of the entire tree.  
We replace the resulting subtree by doing the sequence of rotations given by \Cref{L:restricted-rotations}.  
Since a push or pop in a finger search tree takes $\OO(1)$ amortized rotations, each replaced subtree has $\OO(1)$ amortized size, so the simulation has an $\OO(1)$ amortized slowdown.
Furthermore, it maintains $\OO(1)$-riblessness and $\OO(\log n)$ depth.
This gives the final sequence of trees serving $x_1,\dots,x_n$ of length $\OO(\opt)$, where each tree has depth $\OO(\log n)$ and is ribless, and all rotations happen at depth $O(1)$.
\end{proof}

\subsection{Optimal trees for the split conjecture}\label{sec:split-opt-tree}

Let $x_1,x_2,\dots, x_n$ be any permutation of $[n]$ which defines a sequence of splits.
A static binary search tree $T^{OPT}$ achieving the minimum cost of splits for $x_1,\dots,x_n$ can be built recursively:
make $x_1$ the root of the optimal tree and recursively build its left subtree as an optimal splitting tree for the subsequence of $x_2,x_3,\dots,x_n$
of elements smaller than $x_1$, 
and its right subtree as an optimal splitting tree for the subsequence of elements larger than $x_1$. 
Such a splitting tree achieves the optimal cost $\OO(n)$.
After each split operation, we have a collection of binary search trees for disjoint intervals of $[n]$.
To prove the splitting conjecture for splay trees, we need to maintain a collection of trees that have logarithmic depth and are ribless, with all rotations happening at constant depth.
We need the total cost of the splits to be $\OO(n)$ starting from some tree with those properties.
The initial tree will be obtained from the static optimal tree $T^{OPT}$ for the splitting sequence $x_1,\dots,x_n$ by applying the transformation from the previous section.
That is, we represent the tree using the balanced tree of Bose et al., and we represent each rib (including spines) by the folded red-black tree of Brodal and Rysgaard.
To make changes to such a tree, we will use the technique of Cleary and Taback to perform all rotations at small depth.
Formally, we establish the following theorem.

\begin{theorem}\label{thm:normal-form-split}
For any permutation $x_1,\dots,x_n$ of $[n]$, there exists a sequence of forests $\mathcal{F}_0,\mathcal{F}_1,\dots,\mathcal{F}_\ell$ and a sequence of times $t_1 < \dots < t_n \in [\ell]$,
where $\ell=\OO(n)$ such that
\begin{enumerate}
    \item $\mathcal{F}_0$ consists of a single binary search tree on $[n]$,
    \item For all $i \in [\ell]$ and any $T \in \mathcal{F}_i$, $T$ is a binary search tree on a sub-interval of $[n]$, $T$ is ribless of depth $\OO(\log |T|)$,
    \item $\mathcal{F}_\ell$ consists of $n$ singleton binary search trees,
    \item For all $i \in [\ell]\setminus \{t_1,\dots,t_n\}$, $\mathcal{F}_{i-1} \setminus \mathcal{F}_{i} = \{T\}$, $\mathcal{F}_{i} \setminus \mathcal{F}_{i-1} = \{T'\}$
    where $T'$ is obtained from $T$ by a single rotation at depth at most $\OO(1)$,
    \item For all $j \in [n]$, $\mathcal{F}_{t_j-1} = \mathcal{F}_{t_j}$ and $x_j$ is a singleton tree in $\mathcal{F}_{t_j-1}$, or  $\mathcal{F}_{t_j-1} \setminus \mathcal{F}_{t_j}=\{T\}$, and $\mathcal{F}_{t_j} \setminus \mathcal{F}_{t_j-1} = \{T_L,T_x,T_R\}$
    where $T$ has root $x_j$, $T_L$ is the left subtree of $x_j$ in $T$, $T_R$ is the right subtree, and $T_x$ is a tree with the only node $x_j$. ($T_L$ or $T_R$ might not exist if the corresponding subtrees are empty.)
\end{enumerate}
\end{theorem}

As shown next, it is straightforward to extend the technique of Bose et al.~\cite{BoseCFLa12Deamortizing} to support splits at the root.
Next, we review their construction in more detail to explain how we maintain our collection of trees and how we analyze the total cost of splitting.

Bose et al. represent a binary search tree $T$ with a finger at node $f$ by a binary search tree $T'$ rooted at $f$.
If $m_\mathrm{L}$ is the largest item smaller than $f$ on the path from the root of $T$ to $f$,
and $m_\mathrm{R}$ is the smallest item larger than $f$ on that path, 
then $m_\mathrm{L}$ is the left child of $f$ in $T'$ and $m_\mathrm{R}$ is the right child.
The left subtree of $f$ in $T$ is represented by the right subtree of $m_\mathrm{L}$ in $T'$;
similarly, the right subtree of $f$ in $T$ is represented by the left subtree of $m_\mathrm{R}$.
(If $m_\mathrm{L}$ or $m_\mathrm{R}$ does not exist in $T$, then the representation of the relevant subtree is directly a child of $f$.)
The left subtree of $f$ in $T'$, rooted at $m_\mathrm{L}$, forms a pop-tart onto which nodes smaller than $f$ on the path from the root to $f$ in $T$ were pushed in increasing order,
so $m_\mathrm{L}$ is the top of the pop-tart.
(Each such node was pushed there also with its left subtree in $T$.)
Similarly, the right subtree of $f$ in $T'$ forms a pop-tart of nodes on the root-to-$f$ path in $T$ that are larger than $f$.

Each subtree of $T$ hanging from a node on the root-to-$f$ path in $T$ (including subtrees of $f$) is represented using a \emph{heavy path decomposition}.
Here, we say that an edge is heavy if it leads to a heavier child of a node (or to the left child in the case of a tie).
By the heavier child, we mean the child with a larger subtree.
A maximal path consisting of heavy edges leading to a node $u$ is represented in $T'$ by a tree rooted at $u$,
with the nodes leading to $u$ organized in two pop-tarts that are children of $u$ (reminiscent of the finger construction).
The left pop-tart contains nodes of the path that are smaller than $u$ pushed onto the pop-tart in decreasing order.
The other pop-tart contains nodes of the path that are larger than $u$ pushed in increasing order.
Hence, higher nodes are at the top of the pop-tarts (as opposed to the pop-tarts of nodes leading to $f$ in $T$).

Each subtree hanging from the heavy path ending at $u$ is represented recursively by the same structure, and all those representations are stored as subtrees within the appropriate pop-tarts.

Moving a finger down to a right child $v$ of $f$ in $T$ corresponds to popping $v$ from its pop-tart in the left child of $m_\mathrm{R}$ (together with the representation of one of its children in $T$), rotating it to the root, and pushing $f$, together with the representation of the left child of $v$, onto the pop-tart of $m_\mathrm{L}$.
The representation of the right child of $v$ stays as a left subtree of $m_\mathrm{R}$.
Moving a finger to a left child of $f$ in $T$ is symmetric.
This involves amortized $O(1)$ rotations in the pop-tarts, all of which are contained within an (amortized) $O(1)$-size subtree of $T$.

Similar operations on the pop-tarts allow us to move the finger up or simulate a rotation in $T$.
Details of those operations will not be necessary for us.

\medskip
For the split conjecture, we take the static tree that is optimal for the given splitting sequence $x_1,\dots,x_n$ and represent it as a Bose et al. tree with a finger at the root.
In subsequent steps, we will maintain a collection of disjoint subtrees of $T^{OPT}$, each in the Bose et al. representation with a finger at its root.

Consider one such subtree $T$ of $T^{OPT}$ with root $f$.
The corresponding Bose et al. tree $T'$ has $f$ as its root. The left subtree of $f$ in $T'$ represents the left subtree of $f$ in $T$, and the right subtree of $f$ in $T'$ represents the right subtree of $f$ in $T$, each using the heavy path decomposition.

To split $T'$ at $f$, we create two identical copies of $T'$, referred to as the left and right copies.
In the left copy of $T'$, we move the finger to the left child $u$ of $f$ in $T$.
This transforms the left copy of $T'$ into a tree rooted at $u$. 
The left child of $u$ is the representation of the left subtree of $u$ in $T$, and
the right child of $u$ is $f$.
The left child of $f$ is the representation of the right subtree of $u$ in $T$,
and its right child is the representation of the right child of $f$ in $T$.
$f$ with its right child forms a pop-tart with one inserted item.
We remove (and discard) $f$ with its right-child subtree from the left copy of $T'$.
Its left child will become the right child of $u$.
The left copy becomes a valid Bose et al. representation $T'_\mathrm{L}$ of the subtree of $u$ in $T$.

Symmetrically, we move the finger in the right copy of $T'$ to the right child $v$ of the root $f$
and remove $f$ with its left subtree from the resulting tree.
This gives a valid Bose et al. representation $T'_\mathrm{R}$ of the subtree of $v$ in $T$.
Notice that $T'_\mathrm{L}$ and $T'_\mathrm{R}$ are disjoint.

The transformation of $T'$ involves a constant number of pushes and pops on pop-tart structures of $T'$, together with a constant number of rotations of nodes near the root of $T'$.
The push and pop operations translate into up to $\OO(d)$ rotations in the affected subtrees (but $\OO(1)$ amortized rotations, as we will see), where $d$ is the depth of the trees.

Every pop-tart from $T'$ ends up as part of either $T'_\mathrm{L}$ or $T'_\mathrm{R}$,
and we invoke a constant number of push and pop operations on a constant number of them.
Some of the pop-tarts might disappear if they become empty, and new ones containing a single element might be created.
But there are $\OO(1)$ operations on pop-tarts invoked in total during a given split.

Each rib of each of our Bose et al. trees is represented using a folded red-black tree of Brodal and Rysgaard to make it ribless.
The number of operations on those red-black trees corresponds to the number of rotations within the Bose et al. tree representation.
The actual transformation of $T'$ into $T'_\mathrm{L}$ and $T'_\mathrm{R}$ is done using the technique of Cleary and Taback (\Cref{L:restricted-rotations}):
$T'$ with root $f$ is first transformed into a tree rooted at $f$ with left children $T'_\mathrm{L}$ and right children $T'_\mathrm{R}$ 
using rotations in constant depth and then the tree is split by removing $f$.

If we use $s$ rotations to create $T'_\mathrm{L}$ and $T'_\mathrm{R}$ from the left and right copies of $T'$, then the Cleary and Taback procedure uses $\OO(s)$ rotations.
During the procedure, the depth of $T'$ is maintained at $\OO(\log n)$ (and, in fact, at $\OO(\log |T'|)$), and the tree is ribless at all times.

\medskip
To count the overall cost of splitting our representation of $T^{OPT}$ according to $x_1,\dots,x_n$, we make the following observations.
As we already noted, each internal pop-tart undergoes pushes and pops and is eventually emptied completely, as the final trees are singleton vertices.
We can trace each pop-tart through its corresponding trees. 
If the pop-tart contains $\ell$ items initially,
then it undergoes $t\ge \ell$ pop and push operations.
Those operations trigger a total of $\OO(t)$ rotations to maintain the pop-tart.
(Here we are using the fact that performing $t$ operations on a pop-tart of size at most $t$ requires $\OO(t)$ rotations in total.)
In total, there are $\OO(n)$ push and pop operations in all the pop-tarts, as there are $\OO(1)$ such operations per split.
They trigger $\OO(n)$ total rotations in the Bose et al. representation of all the trees.

Similarly, each rib (including spines) of the Bose et al. representation is represented by a folded red-black tree of Brodal and Rysgaard.
During each rebuild of the Bose et al. representation, we rebuild an initial segment of the rib, which corresponds to a sequence of insertions and deletions in the corresponding folded red-black tree.
The total number of insertions and deletions on all the ribs during a split operation corresponds to the number of rotations performed by the Bose et al. representation.
Each rib of $T'$ ends up either in $T'_\mathrm{L}$ or $T'_\mathrm{R}$ during a split.
A rib of length $\ell$ undergoes $s \ge \ell $ insertion and deletion operations and eventually it completely disappears.
This translates into $\OO(s)$ rotations in total to maintain that rib.

As there are $\OO(n)$ rotations performed in the Bose et al. representation in total, we perform $\OO(n)$ insertions and deletions on the folded red-black trees
and hence $\OO(n)$ rotations on them in total.

Together with a constant number of other rotations done per split, the total number of rotations we perform is $\OO(n)$.
During a split with $t$ rotations, we rebuild a subtree of size $\OO(t)$ containing the root.
Hence, the Cleary and Taback procedure increases the total number of rotations by at most a constant factor.
Thus, the whole sequence of split operations incurs cost $\OO(n)$, which is asymptotically optimal.
Each intermediate tree has logarithmic depth and is ribless.
All rotations happen at depth $\OO(1)$.

\section{The split conjecture}
\label{app:split-conjecture}

Here we demonstrate that our techniques can also be applied to the split conjecture to prove that it holds up to a small multiplicative factor.

\paragraph{The split operation.}
Fix a BST $T$ on some key set $S\subseteq [n]$, and let $x\in S$.
A \emph{split at $x$} first rebuilds the tree so that $x$ becomes the root and then deletes $x$.
After the deletion, the left and right subtrees of $x$ become two independent
BSTs on the key sets $\{y\in S:y<x\}$ and $\{y\in S:y>x\}.$
If we continue performing splits, each subsequent key is split inside the unique
current tree that still contains it.

For splay trees, we restructure the tree simply by splaying $x$.

\medskip

The original motivation for the split conjecture stems from the sequential access theorem~\cite{Tarjan85SequentialAccess} and from the study of specific cases in which the optimal BST has constant amortized cost.
Indeed, one can instead view the sequential access of the elements $1$ through $n$ as a sequence of splits.
It is then natural to ask whether the same bound holds for an arbitrary order of splits.
As previously observed by Lucas~\cite{Lucas91Split}, given an arbitrary tree, one can use a linear number of rotations to restructure it so that the root is the vertex whose key $x$ is split first, and its children are, respectively, the first keys in the splitting sequence that are smaller and larger than $x$.
The cost of each split in the restructured tree is constant, as whenever a key is split, it is the root of its tree; thus, the optimum tree can achieve $O(n)$ cost.
By naturally extending the dynamic optimality conjecture to include splits (and potentially joins), one might expect splay trees to achieve linear cost as well.

\paragraph{Split conjecture~\cite{Lucas88Canonical,Pettie08Deque}.}
Given an arbitrary initial BST $T$ on $[n]$ and letting
$x_1,\dots,x_n$ be any permutation of the keys, the cost of the operations
$\mathrm{split}(x_1),\ \mathrm{split}(x_2),\ \dots,\ \mathrm{split}(x_n)$, starting from the tree $T$, is $O(n)$ for splay trees.

We now prove that the total cost of splitting is $O(n \log\log n \cdot \log^2\log\log n)$, i.e., that the split conjecture holds up to a factor of $\log\log n \cdot \log^2\log\log n$.

\paragraph{Handling root deletes.} We imagine a split proceeding as follows: the optimal tree first accesses the vertex to be deleted, the splay tree then accesses it, and finally both trees delete the root. Thus, the only new operation, compared with those needed to handle dynamic optimality, is the deletion of the root, as the access is the same as if no deletion followed.
After the splay, we remove the deleted value vertex, its corresponding parent internal vertex, and the associated lazy intervals. This is accounted for by applying \textsc{Delete} to the two heap-children in their lazy intervals. Since the rank of the deleted vertex was $0$, the structures of the heavy paths and lazy intervals in the individual subtrees remain unchanged. To compensate for the deletion of the root in the optimal tree, all ranks decrease by one. This, however, does not affect any of the gaps, so no additional work is needed. In total, accounting for the deletion of the root requires only free operations on the interface.

\paragraph{Assumptions on the optimal tree for splitting.} As with the dynamic optimality conjecture, we need to assume that the optimal forest for splits is always ribless, has logarithmic depth, and performs only $O(n)$ additional rotations at depth $O(1)$. This is proved by \Cref{thm:normal-form-split}.

\begin{theorem}\label{thm:split-conjecture}
    Let $T$ be a splay tree over $n$ nodes in any configuration.
    The total cost of the splitting process of $T$ is at most $O(n \log \log n \log ^2 \log\log n)$.
\end{theorem}

\begin{proof}
	    Deleting a splayed vertex requires only free operations. Since the normal form still holds, \Cref{thm:main} implies that the splay tree performs only $O((n + \cost(T^\text{OPT})) \log\log n \log ^2 \log\log n)$ rotations.
    Since $\cost(T^\text{OPT}) \in O(n)$, the splay tree does $O(n)$ root deletes and $O(n \log\log n \log^2\log\log n)$ rotations.
\end{proof}

\end{document}